\documentclass[11pt]{article} 
\usepackage{amsmath}
\usepackage{graphicx}
\usepackage{enumerate}
\usepackage{natbib}
\usepackage{url} 

\newcommand{\blind}{1}

\usepackage[toc,page]{appendix}
\usepackage{float}

\usepackage{hyperref}
\hypersetup{
    colorlinks=true,
    linkcolor=blue,
    filecolor=magenta,      
    urlcolor=cyan,
    citecolor=blue
}
 
\usepackage{amssymb,amsbsy,amsfonts,amsmath,xspace,amsthm}
\usepackage{mathrsfs}
\usepackage{graphicx}

\theoremstyle{theorem}
\newtheorem{theorem}{Theorem}
\newtheorem{definition}{Definition}

\newtheorem{proposition}{Proposition}

\newtheorem{lemma}{Lemma}

\newtheorem{example}{Example}
\newtheorem{assumption}{Assumption}

\usepackage{texmacro_maths}
\usepackage[subrefformat = parens]{subcaption}

\newcommand{\edgedensity}{ \rho }
\newcommand{\similarity}{ \mathrm{sim} }
\newcommand{\lca}{ \mathrm{lca} }

\newcommand{\CH}{ \mathrm{CH} }

\newcommand{\hcC}{ \widehat{\mathcal C} }
\newcommand{\hcT}{ \widehat{\mathcal T} }
\newcommand{\tcC}{ \widetilde{\mathcal C} }

\newcommand{\symdiff}{ \, \Delta \, }

\newcommand{\ace}{ \mathrm{loss} }
\newcommand{\mc}{ \mathrm{sc} }

\newcommand{\Multi}{ \mathrm{Multi} }

\newcommand{\rebuttal}[1]{\textcolor{blue}{#1}}
\newenvironment{newtext}{\par \color{blue} }{\par}

\newcommand{\ba}{ \bar{a} }
\newcommand{\bp}{ \bar{p} }
\newcommand{\hzeta}{ \hat{\zeta} }

\newcommand{\tda}{ \mathrm{td} }
\newcommand{\bua}{ \mathrm{bu} }

\renewcommand{\hC}{ \widehat{C} }

\renewcommand{\hK}{ \widehat{K} }

\newcommand{\adversarial}{ \mathrm{adversarial} }

\newcommand{\uniform}{ \mathrm{uniform} }

\usepackage[bb=dsserif]{mathalpha}
\newcommand{\1}{ \mathbb{1} }

\begin{document}

\def\spacingset#1{\renewcommand{\baselinestretch}%
{#1}\small\normalsize} \spacingset{1}


\if1\blind
{
  \title{\bf When does bottom-up beat top-down in hierarchical community detection? }
  \author{Maximilien Dreveton, Daichi Kuroda, Matthias Grossglauser, Patrick Thiran \hspace{.2cm} \\
  \url{{maximilien.dreveton, daichi.kuroda, matthias.grossglauser, patrick.thiran}@epfl.ch} \\ 
    École Polytechnique Fédérale de Lausanne (EPFL)}
    \date{September 15, 2025}
  \maketitle
} \fi

\if0\blind
{
  \bigskip
  \bigskip
  \bigskip
  \begin{center}
    {\LARGE\bf When does bottom-up beat top-down in hierarchical community detection? }
\end{center}
  \medskip
} \fi

\bigskip
\begin{abstract}
 Hierarchical community detection consists in finding a tree of communities where deeper levels of the hierarchy reveal finer-grained structures. There are two main classes of algorithms for this task. Divisive (\textit{top-down}) algorithms recursively partition nodes into smaller communities until a stopping criterion indicates that no further splits are necessary. In contrast, agglomerative (\textit{bottom-up}) algorithms first identify the smallest community structures and then repeatedly merge the communities by using a \textit{linkage} method. In this work, we prove that a bottom-up algorithm recovers the hierarchy of a hierarchical stochastic block model (HSBM) when the average degree grows unbounded. We also establish the information-theoretic threshold for exact recovery at intermediate depths of the hierarchy and highlight its significance in understanding the limitations of top-down algorithms. Numerical experiments on both synthetic and real datasets demonstrate the superiority of bottom-up methods. In particular, a notable drawback of top-down algorithms is their tendency to produce dendrograms with inversions. These findings contribute to a better understanding of hierarchical clustering techniques and their applications in network analysis.

\end{abstract}

\noindent%

\spacingset{1.3} 


\tableofcontents

\section{Introduction}

A system of pairwise interactions among entities can conveniently be represented by a graph, where the system entities are the nodes and the interactions are the edges. Data collected in such form is increasingly abundant in many disciplines, such as sociology, physics, economics, and biology~\cite{Newman2018_networksbook}. Finding community structures by grouping nodes with similar connection patterns into clusters is one of the most important statistical analysis tasks on networks~\cite{fortunato2010community,avrachenkov2022statistical}. Community structures are often hierarchical. For example, in a co-authorship network, we can partition the researchers, based on their primary discipline (such as mathematics, physics, computer science, etc); each of these fields can be further split into specific sub-disciplines. The sub-division of larger communities into smaller ones provides a finer division of the network. 

Hierarchical communities can naturally be inferred using \new{top-down} approaches, where the process begins by identifying the largest communities at the top of the hierarchy. These communities are then recursively decomposed into smaller sub-communities until a stopping rule indicates no further division is necessary. To identify the communities at the top of the hierarchy, one strategy is to progressively remove edges with the highest edge-betweenness centrality~\cite{girvan2002community} or with the lowest edge-clustering coefficient~\cite{radicchi2004defining}. Another strategy is to bi-partition the network by using spectral clustering ~\cite{dasgupta2006spectral,balakrishnan2011noise,li2022hierarchical}. Top-down algorithms, however, possess several limitations. First, any clustering errors initially made become locked in and propagate to the next rounds of subdivisions, potentially compromising the accuracy of the predicted hierarchy. Second, the recursive splittings can overlook valuable information by ignoring the edges between communities that have already been separated. As a result, these approaches might not fully capture the interconnections between communities across different depths of the hierarchy. 

In contrast, \new{agglomerative} (or \new{bottom-up}) algorithms take a different approach by constructing the hierarchy from the bottom upwards. These algorithms recursively merge smaller communities to form larger ones. Some bottom-up algorithms, such as the one proposed in~\cite{bonald2018hierarchical,pons2005computing,chang2011general}, generate a complete \new{dendrogram}. A dendrogram is a tree whose leaves are individual nodes, whose branches and internal nodes represent merged clusters. The length of every branch measures the similarity of its children. 

Comparing different hierarchical community detection (HCD) algorithms can be done by establishing theoretical guarantees of their performance on random graph models. The \new{hierarchical stochastic block model} (HSBM) is a general model of random graphs containing hierarchical communities. This model defines the hierarchical community structure as a rooted binary tree whose leaves correspond to \new{primitive communities}. Each node belongs to a primitive community, and the interactions between two nodes belonging to communities $a$ and $b$ depend only on the lowest common ancestor between $a$ and $b$ on the tree. 

We first establish that recovering the hierarchy of an HSBM is possible using a bottom-up algorithm when the average degree of the graph is only $\omega(1)$. Earlier studies required stronger degree growth conditions, such as $\Theta(N)$ in~\cite{balakrishnan2011noise}, $ \omega( N^{1/2} \log^{p} N )$ (with $p = 1/2$ in~\cite{cohen2017hierarchical} and $p = 2$ in~\cite{lyzinski2016community}), or $\omega( \log^{1+p} N )$ (with $p = 5$ for~\cite{dasgupta2006spectral} and $p > 1$ fixed for~\cite{li2022hierarchical}). We further discuss these results in Section~\ref{sec:discussion_previous_work}.

In the context of hierarchical communities, it is possible to study community recovery at different depths of the hierarchy. We rigorously establish the information-theoretic threshold for the \textit{exact recovery} of the communities at any intermediate depth of the hierarchy. Exact recovery is the strongest notion of recovery: it is the ability to fully recover the communities when the number of nodes~$N$ tends to~infinity. In particular, we show that exact recovery at a larger depth is more challenging than exact recovery at a smaller depth. We further discuss the implication of this information-theoretic threshold to the success of top-down algorithms in Section~\ref{sec:discussion_top-down_intermediate-depth}.

 In our numerical experiments, we employ a bottom-up algorithm that follows a two-step process. To infer the underlying primitive communities, we first apply a spectral algorithm on the graph's Bethe-Hessian matrix~\cite{saade2014spectral,dall2021unified}. Next, we use a bottom-up approach to build the hierarchy. To provide a comprehensive comparison, we evaluate this bottom-up algorithm against recursive spectral bi-partitioning~\cite{lei2020unified,lei2021consistency}, the most relevant top-down algorithm whose performance has been theoretically established and numerically validated. Our findings on synthetic data sets demonstrate that the bottom-up algorithm achieves exact recovery at intermediate depths, up to the information-theoretic thresholds, whereas recursive spectral bi-partitioning fails to do so. Furthermore, we show that the dendrogram produced by recursive spectral bi-partitioning suffers from \new{inversions}. Inversions occur when the algorithm incorrectly places a lower-depth cluster above a higher-depth cluster in the dendrogram, thus distorting the true hierarchical structure. Such inversions lead to misleading interpretations of the hierarchical relationships within the~data.

The paper is structured as follows. In Section~\ref{sec:HCD_presentation}, we describe top-down and bottom-up approaches for hierarchical community detection. The HSBM is defined in Section~\ref{sec:tree_recovery}, where we also derive conditions for recovering the hierarchy by the {linkage} procedure. Next, we investigate the exact recovery at intermediate depths of the hierarchy in Section~\ref{sec:intermediate_exact_recovery}. We discuss these results in light of the existing literature in Section~\ref{sec:discussion}. Finally, Section~\ref{section:numerical_results} is devoted to the numerical experiments.

\paragraph{Notations}

We denote by $[N]$ the set $\{1, \cdots, N \}$, by $\Ber(p)$ a Bernoulli random variable with parameter $p$. 
The Frobenius norm of a matrix is denoted $\| \cdot \|_F$. 
The \new{\Renyi divergence} of order $t \in (0,1)$ of a Bernoulli distribution $\Ber(p)$ from another distribution $\Ber(q)$ is defined as $\dren_t\left(\Ber(p) \| \Ber(q) \right) = \frac{1}{t-1} \log \left( (1-p)^t (1-q)^{1-t} + p^t q^{1-t} \right) $. 
For $t = 1/2$, we write $ \dren_{1/2}\left(\Ber(p), \Ber(q) \right) = -2\log \left( \sqrt{(1-p)(1-q)} + \sqrt{pq} \right)$. 

We focus on undirected graphs $G = (V,E)$ whose node set is $V = [N]$ and adjacency matrix $A = \left( A_{ij} \right) \in \{0,1\}^{N \times N}$. For a subset~$V_1$ of the node set $V$, we let $G[V_1]$ be the subgraph of $G$ induced by~$V_1$. 
We denote by $\cN_{\cT}$ (resp., $\cL_{\cT}$) the internal nodes (resp., the leaves) of a tree $\cT$. For an internal node $u$ of a rooted $\cT$, denote by $\cT[u]$ the sub-tree rooted at $u$. We denote~$\lca(u,v)$ the lowest common ancestor between two nodes $u,v \in \cT$.

\section{Hierarchical Community Detection}
\label{sec:HCD_presentation}

Many networks present a hierarchical community structure. The primitive communities at the bottom of the hierarchy are a collection of subsets $\cC = (C_1, \cdots, C_K)$ that partition the original node set $V$ into $K$ disjoint sets. These primitive communities are the leaves of a rooted tree $\cT$, which defines the hierarchical relationship between the communities. 
This section reviews two main strategies for \new{hierarchical community detection}~(HCD).

\subsection{Divisive (\textit{top-down}) Algorithms}

\new{Divisive (top-down)} HCD algorithms begin with one single community containing all the nodes. This community is recursively split until a selection rule indicates that no further splits are needed. This can be summarized as follows:
\begin{enumerate}
    \item apply a selection rule to decide if the community contains sub-communities. If no, stop; if yes, split into two communities;
    \item recursively repeat step 1 on each of the two sub-communities found. 
\end{enumerate}
Different choices for the stopping rule or the bi-partitioning algorithms have been explored~\cite{dasgupta2006spectral,balakrishnan2011noise,li2022hierarchical}. 

Each recursive splitting of the graph loses some information. For example, consider the two clusters $C_0$ and $C_1$ obtained after the first split. The next step splits $C_0$ (resp., $C_1$) into two subclusters $C_{00}$ and $C_{01}$ (resp., $C_{10}$ and $C_{11}$), based only on the induced sub-graph $G[C_0]$ (resp., $G[C_1]$). As a result, the clustering of $C_0$ does not take into account the edges from~$C_0$ to~$C_1$. If the edge densities between~$C_{00}$ and~$C_{1}$ are different from the edge densities between~$C_{01}$ and~$C_1$, then this valuable information is not used by a top-down algorithm. 

Furthermore, the resulting tree is unweighted. Although it is possible to compute similarities between pairs of predicted clusters (for example using the edge density between two clusters), these similarities are not guaranteed to increase with the depth of the dendrogram; inversions can occur. These inversions, as we will emphasize in the numerical section, lead to discrepancies in the hierarchical structure conveyed by the dendrogram.

\subsection{Agglomerative (\textit{bottom-up}) Algorithms}

\new{Agglomerative (bottom-up)} HCD algorithms construct a sequence of clusterings in an ascending manner, where the dendrogram is created from its leaves to its root. The initial clustering corresponds to the leaves of the dendrogram, and the hierarchy is progressively built up by iteratively merging the most similar clusters.

A first approach initially assigns each node to a separate community and then merge the clusters that minimize a distance metric~\cite{newman2004fast,pons2005computing,chang2011general,bonald2018hierarchical}. Although these methods generate a complete dendrogram, determining the depth in the hierarchy where the community structure becomes meaningful is an old problem~\cite{mojena1977hierarchical,gao2022selective}.

A second approach directly estimates the bottom clusters by using a flat (non-hierarchical) graph-clustering algorithm. These bottom clusters are then sequentially merged two by two, based on a similarity measure. 
This can be summarized as follows:
\begin{enumerate}
 \item apply a graph-clustering algorithm to find the bottom clusters;
 \item 
 \begin{enumerate}
   \item compute the similarity between all pairs of bottom clusters;
   \item merge the two clusters that are the most similar,\footnote{When measuring the similarity between pairs of clusters using edge density, an implicit assumption is that the clustering structure exhibits assortativity--meaning that nodes within the same community are more likely to be connected. However, if the community structure is disassortative, clusters with the lowest similarity are merged.} and update the similarities between this new cluster and the existing ones;
   \item repeat step (2b) until all clusters have been merged into a single one.
 \end{enumerate}
\end{enumerate}
The abundant literature on graph clustering provides numerous candidate algorithms for the first stage. We note that, in most practical applications, the true number of bottom clusters~$K$ is often unknown and needs to be inferred at this stage as well.  
The second stage is commonly called \new{linkage}. Various linkage variants exist, depending on the chosen similarity measures and update rules.  In this paper, we adopt the linkage variant used in~\cite{cohen2017hierarchical,cohen2019hierarchical}, employing edge density as the similarity measure between pairs of clusters. The edge density between two node sets $V_1, V_2 \subset V$ is defined by 
\begin{align}
\label{eq:def_edgedensity}
 \edgedensity \left( V_1, V_2 \right) \weq \frac{ w\left( V_1, V_2 \right) }{ |V_1| \cdot |V_2| }
 \quad \text{ where } \quad
 w\left( V_1, V_2 \right) \weq \sum_{i \in V_1} \sum_{j \in  V_2} A_{ij}.
\end{align}
Recomputing the similarities in Step 2b is done as follows. Let $\hC_1, \cdots, \hC_{\hK}$ be the clusters initially predicted by the flat graph-clustering algorithm. 
Suppose that at some step of the algorithm, clusters $\hC_{k_1}$ and $\hC_{k_2}$ are the most similar and hence are merged to give a new cluster $\hC_{k_1 \cup k_2 } = \hC_{k_1} \cup \hC_{k_2}$. We have for $\ell \not \in \{k_1,k_2\}$, 
\begin{align}
\label{eq:extension_edgedensity_set_of_clusters}
 \edgedensity \left( \hC_{ k_1 \cup k_2 }, \hC_{ \ell } \right) \weq 
 \frac{ \left| \hC_{k_1} \right| }{ \left| \hC_{k_1 \cup k_2 } \right| } \edgedensity \left( \hC_{k_1}, \hC_{\ell} \right) 
 +
 \frac{ \left| \hC_{k_2} \right| }{ \left| \hC_{k_1 \cup k_2 } \right| } \edgedensity \left( \hC_{k_2}, \hC_{\ell} \right).
\end{align}
Therefore, the edge density naturally defines an \new{average-linkage} procedure for merging the clusters identified through the initial flat clustering. In particular, this procedure guarantees that the resulting dendrogram will be free from inversions, thus ensuring a coherent hierarchical structure~\cite{murtagh2017algorithms}. Algorithm~\ref{algo:average_linkage} provides a concise summary of this~process. 

\begin{algorithm}[!ht]
\KwInput{Graph $G = (V,E)$, partition $\hcC$ of $V$. 
}
 For all $k,\ell \in \left[ \hK \right]$, compute $\edgedensity\left( \hC_k, \hC_\ell \right)$ as in~\eqref{eq:def_edgedensity}; \\
  \textbf{while} $\left| \hcC \right| \ge 2$ \textbf{do:}
 \begin{itemize}
   \item Let $\hcC \leftarrow \hcC \, \backslash \, \left\{ \hC_k, \hC_{\ell} \right\}$ where $\hC_k, \hC_\ell \in \argmax\limits_{C \ne C' \in \hcC } \edgedensity \left( C, C' \right)$;
   \item For any $C \in \hcC$ compute $\edgedensity \left( \hC_{k  \cup \ell}, C \right)$ as defined in~\eqref{eq:extension_edgedensity_set_of_clusters};
   \item Let $\hcC \leftarrow \hcC \cup \left\{ \hC_{k\cup \ell} \right\}$.
 \end{itemize}
\KwReturn{The tree $\hcT$ capturing the sequence of nested merges in the while-loop.}
\caption{Average-linkage.}
\label{algo:average_linkage}
\end{algorithm}

\section{Tree Recovery from the Bottom}
\label{sec:tree_recovery}

We study the asymptotic performance of Algorithm~\ref{algo:average_linkage} on a class of random graphs with hierarchical community structures. We define the model in Section~\ref{subsection:HSBM}. We prove that Algorithm~\ref{algo:average_linkage} recovers the hierarchical tree when the average degree of the graph grows unbounded in Section~\ref{section:tree_recovery_bottom}. Finally, we study the bounded degree regime in Section~\ref{section:tree_recovery_bounded_degree}.

\subsection{Hierarchical Stochastic Block Model}
\label{subsection:HSBM}

The hierarchical stochastic block model (HSBM) is a class of random graphs whose nodes are partitioned into latent hierarchical communities. Before defining this model formally, let us introduce some notations. Each node $u$ of a rooted binary tree $\cT$ is represented by a binary string as follows. The root is indexed by the empty string $\emptyset$. Each non-root node $u$ of the tree is labeled by the unique binary string $u = u_1u_2\cdots u_q\cdots$ that records the path from the root (\textit{i.e.,} $u_q = 1$ if step $q$ of the path is along the right branch of the split and $u_q = 0$ otherwise). The depth of node~$u$ is denoted by $|u|$ and coincides with its distance from the root. Finally, with this parametrization, the \new{lowest common ancestor} $\lca(u,v)$ of two nodes $u,v \in \cT$ is the \textit{longest common prefix} of the binary strings $u$ and $v$. This is also the common ancestor of $u$ and $v$ with the largest depth.

Here, we denote by $\cL_\cT$ the leaves of the tree~$\cT$. We will assign each node of the graph $G$ to one leaf of~$\cT$ and denote by $C_a$ the set of nodes assigned to the leaf $a$. This forms the \new{primitive communities}
$\cC = \left( C_{a} \right)_{a \in \cL_\cT}$. 
Any internal node~$u$ of the tree is associated with a \new{super-community} $C_u$ such that $C_u = \cup_{a \in \cL_{\cT[u]}} C_{a}$ where $\cL_{\cT[u]}$ denotes the leaves of the sub-tree of~$\cT$ rooted at $u$. In particular, we have $C_{\emptyset} = V$ and $C_v \subset C_{u}$ if $v$ is a child of $u$. 

We suppose that the probability $p_{ab}$ of having a link between two nodes belonging to the primitive communities $C_{a}$ and $C_{b}$ depends only on the lowest common ancestor of $a$ and~$b$. Hence, we denote the probability $p_{ab}$ by $p(\lca(a,b))$. In an assortative setting, we have $p(u) < p(v)$ if $v$ is a child of $u$. 

\begin{definition}
Let $N$ be a positive integer, $\cT$ a rooted binary tree with $|\cL_{\cT}| = K$ leaves, and $\pi = (\pi_a)_{a \in \cL_{\cT}}$ a probability vector. Let $p \colon \cT \to [0,1]$ be an increasing function, meaning that for any $u, v \in \cT$, if $v$ is a child of $u$ we have $p(u) < p(v)$. A \new{hierarchical stochastic block model} (HSBM) is a graph $G = (V,E)$ such that $ V = [N]$ and
\begin{enumerate}
 \item each node $i \in [N]$ is independently assigned to a community $C_{a}$ where $a$ is sampled from $[K]$ according to $\pi$; 
 \item two nodes $i \in C_{a}$ and $j \in C_{b}$ are connected with probability $p( \lca( a ,b ) )$. 
\end{enumerate}
\end{definition}

An important particular case of HSBM is the \new{binary tree SBM} (BTSBM), in which the tree~$\cT$ is full and balanced, and the probability of a link between two nodes in clusters $C_a$ and~$C_b$ depends only on the depth of $\lca(a,b)$, \textit{i.e.,} $p(\lca(a,b)) = p_{|\lca(a,b)|}$ for all $a,b \in \cL_{\cT}$. 
In particular, the number of communities of a BTSBM is $K = 2^{d_{\cT}}$, and assortativity implies that $p_0 < p_1 < \cdots < p_{d_\cT}$. We illustrate an HSBM and a BTSBM in Figure~\ref{fig:illustration_HBM_BTSBM}. 
Observe that the BTSBM defines the same model as in~\cite{li2022hierarchical}, while the HSBM is a natural extension where the tree is still binary but not necessarily full and complete.

\begin{figure}[!ht]
 \centering
 \begin{subfigure}[b]{0.4\textwidth}
  \includegraphics[width=\textwidth]{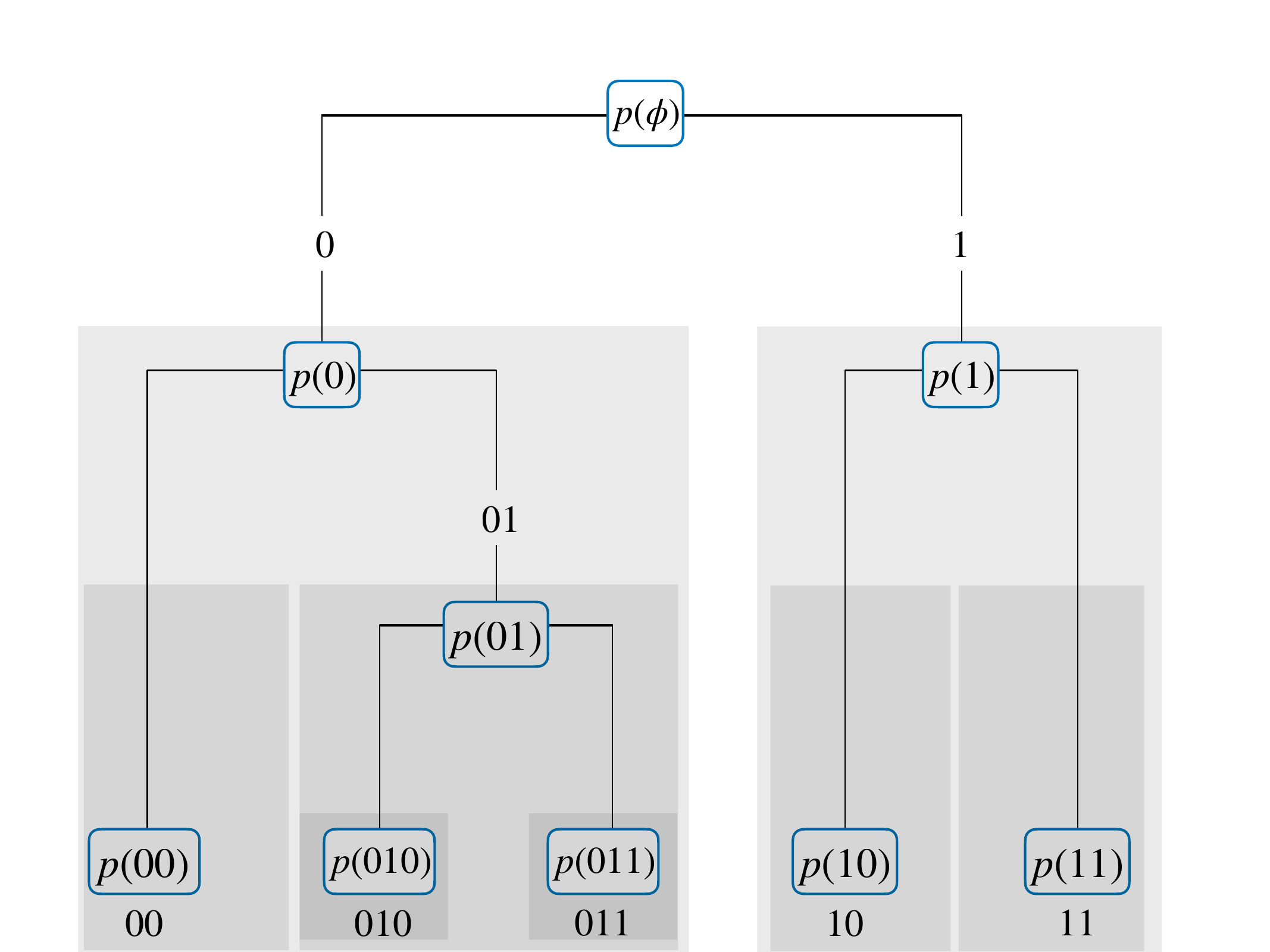}
  \caption{Hierarchical stochastic block model}
  \label{fig:illustration_HBM}
 \end{subfigure}
 \hfil
 \begin{subfigure}[b]{0.4\textwidth}
  \includegraphics[width=\textwidth]{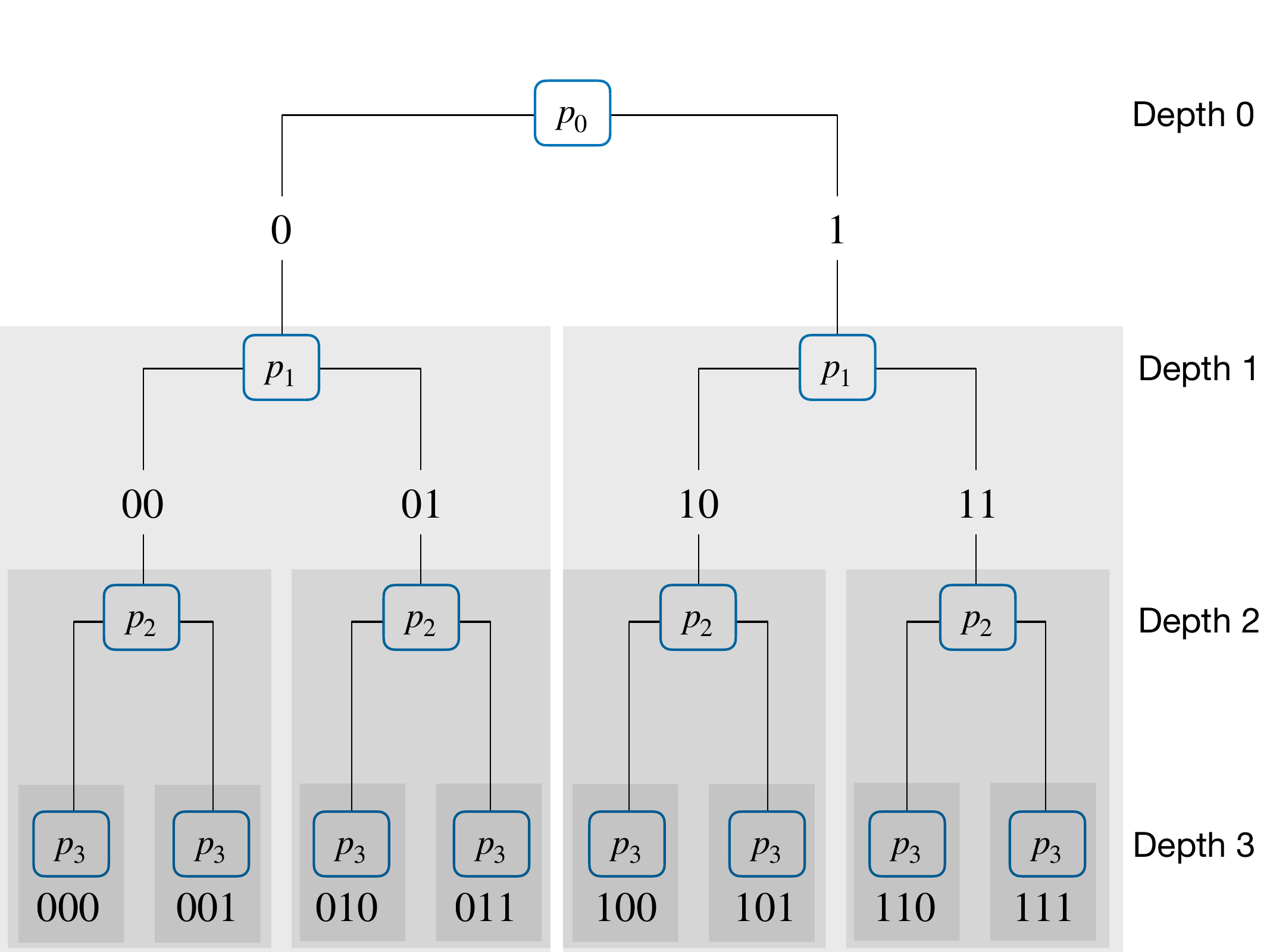}
  \caption{Binary tree SBM of depth 3}
  \label{fig:illustration_BTSBM}
 \end{subfigure}
 \caption{Examples of (a) an HSBM and (b) a BTSBM, with the binary string representation of each node. The link probabilities are $p(u)$ for the HSBM and $p_{|u|}$ for the BTSBM. The grey-colored rectangles represent the super-communities.} 
 \label{fig:illustration_HBM_BTSBM}
\end{figure}

In our results, we use the following assumptions.
\begin{assumption}[Fixed hierarchy]\label{assumption:fixed_quantites}$\cT$ is a rooted binary tree with $K$ leaves and is independent of $N$. The probability vector $\pi \in (0,1)^K$ is also independent of $N$ and satisfies $\min_{a \in K} \pi_a > 0$. 
\end{assumption}
\begin{assumption}[Asymptotic scaling]\label{assumption:scalings} The edge connection probabilities $p$ can be written as $p(t) = a(t) \delta_N$, where $(a(t))_{t \in \cT}$ are positive constants (independent of $N$).
\end{assumption}
Under Assumption~\ref{assumption:scalings}, the average degree is $\Theta(N \delta_N)$, and we call $\delta_N$ the sparsity factor.

\subsection{Tree Recovery with Growing Average Degree}
\label{section:tree_recovery_bottom}

In this section, we study the recovery of the hierarchical tree $\cT$ and the bottom communities $\cC = \left( C_{\ell_1}, \cdots, C_{\ell_K} \right)$ associated with the $K$ leaves of $\cT$. For an estimator $\hcC = \left( \hC_{1}, \cdots, \hC_K \right)$ of~$\cC$ verifying $|\cC| = |\hcC| = K$, the \new{number of mis-clustered nodes} is
\begin{align}
\label{eq:def_ace}
 \ace \left( \cC, \hcC \right) \weq \min_{\tau \in \cS_{[K]} } \sum_{ k=1 }^K \left| C_{\ell_k} \symdiff \hC_{\tau(k)} \right|, 
\end{align}
where $B \symdiff C = \left( B \cup C \right) \backslash \left( B \cap C \right)$ denotes the symmetric difference between two sets $B, C \subset V$. The minimum is taken over the symmetric group $\cS_{[K]}$ of all permutations of~$[K]$, because we recover the bottom communities only up to a global permutation of the community labels. We study sequences of networks indexed by the number of nodes $N$ and for which the interaction probabilities might depend on $N$. 
 An estimator $\hcC$ of $\cC$ achieves  \new{exact recovery} if $ \ace \left( \cC, \hcC \right) \asto 0 $, \new{almost exact recovery} if $N^{-1} \ace \left( \cC,\hcC \right) \pto 0$ and \new{weak recovery} if $\pr ( N^{-1} \ace \left( \cC,\hcC \right) < \|\pi\|^2_2 - \varepsilon)\to 0$ for all $\varepsilon > 0$.\footnote{The interpretation is as follows. An exact estimator makes no mistakes, while an almost exact estimator makes $o(N)$ mistakes. A weak estimator makes $O(N)$ mistakes but still outperforms a naive random guessing, which ignores the graph structure and classifies a node $i$ into community $a$ with probability $\pi_a$, independently of other nodes.}


The number of edges between two communities $C_a$ and $C_b$ is binomially distributed with mean $|C_a| \cdot |C_b| \cdot p_{ab}$. Therefore, if $ p_{ab} = \omega( N^{-2})$, this binomial random variable is concentrated around its mean and we have $\edgedensity( C_a, C_b ) = (1+o(1)) p_{ab}$. 
This suggests that, once the primitive communities are almost exactly recovered, the average-linkage procedure successfully recovers the tree from its leaves if the edge probability between different communities is $\omega \left( N^{-2} \right)$. The proof is more involved because (i) we compute $\edgedensity( \hC_a,\hC_b)$ and not $\edgedensity(C_a,C_b)$, and (ii) the estimator $\hcC$ is correlated with the graph structure. In particular, to establish that $\edgedensity( \hC_a,\hC_b) = (1+o(1)) p_{ab}$, we impose $p_{ab} = \omega(N^{-1})$.  

\begin{theorem}
\label{thm:performance_bottomUp}
 Consider an assortative HSBM. 
 \rebuttal{Suppose that Assumptions~\ref{assumption:fixed_quantites} and~\ref{assumption:scalings} hold,} with $N \delta_N = \omega( 1 )$. 
 Let $\hcC$ be an estimator of $\cC$, possibly correlated with the graph edges and such that $|\hcC| = |\cC|$. If $\hcC$ is almost exact, then Algorithm~\ref{algo:average_linkage} recovers~$\cT$ (starting from~$\hcC$). 
\end{theorem}

We prove Theorem~\ref{thm:performance_bottomUp} in Appendix~\ref{section:proof_thm:performance_bottomUp}. The assumption of the existence of an almost exact estimator $\hcC$ of $\cC$ is not limiting here. Indeed, under the assumptions of Theorem~\ref{thm:performance_bottomUp}, we can obtain such an estimator by the spectral algorithm of~\cite{yun2016optimal} or by the agnostic-degree-profiling algorithm of~\cite{abbe2015recovering}. 
Because the average degree is of the order $\Theta(N \delta_N)$, Theorem~\ref{thm:performance_bottomUp} ensures that Algorithm~\ref{algo:average_linkage} recovers the hierarchy when the average degree grows unbounded.\footnote{When $N \rho_N = \omega(1)$, the graph may be disconnected; however, only $o(N)$ nodes are excluded from the largest connected component, making almost exact recovery possible (but not exact recovery).} 

\subsection{Tree Recovery with Bounded Average Degree}
\label{section:tree_recovery_bounded_degree}

For simplicity, throughout this section, we assume that $\cT$ is a full and balanced tree, and the edge probabilities satisfy $p(t) = a(t) \delta_N$ where $(a(t))_{t \in \cT}$ are positive constants (independent of $N$),  with $N \delta_N = \Theta(1)$. This bounded average degree regime is more challenging, as no algorithm can achieve almost exact recovery.\footnote{
When $\delta_N = \Theta(1/N)$ and the average degree exceeds 1, there exists a single giant component containing $\Theta(N)$ vertices, while all other components are of size $o(N)$. Importantly, this giant component includes vertices from all $K$ communities; thus, there is no scenario in which two communities are entirely disconnected. Therefore, even in this extremely sparse regime, it remains possible to infer the community structure and estimate inter-community edge densities. We have clarified this point in the revised text to aid reader understanding.
}
However, when $N \delta_N \ge C$ where $C$ is a quantity depending only on $a(t)$ but independent of $N$, some algorithms achieve weak recovery~\cite{abbe2017community}. This leads to the following question: \textit{can average-linkage recover the tree, but starting with a weak estimator $\hcC$ of $\cC$ instead of an almost exact estimator?}

To answer this question, we must handle the fact that a weak estimator misclassifies $\Theta(N)$ nodes, as opposed to $o(N)$ for an almost exact estimator. Recall from the discussion before Theorem~\ref{thm:performance_bottomUp} that the correlation between $\hcC$ and the graph $G$ makes it challenging to study the concentration of $\edgedensity( \hC_a,\hC_b)$. As a result, we can no longer guarantee that $\edgedensity( \hC_a,\hC_b) = (1+o(1)) p_{ab}$ when $\hcC$ is only a weak estimator. To mitigate this issue, we split $G$ into two graphs $G^1$ and $G^2$ and proceed in two steps. We first apply the flat-clustering algorithm on $G^1$ to obtain~$\hcC^1$, and next estimate the edge density by 
 \begin{align}
 \label{eq:def_edge_density_graph_splitting}
 \edgedensity_2( \hC_a^1,\hC_b^1) \weq \frac{ \sum_{ \substack{ i \in \hC_a^1, j \in \hC_b^1 } } A_{ij}^2 }{ \left| \hC_a^1 \right| \cdot \left| \hC_b^1 \right| },
\end{align}
where $A^2$ is the adjacency matrix of $G^2$. We call this two-step technique \new{graph-splitting}.

 \begin{definition}[Graph-splitting] The graph-splitting of a graph $G = (V,E)$ with split-probability $\gamma$ generates two random graphs $G_1 = (V,E_1)$ and $G_2 = (V,E_2)$ having the same node set as $G$. The graph $G_1$ is formed by independently sampling each edge of $G$ with probability $\gamma$. The graph $G_2$ consists of the remaining edges that were not includesd in $G_1$,\textit{i.e.,} $E_2 = E \backslash E_1$. 
 \end{definition}

When $G$ is an HSBM with edge connection probabilities $p_{ab}$, the subgraphs $G_1$ and $G_2$ are also HSBMs with the same community structure as $G$, but with respective edge connection probabilities $\gamma p_{ab}$ and $(1-\gamma)p_{ab}$. As a result, graph-splitting provides an estimator $\hcC^1$ that is now independent of $G_2$, which solves the issue of correlation between $\hcC$ and $G$. However, this estimator still misclassifies $\Theta(N)$ nodes, and we make the following simplifying assumption regarding the misclassified nodes. 

 \begin{assumption}
 \label{assumption:estimator}
Assume that there exists a weakly consistent estimator $\hcC^1 = (\hC_1, \cdots, \hC_K)$, that is an estimate of $\cC = (C_1, \cdots, C_K)$, and let $O_{ab} = C_a \cap \hcC^1_b$ be the nodes in cluster $C_a$ but assigned to cluster~$\hcC^1_b$. We suppose that 
  \[ 
   \frac{ O_{ab} }{ N/K } \weq (1+o(1)) \, \zeta(|\lca(a,b)|),
  \]
  where $\zeta$ is a non-decreasing function of the depth satisfying $\zeta(d-1)+\zeta(d-2) < 2 \zeta(d)$. 
  
 \end{assumption}
 Because all communities of a BTSBM have expected size $N/K$, the quantity $\frac{ O_{ab} }{ N/K }$ represents the fraction of vertices in cluster $a$ that are clustered in cluster $b$. Hence, the assumption that $\zeta$ is non-decreasing implies that a node is more likely to be clustered into a cluster closer to its true cluster rather than one further away in the tree. Because $a=b$ if and only if $|\lca(a,b)|= d$, the quantity $\zeta(d)$ is the probability that a node is correctly clustered. In contrast, $\zeta(h)$ for $0 \le h \le d-1$ is the probability that a node from cluster $a$ is misclustered into cluster $b$ whose least common ancestor with $a$ is at depth $h$. The assumption $\zeta(d-1)+\zeta(d-2) < 2 \zeta(d)$ ensures that the proportion of misclustered nodes is not too large, and is automatically verified if $\zeta$ is strictly increasing. The following theorem establishes that Algorithm~\ref{algo:average_linkage} successfully recovers the hierarchy, even when starting with an estimator that makes $\Theta(N)$ mistakes.

 \begin{theorem}
\label{prop:bottomup_robustness}
 Let $G$ be a BTSBM. Suppose that Assumptions~\ref{assumption:fixed_quantites} and~\ref{assumption:scalings} hold and $N \delta_N = \Theta(1)$. Let $G_1$ and $G_2$ be the two graphs obtained from $G$ by graph-splitting with any $\gamma$ satisfying $\gamma = 1 - o(1)$ and $1 - \gamma = \omega(1/N)$. Let $\hcC^1$ be a clustering obtained from $G^1$ and satisfying Assumption~\ref{assumption:estimator}. Then Algorithm~\ref{algo:average_linkage} correctly recovers $\cT$ (using the edge density defined in Equation~\eqref{eq:def_edge_density_graph_splitting}). 
\end{theorem}
The assumption $1 - \gamma = \omega(1/N)$ ensures that the edge density of $G_2$ is $\omega(N^{-2})$, and hence that $G_2$ contains enough edges so that $\edgedensity_2( \hC_a^1,\hC_b^1)$ defined in Equation~\eqref{eq:def_edge_density_graph_splitting} concentrates around its expectation. The condition $\gamma = 1 - o(1)$ ensures that if weak recovery is possible from $G$, then it also remains possible from $G_1$: in other words, graph-splitting does not result in a loss of information.   

\begin{example}
\label{exa:uniform_mistakes}
 Denote by $\eta \ge 0$ the proportion of misclustered nodes and suppose that $\zeta(0) = \cdots = \zeta(d-1)$. This imposes that $\zeta(d) = 1-\eta$ and $\zeta(h) = \eta / (K-1)$ for $0 \le h \le d-1$. The condition $\zeta(d-1)+\zeta(d-2) < 2 \zeta(d)$ in Assumption~\ref{assumption:estimator} is equivalent to $\eta < 1 - 1 / K$.  
\end{example}

\section{Exact Recovery at Intermediate Depths}
\label{sec:intermediate_exact_recovery}

We now focus on the exact recovery of communities at intermediate depths within the hierarchy and establish tight information-theoretic conditions determining the feasibility of achieving exact recovery at a specific depth. It is intuitively expected that recovering super-communities of smaller depths should be comparatively easier than the recovery at larger depths. Therefore, we expect scenarios where the exact recovery of the primitive communities might be unattainable, but where the exact recovery of super-communities at intermediate depths remains achievable. To provide context, we initially recapitulate key findings on exact recovery in non-hierarchical stochastic block models (SBMs) in Section~\ref{subsec:CH}. We then present our main results in Section~\ref{subsec:exact_intermediate}, which gives the precise conditions for exact recovery at intermediate depths. Although the results of this section do not specifically involve hierarchical algorithms, we discuss in Section~\ref{sec:discussion_top-down_intermediate-depth} implications of exact recovery at intermediate depths for the theoretical guarantees of top-down algorithms.

\subsection{Chernoff-Hellinger Divergence}
\label{subsec:CH}
The hardness of separating nodes that belong to a primitive community $a \in \cL_{\cT}$ from nodes in community $b \in \cL_{\cT}$ is quantified by the \new{Chernoff-Hellinger divergence},\footnote{\cite{abbe2015community} originally defined the Chernoff-Hellinger divergence as 
$$\CH_{AS}( a, b ) \weq \sup_{ t \in (0,1) } (1-t) \sum_{ c \in \cL_{\cT} } \pi_c \left( tp_{ac} + (1-t) p_{bc} - p_{ac}^t p_{bc}^{1-t} \right).$$ When $p_{ab} =o(1)$, we have 
$(1-t) \dren_t\left( \Ber\left( p_{ac} \right) \| \Ber\left( p_{bc} \right) \right) = (1+o(1)) \left( tp_{ac} + (1-t) p_{bc} - p_{ac}^t p_{bc}^{1-t}  \right)$. Hence, our definition of $\CH(a,b)$ in~\eqref{eq:CH_definition} coincides with the original definition $\CH_{AS}(a,b)$ (up to second-order terms). 
} denoted by $\CH (a,b) = \CH( a, b, \pi, p, \cT )$ and defined by
\begin{align}
\label{eq:CH_definition}
   \CH( a, b ) \weq \sup_{ t \in (0,1) } (1-t) \sum_{ c \in \cL_{\cT} } \pi_c \dren_t\left( \Ber\left( p_{ac} \right) \| \Ber\left( p_{bc} \right) \right),
\end{align}
where $\dren_t$ denotes the \Renyi divergence of order $t$. 
The key quantity assessing the possibility or impossibility of exact recovery in an HSBM is the minimal {Chernoff-Hellinger divergence} between all pairs of clusters. We denote it by $I = I( \pi, p, \cT)$, and it is defined by
\begin{align}
\label{eq:divergence_HSBM}
 I \weq \min_{ \substack{ a,b \in [K] \\ a \ne b } } \ \CH( a, b ).
\end{align}
The condition for achieving exact recovery of the communities displays a phase transition. More precisely, when $p_{ab} = \Theta( \log N / N)$, exact recovery of the primitive communities is possible if $\lim \frac{N I}{\log N} > 1$ and is impossible when $\lim \frac{N I}{\log N} < 1$~\cite{abbe2015community,abbe2015recovering}. 

\begin{example}
\label{exa:full_recovery_BTSBM}
  Consider a BTSBM with $K=2^d$ communities and $\pi_a = 1/K$ for all $a \in [K]$. Suppose that for all~$t \in [d] \colon p_{t} = a_t \log N / N$, with $a_t$ positive constants. Simple computations (see for example~\cite[Lemma~6.6]{lei2020unified}) yield that 
  \begin{align*}
    I \weq \frac1K \drenh \left( \Ber(p_{d}) , \Ber(p_{d-1}) \right) 
    \weq \frac{1+o(1)}{K} \left( \sqrt{ a_d } - \sqrt{ a_{d-1} } \right)^2 \, \frac{\log N}{N}.
  \end{align*}
  Exact recovery of $\cC$ is possible if $ \left( \sqrt{ a_d } - \sqrt{ a_{d-1} } \right)^2 > K $. This condition involves only~$a_d$ and $a_{d-1}$, but not $a_{d-2}, a_{d-3}, \cdots, a_0$.
\end{example}

\subsection{Exact recovery at Intermediate Depths}
\label{subsec:exact_intermediate}

Although the exact recovery of the primitive communities is feasible when the quantity~$I$, as defined in Equation~\eqref{eq:divergence_HSBM}, exceeds the threshold of $\log N / N$, this insight alone does not provide the requirements for achieving exact recovery of the super-communities at a smaller depth. In this section, we answer this question by determining the information-theoretic threshold for achieving the exact recovery of super-communities at a specific depth. 

We say that a node $u$ of $\cT$ is at depth $q$ if its distance from the root is $q$, or if $u$ is a leaf whose distance from the root is less than $q$. The second condition  is only needed if the binary tree is not perfect, as we highlight in Example~\ref{exa:mega_communities}. We denote by $\cS_q$ (resp., $\cS_{\le q}$) the set of nodes at depth $q$ (resp., at a depth less than or equal to $q$), \textit{i.e.,}
\begin{align*}
 \cS_q \weq \Big\{ u \in \cT \colon \big( |u| = q \big) \text{ or } \big( |u| < q \text{ and } u \text{ is a leaf} \big) \Big\} \quad \text{ and } \quad \cS_{\le q} \weq \bigcup\limits_{r \le q} \cS_r.
\end{align*}
The set of \new{super-communities at depth $q$} is hence
\begin{align}
\label{eq:def_supercommunities}
 \mc\left( q, \cC, \cT \right) \weq \left\{ C_u \right\}_{ u \in \cS_q } \quad \text{ with } \quad C_u \weq \bigcup_{ a \in \cL_{\cT} \colon a_{1 : q } = u } C_{a}. 
\end{align}
\begin{example}
 \label{exa:mega_communities} For the tree of Figure~\ref{fig:illustration_HBM} we have
 \begin{align}
 \label{exa:mc_HBM}
 \mc\left( q, \cC, \cT \right) \weq 
 \begin{cases}
  \{ C_0, C_1 \} & \text{ if } q = 1, \\
  \{ C_{00}, C_{01}, C_{10}, C_{11} \} & \text{ if } q = 2, \\
  \{ C_{00}, C_{010}, C_{011}, C_{10}, C_{11} \} & \text{ if } q = 3.
 \end{cases}
 \end{align}
 Because the tree of Figure~\ref{fig:illustration_HBM} is not full and complete, the community $C_{10}$ is present at depth $2$ and at depth $3$. In contrast, for the tree of Figure~\ref{fig:illustration_BTSBM}, we have
 \begin{align*}
 \mc\left( q, \cC, \cT \right) \weq 
 \begin{cases}
  \{ C_0, C_1 \} & \text{ if } q = 1, \\
  \{ C_{00}, C_{01}, C_{10}, C_{11} \} & \text{ if } q = 2, \\
  \{ C_{000}, C_{001}, C_{010}, C_{011}, C_{100}, C_{101}, C_{110}, C_{111} \} & \text{ if } q = 3.
 \end{cases}
 \end{align*}
\end{example}

The recovery of super-communities at depth $q$ can be affected by mistakes made at depth $q' > q$. For example, if we consider the super-communities given in~\eqref{exa:mc_HBM}, then the recovery at depth $2$ of $C_{01}$ can be exact, even if nodes belonging to~$C_{010}$ are mistakenly classified in $C_{011}$. But, this is not the case if nodes in~$C_{010}$ are mistakenly classified in~$C_{00}$. 
As a result, the exact recovery of the super-communities at depth~$2$ might still be achievable even if the communities at depth~$3$ cannot be recovered, as long as the errors occur within a super-community. This can occur, for instance, if $C_{010}$ and $C_{011}$ are very hard to separate, whereas $C_{010}$, $C_{00}$, $C_{10}$ and $C_{11}$ are easy to separate. We recall from Section~\ref{section:tree_recovery_bottom} that the hardness to distinguish two communities is expressed in terms of a Chernoff-Hellinger divergence~\eqref{eq:CH_definition}. 
The difficulty in separating the primitive communities that do not belong to the same super-community at depth $q$ is quantified by the minimum Chernoff-Hellinger divergence taken across all pairs of primitive communities that do not belong to the same super-community at depth $q$. This is the quantity $I_q = I(q, \pi, p, \cT)$ defined as
\begin{align}
\label{eq:divergence_BTSBM_intermediateLevels}
 I_q \weq \min_{ \substack{ a \ne b \in \cL_{\cT} \\ \lca(a, b) \in \cS_{\le q-1} } } \CH\left( a, b \right),
\end{align}
where the condition $\lca(a, b) \in \cS_{\le q-1}$ ensures that the lowest common ancestor of $a$ and $b$ has a depth less than or equal to $q-1$. When $q = d$, the minimum in Equation~\eqref{eq:divergence_BTSBM_intermediateLevels} is taken over all pairs of primitive communities, and the divergence $I_d$ is equal to the divergence~$I$ defined in~\eqref{eq:divergence_HSBM}. 
Theorem~\ref{thm:exact_recovery_intermediate_levels} states the condition for recovering the communities at depth~$q$.

\begin{theorem}
\label{thm:exact_recovery_intermediate_levels}
 Let $G$ be an HSBM  
 \rebuttal{and suppose Assumption~\ref{assumption:fixed_quantites} and~\ref{assumption:scalings} hold.} 
 Let $q \in \{1, \cdots, d_{\cT} \}$ and denote by $I_q$ the quantity defined in~\eqref{eq:divergence_BTSBM_intermediateLevels}. 
The following holds:
\begin{enumerate}[(i)]
 \item exact recovery of the super-communities at depth~$q$ is impossible if $\limsup \frac{N I_q}{\log N} < 1 $; 
 \item \label{thm:exact_recovery_intermediate_levels_positive_statement} if $\liminf \frac{N I_q}{\log N} > 1$, 
 then exact recovery of the super-communities at depth~$q$ is possible.
\end{enumerate}
\end{theorem}

Because $I_q$ is the minimum taken over all $\cS_{\le q-1}$, the quantity $I_q$ is non-increasing in~$q$. This is the reason why exact recovery at a lower intermediate depth $q'$ is easier than recovery at a higher depth $q > q'$. 
We prove Theorem~\ref{thm:exact_recovery_intermediate_levels} in Appendix~\ref{appendix:proof_thm:exact_recovery_intermediate_levels}. Although the quantity~$I_q$ defined in~\eqref{eq:divergence_BTSBM_intermediateLevels} is generally hard to simplify, the following lemma provides a simple expression for~BTSBMs.

\begin{lemma}
\label{lemma:expression_Iq_BTSBM}
 For a BTSBM with $K=2^d$ balanced communities ($\pi_a = 1/K$ for all $a \in \cL_\cT$), the minimum Chernoff-Hellinger divergence defined in Equation~\eqref{eq:divergence_BTSBM_intermediateLevels} is
 \begin{align*}
  I_q \weq \frac1K  \left( \dren_{1/2}( \Ber(p_{q-1}), \Ber( p_d ) ) + \sum_{ k=1 }^{d-q} 2^{k-1} \dren_{1/2} \left( \Ber(p_{q-1}), \Ber( p_{d-k} ) \right) \right).
 \end{align*}
\end{lemma}
Lemma~\ref{lemma:expression_Iq_BTSBM} shows that when $q = d$, we have $I_{ d } = K^{-1} \dren_{1/2} \left( \Ber(p_{d-1}), \Ber(p_d) \right) $, and hence we recover the divergence $I$ defined in~\eqref{eq:divergence_HSBM} (see also Example~\ref{exa:full_recovery_BTSBM}).

\section{Discussion}
\label{sec:discussion}

\subsection{Previous Work on Exact Recovery in HSBM}
\label{sec:discussion_previous_work}

Early works by Dasgupta et al.~\cite{dasgupta2006spectral} establish conditions for the exact recovery by {top-down} algorithms in an HSBM. Nonetheless, the recovery is ensured only in relatively dense regimes (specifically, average degree at least $\log^6 N$), and the algorithm requires an unspecified choice of hyper-parameters. Balakrishnan et al.~\cite{balakrishnan2011noise} show that a simple recursive spectral bi-partitioning algorithm recovers the hierarchical communities in a class of hierarchically structured weighted networks, where the weighted interactions are perturbed by sub-Gaussian noise. Nonetheless, the conditions require again a dense regime (\cite[Assumption 1]{balakrishnan2011noise} states that all weights are strictly positive, hence the average degree scales as $\Theta(N)$), and the proposed algorithm has no stopping criterion. 
Recent work by Li et al.~\cite{li2022hierarchical} demonstrates that the same recursive spectral bi-partitioning algorithm exactly recovers the communities, when the average degree grows as $\Omega( \log^{2+\epsilon} N )$ for some $\epsilon >0$ (this condition can be further relaxed to $\Omega( \log N )$ by a refined analysis~\cite{lei2020unified}). The analysis of~\cite{li2022hierarchical} also allows for an unbounded number of communities and provides a consistent stopping criterion. 

The \textit{linkage++} algorithm proposed by Cohen-Addad et al.~\cite{cohen2017hierarchical,cohen2019hierarchical} is a {bottom-up} algorithm that first estimates $K$ primitive communities by SVD and then successively merges them by using a {linkage} procedure. As SVD requires the number of clusters $K$ as an input, the authors propose to run the algorithms for every $K = 2, \cdots, O(\log N)$ and to choose the optimal $K$ that leads to the hierarchical tree with the smallest Dasgupta-cost~\cite{dasgupta2016cost}. Moreover, their analysis requires a dense setting in which the average degree grows as $\Omega( \sqrt{N \log N} )$. 

Finally, in a recent paper, Fang and Rohe introduce the T-Stochastic Graph model \cite{fang2023t}. This model captures hierarchical structures in networks by defining connection probabilities based on distances within an unrooted and potentially non-binary tree $T$. To reconstruct the latent hierarchy, the authors propose a bottom-up algorithm that integrates spectral clustering with the Neighbor-Joining algorithm. Unlike average linkage, which greedily merges the closest clusters, Neighbor-Joining iteratively selects pairs of clusters that minimize a criterion incorporating both pairwise distances and their average distances to all other nodes. Their key theoretical result, Theorem 4.3, establishes asymptotic recovery of the latent hierarchy when the average degree of the graph is $\Omega(\log^{11.1} n)$. 


\subsection{Top-Down HCD and Exact Recovery at Intermediate Depth}
\label{sec:discussion_top-down_intermediate-depth}

 Exact recovery of communities at intermediate hierarchical depths is essential to the success of top-down algorithms, as current theoretical analyses \cite{li2022hierarchical,lei2020unified} assume perfect bi-partitioning at each of the first $\ell$ recursive steps. This requires exact recovery of communities from depth $0$ to depth $\ell$ in the hierarchy, which is only achievable above the information-theoretic threshold for exact recovery at depth $\ell$. Because exact recovery requires graph connectivity, the average degree must be $\Omega(\log N)$. In the following section, we show that the information-theoretic threshold is less stringent than the conditions in \cite{li2022hierarchical,lei2020unified}. Although there may be room to improve on the conditions of \cite{li2022hierarchical,lei2020unified}, current proof techniques cannot guarantee hierarchical recovery when the average degree is $o(\log N)$.

 This raises the following question: \textit{Can we establish theoretical guarantees for top-down algorithms when the average degree grows slower than $\log N$?} Current proof techniques are constrained by their requirement that all the $\ell$ first recursive steps must be error-free. This requirement ensures that the divisive algorithm recurse with HSBMs at each step. Indeed, consider an HSBM $G=(V,E)$ where a bi-partitioning algorithm produces partition $\hcC = (\hC_1, \hC_2)$ of $V$, yielding the subgraphs $\hG_1 = G[\hC_1]$ and $\hG_2 = G[\hC_2]$. When $\hcC$ exactly recovers depth-1 communities, we have either $\hC_1 = C_1$ and $\hC_2 = C_1$, or $\hC_1 = C_1$ and $\hC_2 = C_2$, ensuring that $\hG_1$ and $\hG_2$ are HSBMs. 
 
 This property no longer holds when $\hcC$ contains errors. Misclustered nodes create dependencies between edges in $\hG_1$ (resp., in $\hG_2$), as their assignments are determined by the bi-partitioning algorithm applied to $G$, so that $\hG_1$ and $\hG_2$ might bot be HSBMs any more. These correlations are complex to analyze, but we can bypass this difficulty and reframe the problem by viewing $\hG_1$ and $\hG_2$ as two perturbed HSBMs, where the perturbations result from connectivity modifications made by an adversary over a subset of nodes, whose size is upper-bounded by the number of misclustering errors. This formulation aligns with recent work on clustering robustness under adversarial perturbations~\cite{cai2015robust,liu2022minimax,ding2023reaching}, and could strengthen the theoretical guarantees for top-down algorithms. Furthermore, this adversarial framework explains why top-down algorithms using spectral clustering—a method known to be sensitive to arbitrary perturbations—are suboptimal.

\subsection{Previous Work on Exact Recovery at Intermediate Depth}

Sufficient conditions for the exact recovery of super-communities in BTSBMs by recursive spectral bi-partitioning algorithms were established in~\cite{li2022hierarchical}, under the assumptions that all connection probabilities have the same scaling factor $\delta_N$ (\textit{i.e.,} $p(t) = \Theta ( \delta_N )$ for all $t \in \cT$) verifying $\delta_N = \Omega(\log^{2+\epsilon} N / N)$. The result was later refined in~\cite{lei2020unified} to assume only that $\delta_N = \Omega(\log N / N)$
and in~\cite{lei2021consistency} to enable different scaling factors at different depths of the hierarchy. 
To compare these results with Theorem~\ref{thm:exact_recovery_intermediate_levels}, let us consider a BTSBM with $p(t) = a_{|t|} \log N / N$ where $a_{|t|}$ is constant independent of $N$. The average connection probability in a super-community of depth $q$ is then $\ba_q \log N / N$ where $\ba_q = \frac{1}{2^{d-q} }\left( a_d + \sum_{k=1}^{d-q} 2^{k-1} a_{d-k} \right)$. Then \cite[Theorem~6.5]{lei2020unified} states that the {recursive spectral bi-partitioning} algorithm exactly recovers all the super-communities up to depth $\ell$ if
\begin{align}
\label{eq:condition_intermediate_levels_top_down}
\left( \sqrt{ \ba_q} - \sqrt{a_{q-1}} \right)^2 > 2^q \quad \text{ for all } q \in [ \ell ].
\end{align}
The condition $\left( \sqrt{ \ba_q} - \sqrt{a_{q-1}} \right)^2 > 2^q $ is the exact recovery threshold of a symmetric SBM with $2^q$ communities and intra-community (resp., inter-community) link probabilities $\ba_q \log N / N$ (resp., $a_{q-1} \log N / N$). Although the depth~$q$ of the hierarchy is composed of $2^q$ super-communities with average intra-connection probability $\ba_q \log N / N$ and inter-connection probability $a_{q-1} \log N / N$, Condition~\eqref{eq:condition_intermediate_levels_top_down} does not match the information-theoretic threshold derived in Theorem~\ref{thm:exact_recovery_intermediate_levels}. 
Indeed, the structure of a super-community at depth $q$ is not an \Erdos-\Renyi random graph with connection probability $\ba_q \log N / N$, but is instead an SBM with $2^{d-q}$ primitive communities.

For all $q \in [d]$, let us define $J_q^{ \tda }$ and $J_q^{ \bua }$ by
\begin{align}
\label{eq:def_Jqtd}
 J_q^{ \tda } & \weq \frac{1}{2^d} \left( \sqrt{ a_d + \sum_{k=1}^{ d-q } 2^{k-1} a_{d-k}  } - \sqrt{ 2^{d-q} a_{q-1} } \right)^2,  \\
 \label{eq:def_Jqbua}
 J_q^{ \bua } & \weq \frac{1}{2^d} \left( \left( \sqrt{ a_{q-1} } - \sqrt{ a_d } \right)^2 + \sum_{ k=1 }^{d-q} 2^{k-1} \left( \sqrt{ a_{q-1} } - \sqrt{ a_{d-k} } \right)^2 \right).
\end{align}
Condition~\eqref{eq:condition_intermediate_levels_top_down} for exact recovery of the super-communities up to depth $\ell$ is equivalent~to
\begin{align}
\label{eq:conditions_intermediate_exact_recovery_top-down}
 \min_{ q \in [\ell] } J_q^{ \tda } > 1, 
\end{align}
whereas the corresponding condition for a bottom-up algorithm given by Theorem~\ref{thm:exact_recovery_intermediate_levels}\eqref{thm:exact_recovery_intermediate_levels_positive_statement} and Lemma~\ref{lemma:expression_Iq_BTSBM} can be recast~as
\begin{align}
\label{eq:conditions_intermediate_exact_recovery_bottom-up}
J_\ell^{ \bua } > 1.
\end{align}

Let us compare~\eqref{eq:conditions_intermediate_exact_recovery_top-down} and~\eqref{eq:conditions_intermediate_exact_recovery_bottom-up}. First, notice that for some choices of $a_0, \cdots, a_d$, the function $q \mapsto J_q^{ \tda }$ given in~\eqref{eq:def_Jqtd} is not monotonically non-increasing,\footnote{Let $d=3$. For $(a_0,a_1,a_2,a_3) = (2.2, 2.5, 3, 25)$ we have $(J_1^{\tda}, J_2^{\tda}, J_3^{\tda}) = (0.96, 1.17, 1.33)$, but for $(a_0,a_1,a_2,a_3) = (3,9,15,21)$ we have $(J_1^{\tda}, J_2^{\tda}, J_3^{\tda}) = (1.89, 0.39, 0.06)$. Finally, for $(a_0,a_1,a_2,a_3) = (2.2,2.4,4,22)$ we have $(J_1^{\tda}, J_2^{\tda}, J_3^{\tda}) = (0.85,1.02,0.90)$. Thus, the quantities $J_1^{\tda}, J_2^{\tda}, J_3^{\tda}$ can be increasing, decreasing or interlacing.} which contradicts the intuitive fact that the recovery at larger depths is harder than the recovery at smaller depths. Moreover, although $J_d^\bua = J_d^{\tda}$, Lemma~\ref{lemma:bottom_up_vs_top_down} (deferred in Appendix~\ref{subsection:proof_bottom_up_vs_top_down}) shows that $J_q^{ \bua } > J_q^{ \tda }$ for all $q \in [d-1]$, hence Condition~\eqref{eq:conditions_intermediate_exact_recovery_top-down} is strictly more restrictive than Condition~\eqref{eq:conditions_intermediate_exact_recovery_bottom-up}. 
It is not clear whether the results of~\cite{lei2020unified} can be further improved as exact recovery by {spectral bi-partitioning} requires entry-wise concentration of the eigenvectors, which is established using sophisticated $\ell_{2 \to \infty}$ perturbation bounds.

\section{Numerical Results}
\label{section:numerical_results}

In this section, we investigate the empirical performance of different HCD strategies on synthetic and real networks. More precisely, we compare Algorithm~\ref{algo:average_linkage} (bottom-up) to the recursive bi-partitioning algorithm of~\cite{li2022hierarchical} (top-down). To obtain the flat clustering input for Algorithm~\ref{algo:average_linkage}, we use spectral clustering with the Bethe-Hessian matrix, as proposed in~\cite{dall2021unified}. This spectral algorithm does not require prior knowledge of the number of communities $K$ and has been shown to perform well on both synthetic and real datasets.\footnote{The code for this algorithm is available online at \url{https://lorenzodallamico.github.io/codes/}.} Moreover, recent theoretical work has established that the Bethe-Hessian matrix can be used to recover both the number of communities and the communities~\citep{stephan2024community}. 
Our code is available on GitHub at \url{https://github.com/daichikuroda/bottom_up_HCD}.

\subsection{Synthetic Data Sets}

We compare bottom-up and top-down approaches on synthetic data sets by computing the \new{accuracy at different depths}. We define accuracy at depth $q$ as $1 - \ace \left(\mc(q,\cC,\cT), \mc(q,\hcC,\hcT) \right) / N$, where $\ace$ is given by~\eqref{eq:def_ace} and $\mc(q,\cC,\cT)$ (resp., $\mc(q,\hcC,\hcT)$) are the ground truth (resp., predicted) super-communities defined in~\eqref{eq:def_supercommunities}.

\subsubsection{Binary Tree SBMs}
\label{subsection:experiments_BTSBMs}
We first generate synthetic BTSBMs of depth $d$, where the probability of an edge between two nodes whose lowest common ancestor has depth $k$ is $p_k = a_k \log N / N$. We compare the accuracy of {top-down} and {bottom-up} at each depth. We show in Figure~\ref{fig:performance_algos_balanced_BTSBM} the results obtained on BTSBMs of depth $3$, with $400$ nodes in each primitive community (thus $N = 2^3 \times 400 = 3200$ nodes in total). We let $a_0 = 40$, $a_3 = 100$, and the values of $a_1$ and~$a_2$ vary in the range ($a_0, a_3$), with the condition $a_1 < a_2$. Solid lines in each panel show the exact recovery threshold of the given method at different depths. We observe the strong alignment between the theoretical guarantees and the results of the numerical~simulations.

\begin{figure}[!ht]
 \centering
 \includegraphics[width=0.9\textwidth]{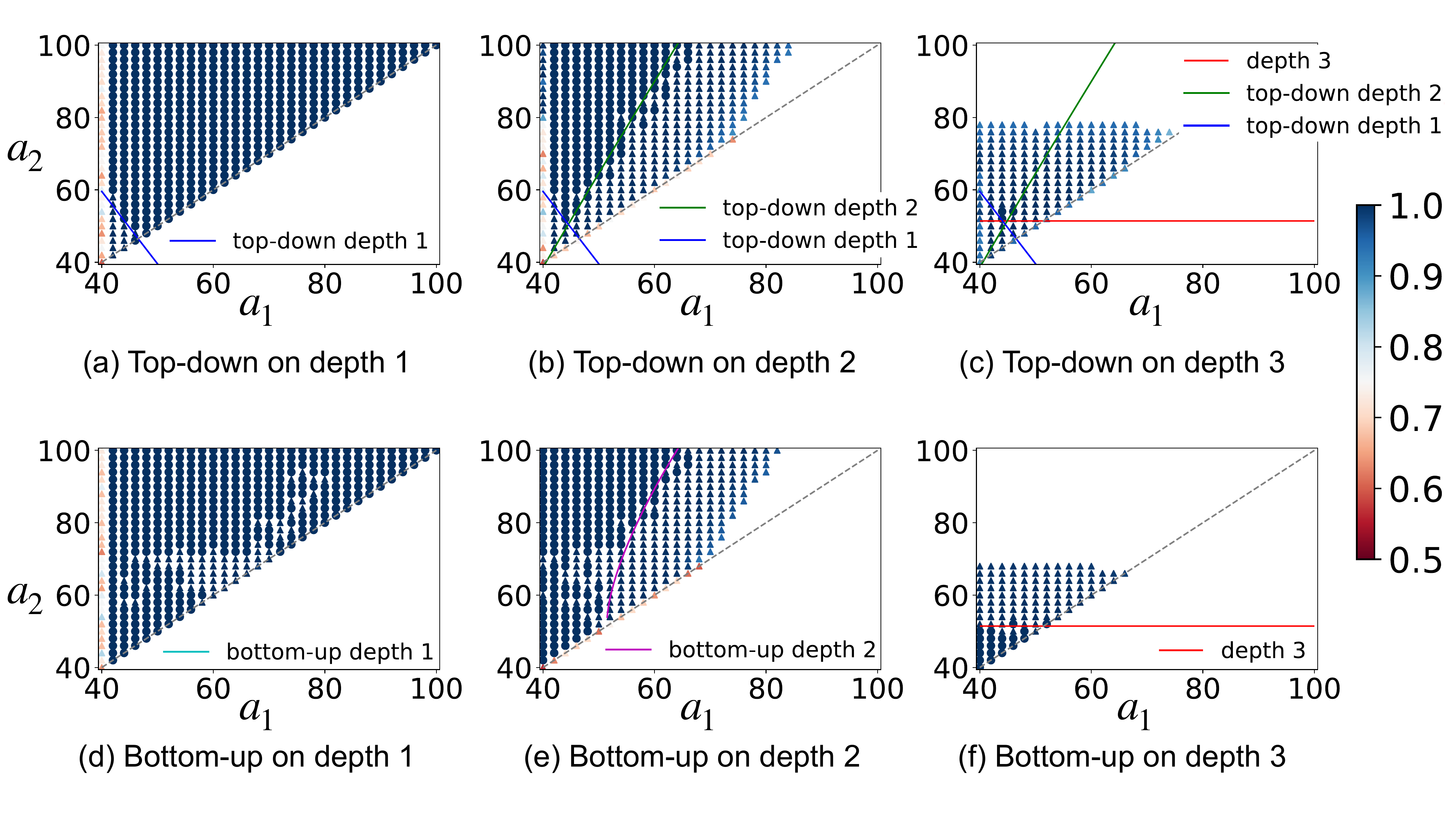}
 \caption{Performance of {bottom-up} and {top-down} algorithms on BTSBMs of depth~3, $ N = 3200$ nodes, and interaction probabilities $p_k = a_k \log N / N$, where $a_0 = 40$ and $a_3 = 100$, as a function of $a_1$ and $a_2$. We vary $a_1 \le a_2$ from~$a_0$ to~$a_3$.
 The empirical performance of the algorithms is measured by the accuracy at each depth, given by the color scale (results are averaged over 10 realizations). 
 Large circles represent exact recovery (\textit{i.e.,} perfect accuracy on each of the 10 runs), and small crosses represent a non-exact recovery. 
 The colored solid lines delimit the theoretical exact recovery thresholds for each algorithm on the various depths (given by Equations~\eqref{eq:conditions_intermediate_exact_recovery_top-down}-\eqref{eq:conditions_intermediate_exact_recovery_bottom-up}); for a given depth $q$, these equations provide a single condition for bottom-up, but $q$ conditions for top-down. At depths~1 and~2, the regimes where exact recovery can be achieved are the areas above the solid line(s). At depth~3, the area lies below the threshold drawn by the red line (and above the blue and green lines for top-down; this area forms a small triangle).  
 }
 \label{fig:performance_algos_balanced_BTSBM}
\end{figure}

\subsubsection{Robustness of Linkage to Misclustering Errors}
\label{subsec:robustness_numerical}
Finally, we evaluate the performance of Algorithm~\ref{algo:average_linkage} when the number of misclassified nodes increases. We generate BTSBMs with full and balanced trees of depth $3$ with 200 nodes in each community ($N=1600$ in total), and  interaction probabilities given by $p_k = 0.08 \beta^{3-k}$, where $\beta$ is a parameter varying between $0$ and $1$. At the two extremes, $\beta = 0$ leads to a graph with $K$ communities disconnected from each other, while $\beta=1$ leads to an \Erdos-\Renyi graph without community structure. 

Figures~\ref{fig:bounded_confusion1} and~\ref{fig:bounded_confusion2} show two confusion matrices, which are $K$-by-$K$ matrices where each element $(i,j)$ is $|C_i \cap \hC_{\tau^*(j)}|$, for two different average degree. We observe that mislabeled nodes are typically assigned to a community close to their true one. For instance, a misclassified node from community 0 may be incorrectly clustered into community 1 or 2, but it is highly unlikely to be placed in community 7. To reinforce this observation and further evaluate the validity of Assumption~\ref{assumption:estimator}, we estimate the probability of misclassification at level~$\ell \in [K]$, corresponding to the quantity $\zeta(\ell)$ in Assumption~\ref{assumption:estimator}, by
\begin{align*}
  \hzeta(\ell) \weq \frac{1}{N} \sum_{k=1}^K \sum_{m=k}^K \1\{|\lca(C_k, C_m)| = \ell \} \cdot \left|C_{k} \cap \hC_{\tau^*(m)}\right|,
\end{align*}
where $\tau^*= \argmin_{\tau \in \cS_{[K]} } \sum_{ k=1 }^K \left| C_{k} \symdiff \hC_{\tau(k)} \right|$ is the optimal permutation of cluster labels. Figure~\ref{fig:bounded_hzeta} shows that, for any given $\beta$, we have $\hzeta(0) < \hzeta(1) < \hzeta(2) < \hzeta(3)$, showing that~$\hzeta$ is non-decreasing with $\ell$. This observation aligns with Assumption~\ref{assumption:estimator}. 
Finally, Figure~\ref{fig:bounded_treeRecovery} shows the tree recovery success rates with and without graph-splitting. A successful tree recovery occurs when the following condition is satisfied:
\begin{align*}
  \prod_{k=1}^K \prod_{m=k+1}^K \1\left\{|\lca(C_k, C_m)| \weq |\lca(\cC_{\tau^*(k)}, \cC_{\tau^*(m)})| \right\}.
\end{align*}
The legend indicates the value of $\gamma$ used for graph splitting, which represents the proportion of edges allocated to estimating bottom community labels via the flat clustering algorithm, while the remaining edges are used to construct the hierarchy with Algorithm~\ref{algo:average_linkage}. We observe that Algorithm~\ref{algo:average_linkage} maintains a high tree recovery rate even when the estimated bottom clusters contain many misclustered nodes. 
 When the expected degree is~$5$, Figure~\ref{fig:bounded_hzeta} shows that $\hzeta(d) = \hzeta(3) \approx 0.6$, meaning that approximately 40\% of the nodes are misclustered. Nevertheless, Figure~\ref{fig:bounded_treeRecovery} indicates that Algorithm~\ref{algo:average_linkage} achieves a tree recovery success rate of around 50\%. This demonstrates the robustness of Algorithm~\ref{algo:average_linkage} to a high proportion of misclustered nodes. Furthermore, while we introduced graph splitting for the theoretical analysis, our empirical results show that it may not be needed in practice.

\begin{newtext}
\begin{figure}[!ht]
\centering
  \begin{subfigure}[b]{0.4\textwidth}
  \centering
  \includegraphics[width=\textwidth]{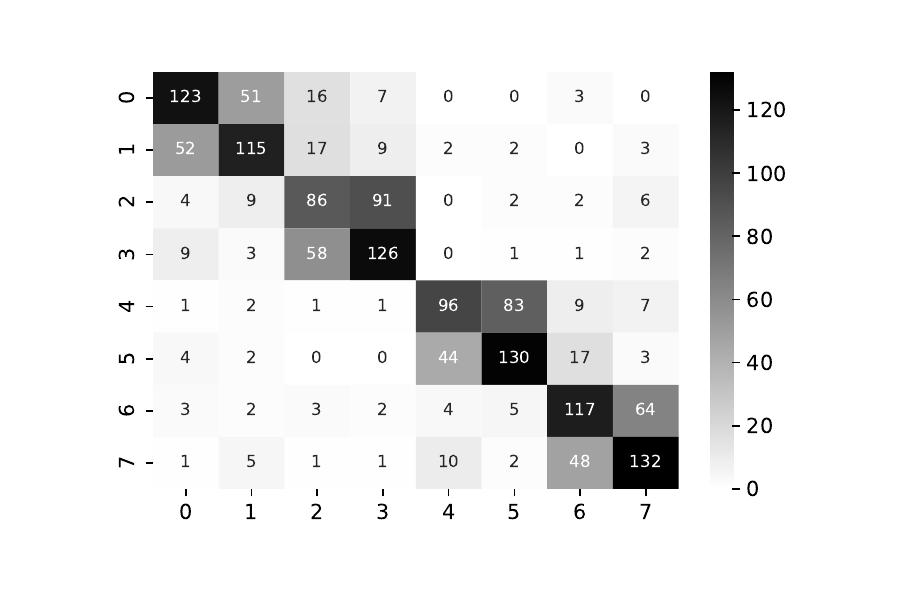}
  \caption{Average degree $5$}
  \label{fig:bounded_confusion1}
 \end{subfigure}
 \hfil
\begin{subfigure}[b]{0.4\textwidth}
  \centering
  \includegraphics[width=1\textwidth]{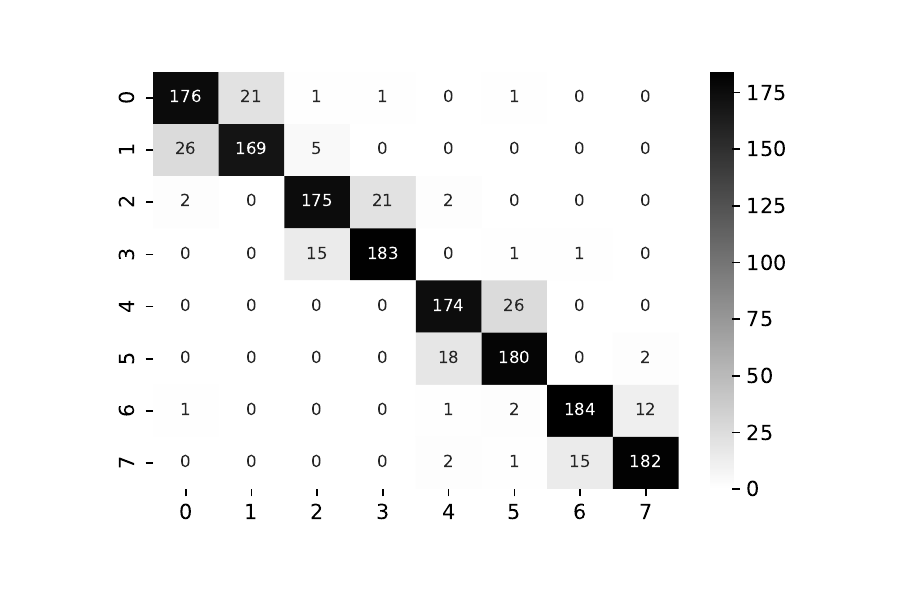}
  \caption{Average degree 10}
  \label{fig:bounded_confusion2}
 \end{subfigure}
\hfil
\begin{subfigure}[b]{0.4\textwidth}
\centering
     \includegraphics[width=\textwidth]{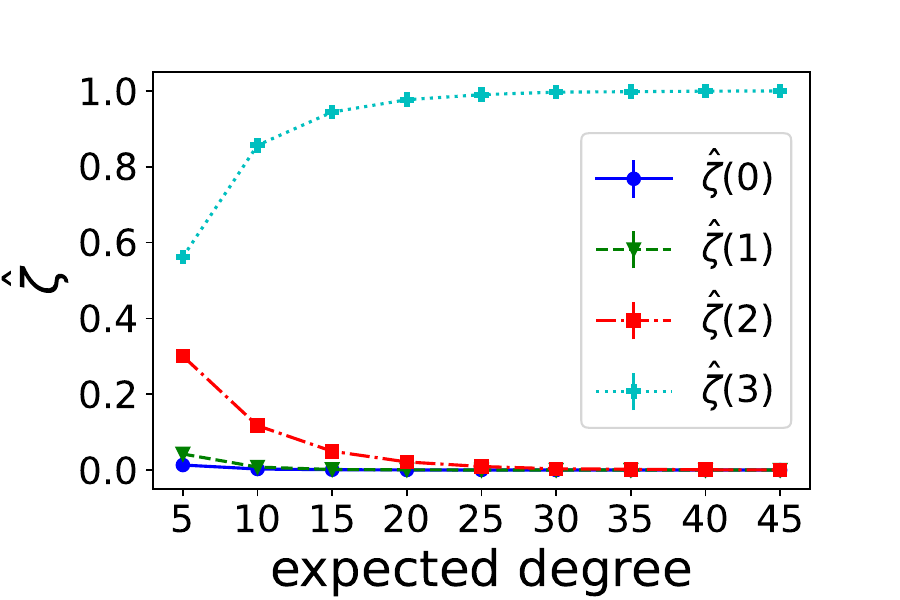}
     \caption{Types of errors}
       \label{fig:bounded_hzeta}
 \end{subfigure} \hfil 
\begin{subfigure}[b]{0.4\textwidth}
\centering
  \includegraphics[width=\textwidth]{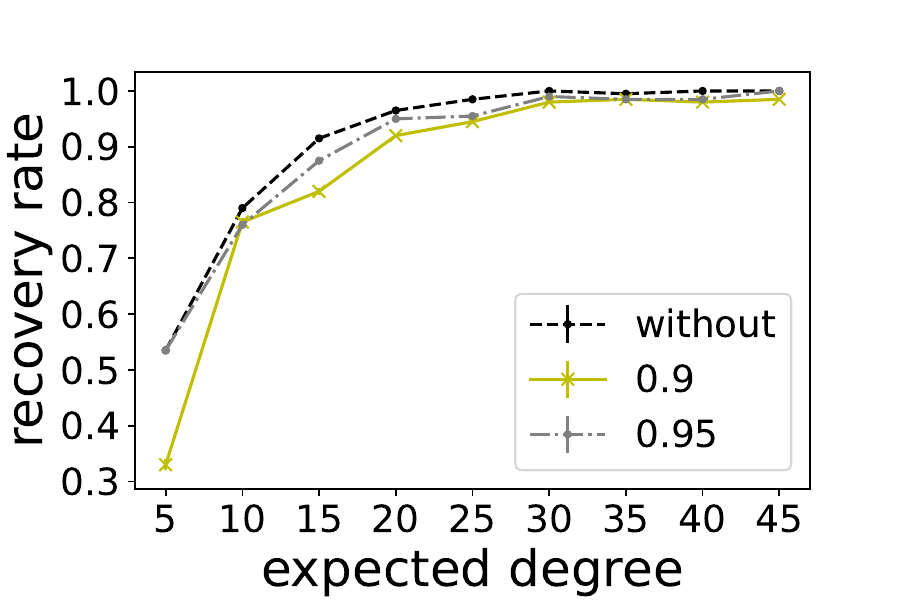}
  \caption{Tree recovery rate
  }
  \label{fig:bounded_treeRecovery}
 \end{subfigure}
 \caption{BTSBMs with depth 3, $N = 1600$, and $\beta = 0.3$. Figures~\ref{fig:bounded_confusion1} and~\ref{fig:bounded_confusion2} show two confusion matrices when the expected degree equals 5 and 10. 
 Figure~\ref{fig:bounded_hzeta} shows the evolution of $\zeta$ as a function of the expected degree.
 Figure~\ref{fig:bounded_treeRecovery} shows the tree recovery success rate with and without graph splitting. 
 Results of Figures~\ref{fig:bounded_hzeta} and~\ref{fig:bounded_treeRecovery} are averaged over 200~realizations.}
\label{fig:bounded_degree}
\end{figure}
\end{newtext}

\subsection{Real Datasets}

\subsubsection{High-School Contact Dataset}

 We evaluate HCD algorithms on a dataset of face-to-face interactions among 327 high school students from Lycée Thiers in Marseilles, France~\cite{Mastrandrea_Fournet_Barrat_2015} (\url{http://www.sociopatterns.org/}). The network consists of 9 class-based communities, with weighted edges representing proximity encounters over five days. These classes also correspond to four academic specializations: mathematics \& physics (MP), physics \& chemistry (PC), engineering (PSI), and biology (BIO). The MP and BIO groups contain three sub-classes each, while PC has two and PSI one. The hierarchical structure should reflect specialties at higher depths and individual classes at lower depths.

The results of both HCD algorithms are shown in Figure~\ref{fig:results_highschool}. The bottom-up algorithm consistently detects 31 communities, while the {top-down} algorithm predicts an average of 8.93 (8 communities in 7 runs, 9 in 93 runs). Both align well with the ground truth, with adjusted mutual information (AMI) scores of 0.938 for bottom-up\footnote{AMI is computed after merging the 31 bottom-up communities into 9 super-communities.} and 0.945 for top-down. Notably, bottom-up recovers both class and specialization structures while also detecting smaller subgroups, likely representing groups of friend inside a same class, revealing a richer hierarchical structure than the ground truth itself.

\begin{figure}[!ht]
 \centering
 \includegraphics[width=0.8\textwidth]{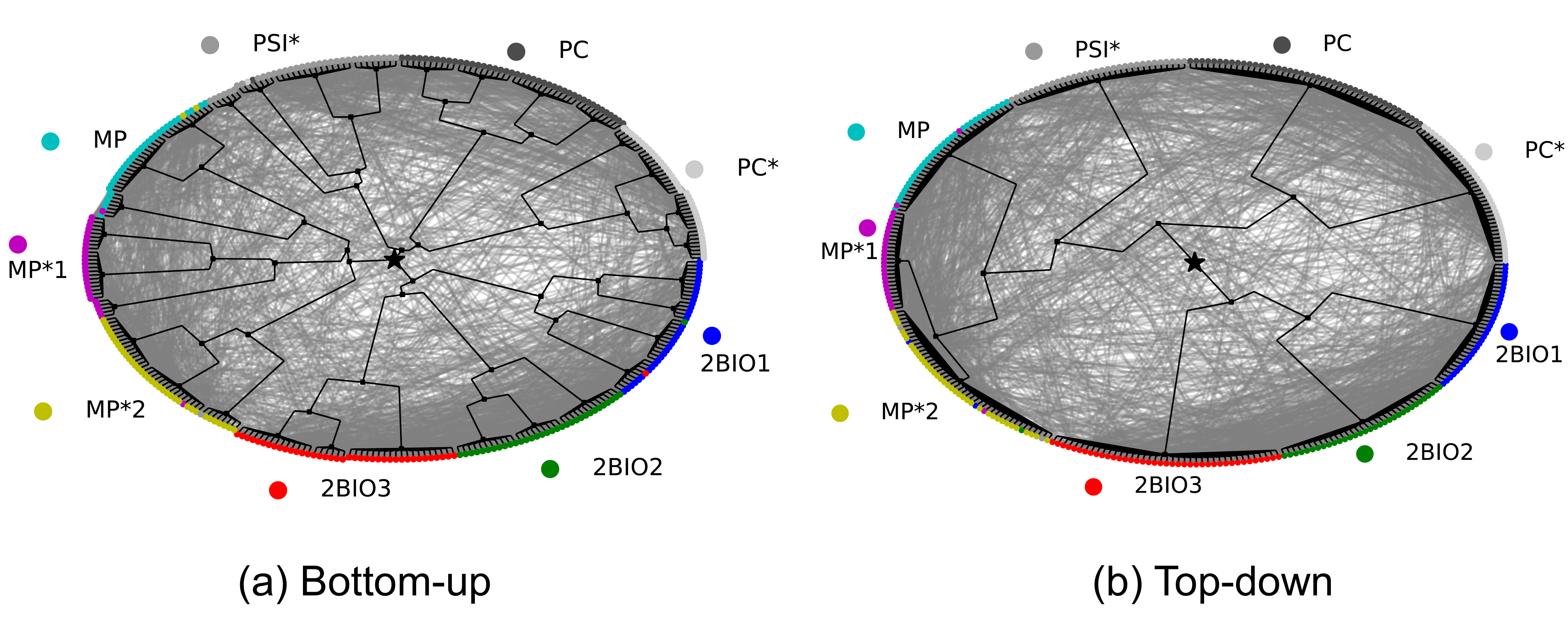}
 \caption{{Bottom-up} and {top-down} algorithms on the {high school} data set. Nodes correspond to the students, colors to the true classes, and edges of the graph are in grey. The hierarchical tree is drawn in black, and its root is marked by a star symbol.}
 \label{fig:results_highschool}
\end{figure}

\subsubsection{Power Grid}
We next consider the power grid of continental Europe from the map of the Union for the Coordination of Transmission of Electricity (UCTE). We use the same data set that was previously used for HCD~\cite{schaub2012markov}. Figures~\ref{fig:bottom_up_power_grid} and~\ref{fig:top_dwon_power_grid} show the outputs of the bottom-up and top-down, respectively. We observe a higher correlation to geographical positions in the output of bottom-up. Moreover, the dendrogram obtained by top-down shows a significant amount of inversions, contrary to the one obtained by bottom-up. 

\begin{figure}[!ht]
\centering
\begin{subfigure}{0.32\textwidth}
    \includegraphics[width=\textwidth]{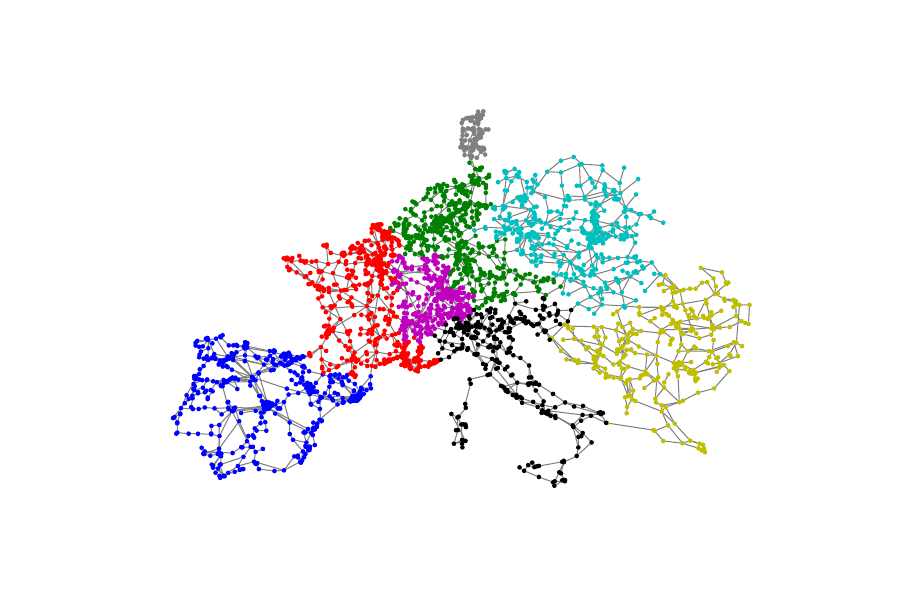}
    \caption{8 communities.}
\end{subfigure}    
\hfil
\begin{subfigure}{0.32\textwidth}
    \includegraphics[width=\textwidth]{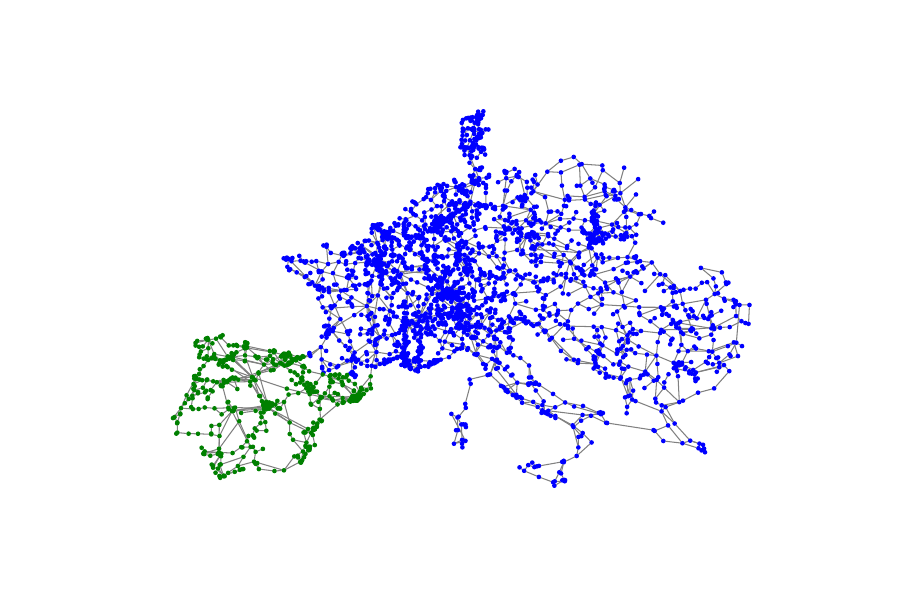}
    \caption{2 communities.}
\end{subfigure}
\hfil
\begin{subfigure}{0.32\textwidth}
    \includegraphics[width=\textwidth]{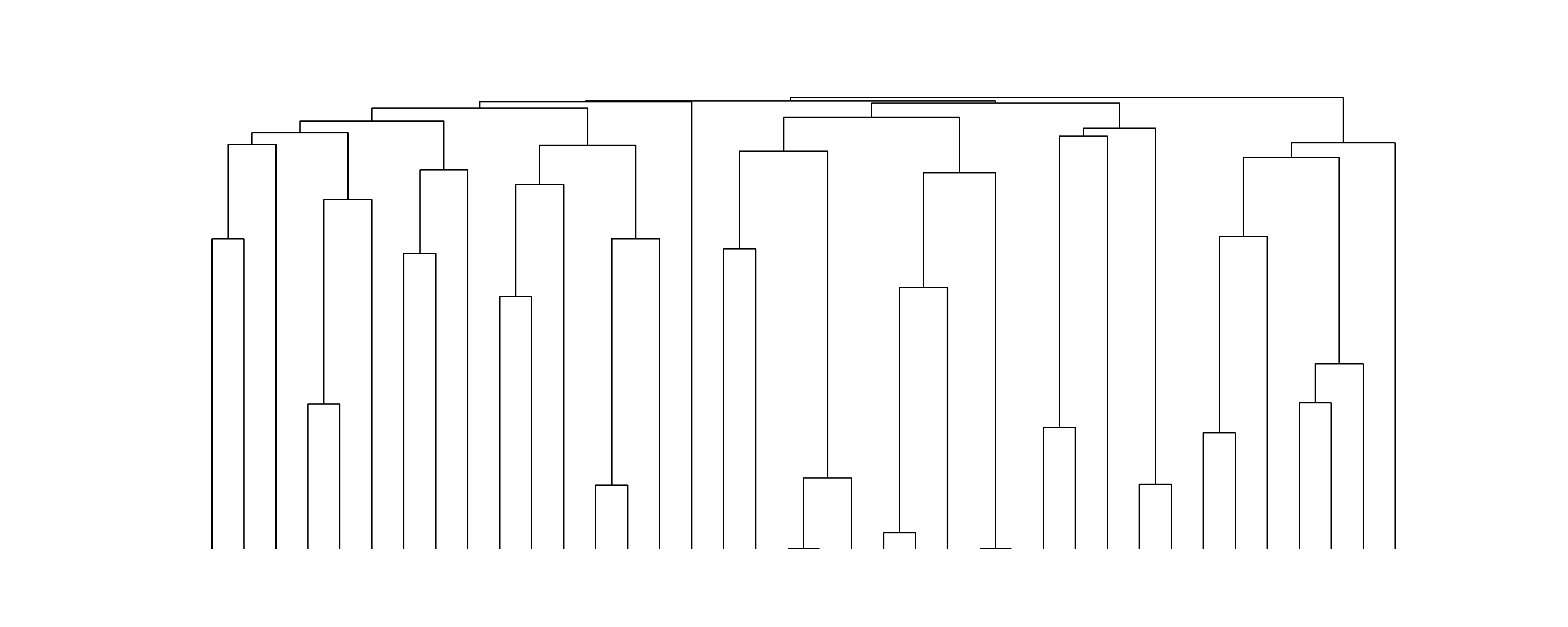}
    \caption{Dendrogram.}
\end{subfigure}    
\caption{{Bottom-up} algorithm on the power-grid network.}
\label{fig:bottom_up_power_grid}
\end{figure}

\begin{figure}[!ht]
\centering
\begin{subfigure}{0.32\textwidth}
    \includegraphics[width=\textwidth]{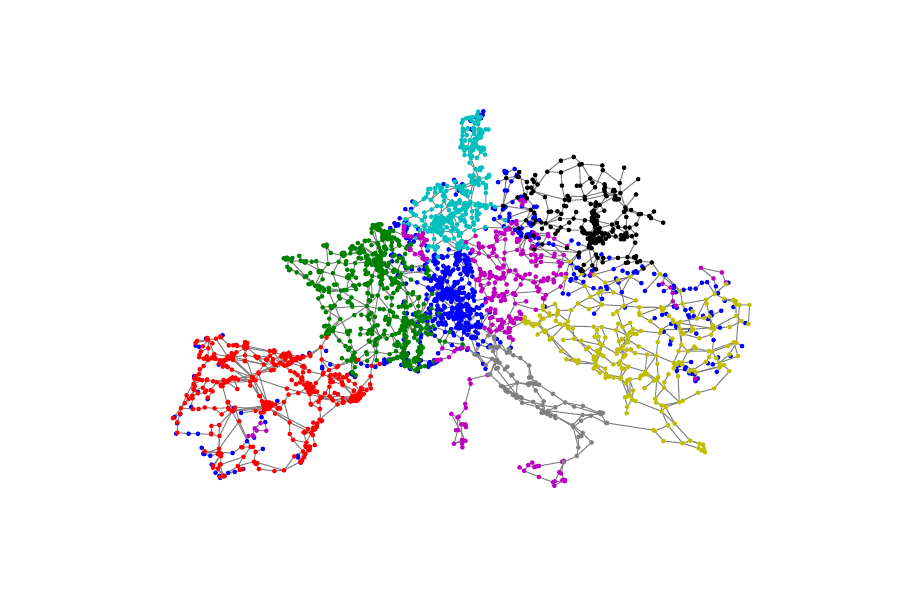}
    \caption{8 communities.}
\end{subfigure}    
\hfill
\begin{subfigure}{0.32\textwidth}
    \includegraphics[width=\textwidth]{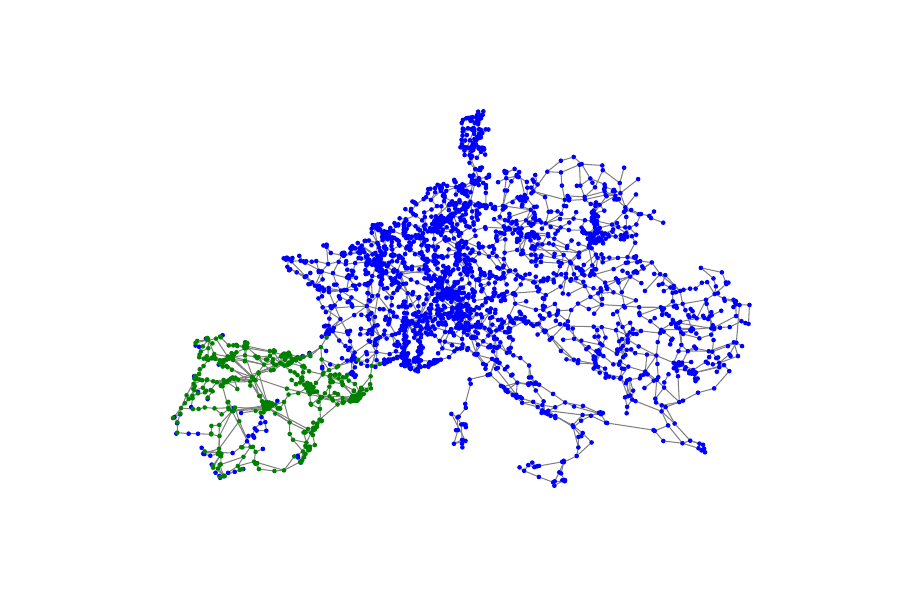}
    \caption{2 communities.}
\end{subfigure}
\hfill
\begin{subfigure}{0.32\textwidth}
    \includegraphics[width=\textwidth]{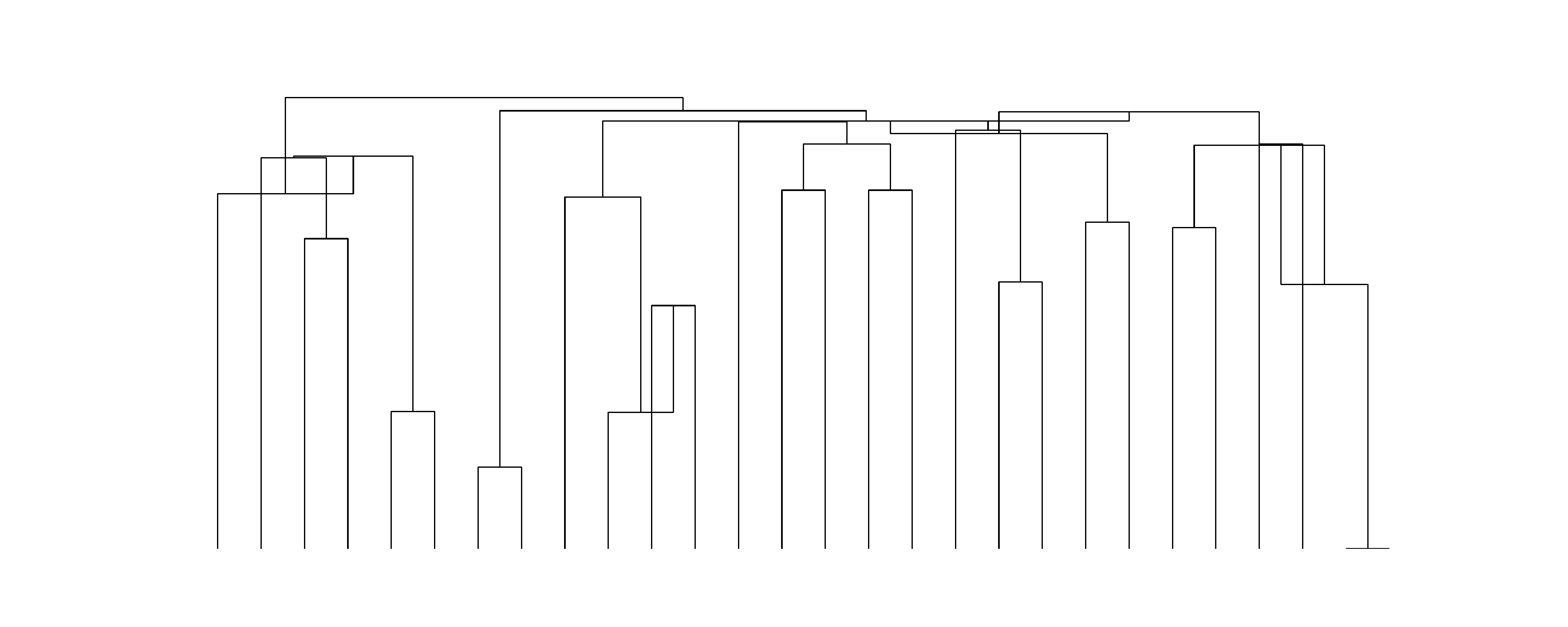}
    \caption{Dendrogram.}
\end{subfigure}    
\caption{{Recursive spectral bi-partitioning} algorithm on the power grid network.}
\label{fig:top_dwon_power_grid}
\end{figure}

\section{Conclusion}

Top-down approaches need to make partition decisions for large communities, without exploiting side information about the internal structure of these communities. 
In a sparse regime, finding a bi-partition is as difficult as finding a flat community structure. As a consequence, some nodes are misclassified, and these misclassifications are increasingly locked in as the algorithm progresses down to smaller communities. 
In contrast, a bottom-up approach inherently exploits lower-level structure, if present. At each step, a bottom-up algorithm needs to classify only the communities at the subsequent higher level, rather than classifying all the nodes individually. This is an easier problem (even if some errors are carried up), because (a) the number of classification decisions is much smaller, and (b) the number of edges available for each decision is much larger (number of edges between two lower-depth communities versus number of incident edges on an individual node). 

In this paper, we have quantified this fundamental advantage within a class of random graph models (HSBM). We have proven that the latent tree of an HSBM can be recovered under weaker conditions than in the literature (average degree scaling as $\omega(1)$). Moreover, we have established that super-communities of intermediate depths could be exactly recovered up to the information-theoretic threshold, thus improving upon the previously known conditions for top-down algorithms. Finally, we have shown that the theoretic advantage of bottom-up carries over to relevant scales and real-world data. Both on synthetic and on real data sets, bottom-up HCD achieves better performance than {top-down} HCD. 

\section*{Acknowledgements}
This work was supported in part by the Swiss National Science Foundation (SNSF) under grant IZBRZ2\_186313.

\bibliographystyle{alpha}
\bibliography{main.bib}

\clearpage

\appendix
\addtocontents{toc}{\protect\setcounter{tocdepth}{1}}

{
  \bigskip
  \bigskip
  \bigskip
  \begin{center}
    {\LARGE\bf Supplementary material \\ When does bottom-up beat top-down in hierarchical community detection? }
\end{center}
  \medskip
}

\section{Proofs for Section~\ref{sec:tree_recovery}}

\subsection{Proof of Theorem~\ref{thm:performance_bottomUp}}
\label{section:proof_thm:performance_bottomUp}

\begin{proof}[Proof of Theorem~\ref{thm:performance_bottomUp}]

With the assumption of this theorem, $\hcC$ is an almost exact estimator of the true clusters $\cC$. 
Let $a \ne b \in \cL_{\cT}$, and recall the definition of the edge density $\edgedensity(\cdot, \cdot)$  in~\eqref{eq:def_edgedensity}.  Lemma~\ref{lemma:convergence_empiricalDistance_betweenClusters} implies that with high probability the edge density between the estimated clusters $\hC_a$ and $\hC_b$ is $\edgedensity( \hC_a, \hC_b ) = (1+o(1)) p_{ab}$. Since the network is assortative and the hierarchy is non-flat, this implies that the first step of the linkage finds the two closest clusters, say $\hC_{a_1}$ and~$\hC_{a_2}$. Then we see that for $b \in \cL_{\cT} \backslash \{a_1, a_2\}$ we have $\lca(a_1,b) = \lca(a_2,b)$ and hence $p_{a_1b} = p_{a_2 b}$. 
Therefore, from~\eqref{eq:extension_edgedensity_set_of_clusters},
\begin{align*}
 \rho\left( \hC_{a_1 \cup a_2}, \hC_{b} \right) \weq \frac{\left| \hC_{a_1} \right|}{ \left| \hC_{a_1} \right| + \left| \hC_{a_2} \right|} \rho\left( \hC_{a_1}, \hC_{b} \right) + 
 \frac{\left| \hC_{a_2} \right|}{ \left| \hC_{a_1} \right| + \left| \hC_{a_2} \right|} \rho\left( \hC_{a_2}, \hC_{b} \right) \weq (1+o(1)) p_{a_1 b}.
\end{align*}
So $\rho \left( \hC_{a_1 \cup a_2}, \hC_{b} \right) = (1+o(1)) p_{a_1 b}$, and repeating this argument, it follows by induction that the \textit{average-linkage} procedure correctly recovers the tree. 
\end{proof}

Let us now state and prove the following auxiliary lemma.

\begin{lemma}
\label{lemma:convergence_empiricalDistance_betweenClusters}
Consider an HSBM with the same assumptions as in Theorem~\ref{thm:performance_bottomUp}, and let $\hcC$ be an almost exact estimator of $\cC$ ($\hcC$ is possibly correlated with the graph edges). Then, for any $a, b \in \cL_{\cT}$ we have $\edgedensity \left( \hC_a, \hC_b \right) = (1+o(1))p_{ab}$.
\end{lemma}

\begin{proof} 

We denote by $Z \in \{0,1\}^{N \times K}$ (resp., by $\hZ$) the one-hot representation of the true communities $\cC$ (resp., of the predicted communities $\hcC$), that is $Z_{ic} = \1(i \in C_c)$ and $\hZ_{ic} = \1(i\in \hC_c)$ for all $i\in [N]$ and $c \in \cL_{\cT}$ (where $\1(\cdot)$ denotes the indicator function). 
Let $a \ne b \in \cL_{\cT}$. 
We shorten the edge density $\edgedensity\left( \hC_a, \hC_b \right)$ by~$\hp_{ab}$. From the definition of the edge density in~\eqref{eq:def_edgedensity}, we have 
\begin{align*}
 \hp_{ab} \weq \frac{ w \left( \hC_a, \hC_b \right) }{ \left| \hC_a \right| \cdot \left| \hC_b \right| } 
 \weq \frac{ \sum_{  i \in \hC_a, j \in \hC_b } A_{ij} }{ \left| \hC_a \right| \cdot \left| \hC_b \right| }
 \weq 
  \frac{ \sum_{i,j} \hZ_{ia} \hZ_{jb} A_{ij} }{ \sum_{i,j} \hZ_{ia} \hZ_{jb} }.
\end{align*}
Therefore, a variance-bias decomposition leads to 
\begin{align}
\label{eq:in_proof:concentration_hpab}
 \left| \hp_{ab} - p_{ab} \right| & \wle 
 \underbrace{\left| \frac{ \sum_{i,j} \hZ_{ia} \hZ_{jb} \left( A_{ij} - \E A_{ij} \right) }{ \sum_{i,j} \hZ_{ia} \hZ_{jb} }  \right|}_{E_1}
 + 
 \underbrace{ \left| \frac{ \sum_{i,j} \hZ_{ia} \hZ_{jb} \E A_{ij} }{ \sum_{i,j} \hZ_{ia} \hZ_{jb} } - p_{ab} \right|}_{E_2}.
\end{align}
Moreover, because $\hcC$ is almost exact, we have $\sum_{i,j} \hZ_{ia} \hZ_{jb} = (1+o(1)) |C_a| \cdot |C_b|$. 
Thus, 
\begin{align*}
 \sum_{i,j} \hZ_{ia} \hZ_{jb}  \weq (1+o(1)) \pi_a \pi_b N^2, 
\end{align*}
using the concentration of multinomial random variables.

Let us bound the two terms $E_1$ and $E_2$ on the right-hand side of~\eqref{eq:in_proof:concentration_hpab} separately. 
To handle the first term, we will use~\cite[Lemma~4.1]{guedon2016community}. Let 
$\bp = 2 N^{-1} (N-1)^{-1}\sum_{i < j} \Var( A_{ij}) $. Because $\Var(A_{ij}) = \Theta(\delta_N)$, we have $\bp = \Theta(\delta_N)$. Moreover, because $\delta_N = \omega( N^{-1} )$, we have $\bp \ge 9 N^{-1}$ for $N$ large enough. 
Therefore \cite[Lemma~4.1]{guedon2016community} ensures that with probability at least $1-e^3 5^{-N}$ we have
\begin{align*}
 \sup_{s,t \in \{-1,1\}^n } \left| \sum_{i,j} \left( A_{ij} - \E A_{ij} \right) s_i t_j \right| \wle 3 N \sqrt{ N \bp}.
\end{align*}
Applying Grothendieck's inequality~\cite[Theorem 3.1]{guedon2016community}, we obtain
\begin{align}
\label{eq:bound_cut_Grothendieck}
 \sup_{ \substack{X_1, \cdots, X_N \\ Y_1, \cdots, Y_N \\ \forall i \in [N] \colon \|X_i \|_2 \le 1 \\ \forall j \in [N] \colon \|Y_j \|_2 \le 1 } } \left| \sum_{i,j} \left( A_{ij} - \E A_{ij} \right) X_i^T Y_j \right| \wle 3 c N \sqrt{ N \bp }, 
\end{align}
where $c$ is Grothendieck’s constant, verifying $0 < c < 2$. Using~\eqref{eq:bound_cut_Grothendieck} with $X_i = \hZ_{ia}$ and $Y_j = \hZ_{jb}$ leads to 
\begin{align*}
 E_1 \weq O\left( \sqrt{ \frac{ \bp }{ N } } \right) \weq o( \delta_N ),
\end{align*}
because $p = \Theta(\delta_N)$ and $N^{-1} = o(\delta_N)$. 

To handle the second term $E_2$ in the right-hand side of~\eqref{eq:in_proof:concentration_hpab}, we first notice that
\begin{align}
\label{eq:in_proof_E2_rewritten}
 E_2 & \weq \frac{ \left| \sum\limits_{i,j} \hZ_{ia} \hZ_{jb} \left( \E A_{ij} - p_{ab} \right) \right| }{ \sum\limits_{i,j} \hZ_{ia} \hZ_{jb} }. 
\end{align}
Moreover, for all $i,j \in [N]$ we have 
\begin{align*}
 \E A_{ij} = \sum_{a', b' \in \cL_{\cT} } p_{a' b'} Z_{ia'}Z_{jb'},
\end{align*}
because $\E A_{ij} = p_{a' b'}$ if $i \in C_{a'}$ and $j \in C_{b'}$. Thus, 
\begin{align*}
 \E A_{ij} - p_{ab} \weq \sum_{ \substack{ a', b' \in \cL_{\cT} \\ (a', b') \ne (a,b) } } (p_{a' b'} - p_{ab}) Z_{ia'}Z_{jb'}.
\end{align*}
Let $c = \max_{a', b'} \frac{p_{a'b'}}{p_{ab}}$. Under our assumptions, $c = \Theta(1)$. We bound the numerator appearing in Equation~\eqref{eq:in_proof_E2_rewritten} as follows:
\begin{align*}
 \left| \sum\limits_{i<j} \hZ_{ia} \hZ_{jb} \left( \E A_{ij} - p_{ab} \right) \right| 
 & \wle p_{ab} \left( c + 1 \right)  \sum_{ \substack{ a', b' \in \cL_{\cT} \\ (a', b') \ne (a,b) } } \sum_{i < j } \hZ_{ia} \hZ_{jb} Z_{ia'}Z_{jb'} \\
 & \wle 2c p_{ab} \sum_{ \substack{ a', b' \in \cL_{\cT} \\ (a', b') \ne (a,b) } } \left| \hC_{a} \cap C_{a'} \right| \cdot \left| \hC_{b} \cap C_{b'}  \right|, 
\end{align*}
since $c \ge 1$. 
Let us denote $V \backslash C_a =  \cup_{a' \ne a} C_{a'}$ by $C_a^c$. We have
\begin{align*}
 \sum_{ \substack{ a', b' \in \cL_{\cT} \\ (a', b') \ne (a,b) } } \left| \hC_{a} \cap C_{a'} \right| \cdot \left| \hC_{b} \cap C_{b'} \right| 
 & \weq \sum_{a' \ne a } \sum_{b' \in \cL_{\cT} } \left| \hC_{a} \cap C_{a'} \right| \cdot \left| \hC_{b} \cap C_{b'} \right| + \left| \hC_{a} \cap C_{a} \right| \sum_{b' \ne b} \left| \hC_{b} \cap C_{b'} \right| \\
 & \weq \left| \hC_a \cap C^c_a \right| \cdot \left| \hC_b \right| + \left| \hC_a \cap C_a \right| \cdot \left| \hC_b \cap C_b^c \right| \\
 & \wle \left( \left| \hC_b \right| + \left| \hC_a \right| \right) \ace \left( \cC,\hcC \right), 
\end{align*}
where the last line uses $|\hC_a \cap C_a| \le |\hC_a|$ and $ \ace(\cC, \hcC) = \sum_{c =1}^K \left| \hC_c \cap C^c_c \right| + \left| \hC_c^c \cap C_c \right|$. 
Hence, 
\begin{align*}
 E_2 
 \wle 2c p_{ab} \left( \frac{ 1 }{ |\hcC_a| } +  \frac{ 1 }{ |\hcC_b| } \right) \ace(\cC, \hcC) .
\end{align*}
The assumption that $\hcC$ is almost exact ensures that $\ace(\cC, \hcC) = o(N)$ and $|\hC_a|, |\hC_b| = \Theta(N)$. Therefore, $E_2 = o(p_{ab})$. This concludes the proof.  
\end{proof}

\subsection{Proofs for Section~\ref{section:tree_recovery_bounded_degree}}
\label{appendix:proof_bottomup_robustness}

In this section, we prove the Theorem~\ref{prop:bottomup_robustness} and Proposition~\ref{prop:robustness_adversarial}. We first start in Section~\ref{subsec:generalLemma_robustness} by stating a lemma that gives a general condition for the robustness of average-linkage with respect to errors present in an estimator $\tcC$ of $\cC$, where we assume that $\tcC$ is independent of the graph $G$. We then show in Sections~\ref{subsec:proof_thm_robustness} and~\ref{subsec:proof_prop_robustness_adversarial} how to prove Theorem~\ref{prop:bottomup_robustness} and Proposition~\ref{prop:robustness_adversarial} from this Lemma~\ref{lemma:bottomup_robustness_general}. The main ingredients of the proof of Lemma~\ref{lemma:bottomup_robustness_general} are given in Section~\ref{subsec:proof_generalLemma}, whereas the most tedious computations are detailed in Section~\ref{subsubsec:detailed_computation}. 

\subsubsection{A General Lemma}
\label{subsec:generalLemma_robustness}
Let us denote by $B(h)$ the number of bottom communities that have similarity $h$ with a bottom community~$a$. By symmetry of BTSBMs, this value $B(h)$ does not depend on $a$, and we have 
\begin{align}
B(h) & \weq
 \begin{cases}
    2^{d-h-1} & h \in \{0, 1, \cdots, d-1 \}, \\
    1  & h = d. 
 \end{cases}
\label{eq:def_B}
\end{align} 
We notice for $0 \leq h_1 \leq d$, $B(h)$ satisfies 
\begin{align}
    \sum_{h=h_1}^d B(h) \weq B(h_1 -1),
\label{eq:sum_B_to_B}
\end{align}
where by a slight abuse of notation we define $B(-1) = 2^{d}$. The value of $B(h)$ is also equal to the number of bottom communities whose similarity with the bottom community $a$ is no less than $h+1$.
Denote by $\bp_{h}$ the expected edge density inside the super community whose similarity with the bottom community $a$ is no less than $h$. By symmetry of BTSBMs, the value of $\bp_{h}$ does not depend on $a$, and we have  
\begin{align}
  \bp_{h} \weq \frac{\sum_{s=h}^d B(s) p(s)}{B(h-1)}.
\label{eq:p_bar_B}
\end{align}
Let us now state the following lemma. 
\begin{lemma}
\label{lemma:bottomup_robustness_general}
 Let $G$ be a BTSBM whose latent binary tree $\cT$ has depth $d_{\cT} \ge 2$. Suppose that $\pi_a = 1/K$ for all $a \in \cL_{\cT}$ and that $\min\limits_{u \in \cT} p(u) = \omega( N^{-2} )$. 
Let $\tcC$ be a clustering obtained from $\cC$ as described in Section~\ref{section:tree_recovery_bottom}, and $\hcT$ be the hierarchical tree obtained by average-linkage from~$\tcC$. 
Then $\hcT = \cT$ iff for all $ h_{ac}\leq d-2$ the following condition is satisfied 
\begin{align}
\label{eq:robustness_threshold}
& \sum_{h_1=h_{ac}+1}^{d-1} \sum_{h_2 = h_1+1}^{d} 2 B(h_1) B(h_2) \left(\zeta(h_1) - \zeta(h_{ac}) \right) \left(\zeta(h_2) - \zeta(h_{ac}) \right) \left( p_{h_1} - p_{h_{ac}} \right) \notag\\
+ & \ \sum_{h_1=h_{ac}+1}^{d-1} \left(B^2(h_1) \left( \zeta(h_1) - \zeta(h_{ac}) \right)^2 \left( \bp_{h_1+1} - p_{h_{ac}} \right) \right) + \left( p_{d-1} - p_{h_{ac}} \right) \left(\zeta(d) - \zeta(h_{ac})\right)^2 \notag\\
+ & \quad 2 \left( p_d - p_{d-1} \right) \left(\zeta(d-1) - \zeta(h_{ac}) \right) \left( \zeta(d) - \frac{\zeta(d-1) + \zeta(h_{ac})}{2} \right) \notag\\
 > & \ 0. 
\end{align}
\end{lemma}
The proof of Lemma~\ref{lemma:bottomup_robustness_general} is tedious, and we defer it to Section~\ref{subsec:proof_generalLemma}. 
Let us first establish Theorem~\ref{prop:bottomup_robustness} and Proposition~\ref{prop:robustness_adversarial} from Lemma~\ref{lemma:bottomup_robustness_general}.

\subsubsection{Proof of Theorem~\ref{prop:bottomup_robustness}}
\label{subsec:proof_thm_robustness}

\begin{proof}[Proof of Theorem~\ref{prop:bottomup_robustness} as a corollary of Lemma~\ref{lemma:bottomup_robustness_general}]
Because of assortativity, $p_{h_2}-p_{h_1} > 0$ for all $h_2 > h_1$ and $\bp_h > p_{h_{ac}}$ for all $h > h_{ac}$. 
Moreover, we have $\zeta(h_1) \wle \zeta(h_2)$ and $\zeta(h) - \zeta(h_{ac})\geq 0$ for all $h_1 < h_2$ and $h > h_{ac}$. 
Therefore, for all $h_{ac} \leq d-2$, we have 
\begin{align*}
\sum_{h_1=h_{ac}+1}^{d-1} & \sum_{h_2 = h_1+1}^{d} 2 B(h_1) B(h_2) \left(\zeta(h_1) - \zeta(h_{ac}) \right) \left(\zeta(h_2) - \zeta(h_{ac}) \right) \left( p_{h_1} - p_{h_{ac}} \right) \notag\\
+ & \sum_{h_1=h_{ac}+1}^{d-1} B^2(h_1) \left( \zeta(h_1) - \zeta(h_{ac}) \right)^2 \left( \bp_{h_1+1} - p_{h_{ac}} \right) + \left( p_{d-1} - p_{h_{ac}} \right) \left(\zeta(d) - \zeta(h_{ac})\right)^2 
\wge 0,
\end{align*}
because each term in the left-hand side of the inequality is non-negative. 

Moreover, because $\zeta(d) > \frac{\zeta(d-1) + \zeta(h_{ac})}{2}$,
$$
\left( p_d - p_{d-1} \right) \left(\zeta(d-1) - \zeta(h_{ac}) \right) \left( \zeta(d) - \frac{\zeta(d-1) + \zeta(h_{ac})}{2} \right) > 0, 
$$
and thus Condition~\eqref{eq:robustness_threshold} is satisfied.
\end{proof}

\subsubsection{An Adversarial Setting where \texorpdfstring{$\zeta$}{zeta} is not Non-decreasing}
\label{subsec:proof_prop_robustness_adversarial}

In an adversarial setting, $\zeta$ might not be non-decreasing. For instance, consider a scenario where misclustered nodes are deliberately assigned to a cluster chosen uniformly at random from those farthest away from their true cluster. Define $\tcC_{\adversarial} = \tcC(\cC,\zeta_{\adversarial})$ where  
\begin{align}
\label{eq:adversarial_setting}
 \zeta_{\adversarial}(h) \weq 
 \begin{cases}
    \frac{\eta}{2^{d-1}} & \text{ if } h = 0, \\
    0 & \text{ if } h \in \{1,\cdots, d-1\}. 
 \end{cases}
\end{align}
Under additional conditions given in Proposition~\ref{prop:robustness_adversarial}, the average-linkage procedure recovers the tree from $\tcC_{\adversarial}$. 

\begin{proposition}
\label{prop:robustness_adversarial}
 Let $G$ be a BTSBM whose latent binary tree $\cT$ has depth $K = \Theta(1)$ leaves. Suppose that $\pi_a = 1/K$ for all $a \in \cL_{\cT}$ and that $\min\limits_{u \in \cT} p(u) = \omega( N^{-2} )$. 
Suppose that $\eta < 1/2$. 
Then the {average-linkage} procedure correctly recovers $\cT$ (starting from~$\tcC_{\adversarial}$) if one of the following conditions is verified:
\[
 (i) \ \bp_1 < p_{d-1} \qquad \text{ or } \qquad (ii) \ \bp_1 \ge p_{d-1} \text{ and } \eta < \eta_-.
\]
where $\bp_1 = \frac{1}{2^{d-1}} \left( p_d + \sum_{k=1}^{d-1} 2^{k-1} p_{d-k} \right)$ and $\eta_- = \frac{p_{d-1}+\bp_1 -2p_0-\sqrt{(\bp_1 - p_{d-1})(\bp_1 - p_0)}}{p_{d-1}+3\bp_1 -4p_0}$.
\end{proposition}

\begin{proof}[Proof of Proposition~\ref{prop:robustness_adversarial} as a corollary of Lemma~\ref{lemma:bottomup_robustness_general}]
When $h_{ac} > 0$, by substituting $\zeta(h)$ with Equation~\eqref{eq:adversarial_setting}, the left hand side of Condition~\eqref{eq:robustness_threshold} can be written as 
\begin{align*}
\left( p_{d-1} - p_{h_{ac}} \right) \left(1 - \eta\right)^2 > 0,
\end{align*}
which therefore satisfies Condition~\eqref{eq:robustness_threshold}.

When $h_{ac} = 0$, we can replace $\zeta(h)$ with Equation~\eqref{eq:adversarial_setting} and denote $\frac{\eta}{2^{d-1}}$ by $\zeta$ for convenience. This allows us to express the left-hand side of Condition~\eqref{eq:robustness_threshold} as follows:
\begin{align}
\label{eq:proposition1_case2_before}
    & \underbrace{\sum_{h_1=1}^{d-1} \sum_{h_2 = h_1+1}^{d-1} 2 B(h_1) B(h_2) \zeta^2 \left( p_{h_1} - p_0 \right)}_{A_1} - \underbrace{\sum_{h_1=1}^{d-1} 2 B(h_1) \zeta \left(1-\eta - \zeta \right) \left( p_{h_1} - p_0 \right)}_{A_2} \notag \\
    + & \ \underbrace{\sum_{h_1=1}^{d-1} B^2(h_1) \zeta^2 \left( \bp_{h_1+1} - p_0 \right)}_{A_3} 
    + \left( p_{d-1} - p_0 \right) \left(1 - \eta - \zeta \right)^2 - 2 \left( p_d - p_{d-1} \right) \zeta \left( 1 - \eta - \frac{ \zeta}{2} \right). 
\end{align}
We notice that
\begin{align}
    A_1 &\weq \sum_{h_1=1}^{d-1} \sum_{h_2 = h_1+1}^{d} 2 B(h_1) B(h_2) \zeta^2 \left( p_{h_1} - p_0 \right) - \sum_{h_1=1}^{d-1} 2 B(h_1) B(d) \zeta^2 \left( p_{h_1} - p_0 \right), \notag\\
    &=^{(a)} \sum_{h_1=1}^{d-1} 2 B^2(h_1) \zeta^2 \left( p_{h_1} - p_0 \right) - \sum_{h_1=1}^{d-1} 2 B(h_1) B(d) \zeta^2 \left( p_{h_1} - p_0 \right),
\label{eq:A_1}
\end{align}
where in (a), we use Equation~\eqref{eq:sum_B_to_B}.
We also have 
\begin{align}
    A_3 &= \sum_{h_1=1}^{d-1} \sum_{h_2 = h_1+1}^d B(h_1) B(h_2) \zeta^2 \left( p_{h_2} - p_0 \right), \notag\\
    &=^{(a)} \sum_{h_2=1}^{d} \sum_{h_1 = 1}^{h_2-1} B(h_1) B(h_2) \left( p_{h_2} - p_0 \right) \zeta^2, \notag\\
    &= \sum_{h_2=1}^{d} \left(\sum_{h_1 = 1}^{d} B(h_1) - \sum_{h_1 = h_2}^{d}B(h_1) \right) B(h_2) \left( p_{h_2} - p_0 \right) \zeta^2, \notag\\
    &=^{(b)} \sum_{h_2=1}^{d} \left(B(0) - B(h_2-1)\right) B(h_2)\left( p_{h_2} - p_0 \right) \zeta^2, \notag\\
    &=^{(c)} \sum_{h_1=1}^{d} B(0)B(h_1)\zeta^2\left( p_{h_1} - p_0 \right) - \left( \sum_{h_1=1}^{d-1} 2B^2(h_1)\zeta^2 \left( p_{h_1} - p_0 \right) + \zeta^2 \left( p_d - p_0 \right) \right)
\label{eq:A_3}
\end{align}
where in (a), we swap sums from $\sum_{h_1=1}^{d-1} \sum_{h_2 = h_1+1}^d$ to $\sum_{h_2=1}^{d} \sum_{h_1 = 1}^{h_2-1}$, in (b), we use Equation~\eqref{eq:sum_B_to_B}. Finally, to obtain the last equality (c), we re-index $h_2$ to $h_1$ and use 
\begin{align}
    B(h-1) \weq 
    \begin{cases}
        2 B(h) & h \in \{0, 1, \cdots, d-1 \},\\
        B(d) & h = d,
    \end{cases}
\label{eq:B_2B}
\end{align}
(Recall also that $B(d) = 1$.) By combining Equations~\eqref{eq:A_1} and \eqref{eq:A_3}, we~obtain
\begin{align*}
   A_1+A_2+A_3 
    & \weq \sum_{h_1=1}^{d-1} 2 B^2(h_1) \zeta^2 \left( p_{h_1} - p_0 \right) - \sum_{h_1=1}^{d-1} 2 B(h_1) \zeta \left(1-\eta \right) \left( p_{h_1} - p_0 \right) \\
    & \qquad +  \sum_{h_1=1}^{d} B(0)B(h_1)\zeta^2\left( p_{h_1} - p_0 \right) - \left( \sum_{h_1=1}^{d-1} 2B^2(h_1)\zeta^2 \left( p_{h_1} - p_0 \right) + \zeta^2 \left( p_d - p_0 \right) \right), \\
    &= - \sum_{h_1=1}^{d-1} 2 B(h_1) \zeta \left(1-\eta \right) \left( p_{h_1} - p_0 \right) + \sum_{h_1=1}^{d} B(0)B(h_1)\zeta^2\left( p_{h_1} - p_0 \right) - \zeta^2 \left( p_d - p_0 \right), \\
    &= - \sum_{h_1=1}^{d} 2 B(h_1) \zeta \left(1-\eta \right) \left( p_{h_1} - p_0 \right) + 2 \zeta \left(1-\eta \right) (p_d - p_0) \\ 
    & \qquad + \sum_{h_1=1}^{d} B(0)B(h_1)\zeta^2\left( p_{h_1} - p_0 \right) - \zeta^2 \left( p_d - p_0 \right), \\
    &=^{(a)} - 2 B(0)\zeta \left(1-\eta \right)\left(\bp_1 - p_0\right) + 2 \zeta \left(1-\eta \right) (p_d - p_0) + B^2(0)\zeta^2\left( \bp_1 - p_0 \right) - \zeta^2 \left( p_d - p_0 \right), \\
    &= - \eta \left(1-\eta \right)\left(\bp_1 - p_0\right) + 2 \zeta \left(1-\eta \right) (p_d - p_0) + \eta^2\left( \bp_1 - p_0 \right) - \zeta^2 \left( p_d - p_0 \right),
\end{align*}
where in (a) we use Equations~\eqref{eq:p_bar_B} and~\eqref{eq:B_2B}, and in the last equality we use $B(0)\zeta = \eta$.

By substituting $A_1+A_2+A_3$, we can rewrite Equation~\eqref{eq:proposition1_case2_before} as
\begin{align*}
&\left( \eta^2 - 2 \eta \left(1-\eta \right) \right) \left(\bp_1 - p_0\right) + 2 \zeta \left(1-\eta \right) (p_d - p_0) - \zeta^2 \left( p_d - p_0 \right) \\
& \qquad + \left( p_{d-1} - p_0 \right) \left(\left(1 - \eta \right)^2 + \zeta^2 - 2 \zeta\left(1 - \eta \right) \right)-\left( p_d - p_{d-1} \right) \zeta \left( 2(1 - \eta) - \zeta \right), \\
& \weq \left( \eta^2 - 2 \eta \left(1-\eta \right) \right) \left(\bp_1 - p_0\right) + \left( p_{d-1} - p_0 \right) \left(1 - \eta \right)^2.
\end{align*}
Therefore, Condition~\eqref{eq:robustness_threshold} is equivalent to
\begin{align}
\left( \eta^2 - 2 \eta \left(1-\eta \right) \right) \left(\bp_1 - p_0\right) + \left( p_{d-1} - p_0 \right) \left(1 - \eta \right)^2 > 0. 
\label{eq:in_proof_scenario3}
\end{align}

If $ p_{d-1} > \bp_1$ we have
\begin{align*}
 \left( \eta^2 - 2 \eta \left(1-\eta \right) \right) \left(\bp_1 - p_0\right) + \left( p_{d-1} - p_0 \right) \left(1 - \eta \right)^2
 & \weq (\bp_1-p_0) \left((1-\eta)^2 + \eta^2-2\eta(1-\eta) \right), \\
 & \weq (\bp_1-p_0) \left( 1 - 2\eta \right)^2,
\end{align*}
and this quantity is positive since $\bp_1 > p_0$. Therefore, $ p_{d-1} > \bp_1$ is a sufficient condition to recover the tree, and this proves point~(1) of Proposition~\ref{prop:robustness_adversarial}.

If $p_{d-1} \leq \bp_1$ we can rewrite the left-hand side of Equation~\eqref{eq:in_proof_scenario3} (by ``completing the square'') as
\begin{align*}
& (p_{d-1} + 3\bp_1 - 4p_0)\eta^2 -2(p_{d-1}+\bp_1-2p_0)\eta + (p_{d-1}-p_0), \\
& \weq (p_{d-1} + 3\bp_1 - 4p_0) \left(\eta -\frac{p_{d-1}+\bp_1-2p_0}{p_{d-1} + 3\bp_1 - 4p_0} \right)^2  + (p_{d-1}-p_0) - \frac{(p_{d-1}+\bp_1-2p_0)^2}{p_{d-1} + 3\bp_1 - 4p_0}, \\
& \weq (p_{d-1} + 3\bp_1 - 4p_0) \left(\eta -\frac{p_{d-1}+\bp_1-2p_0}{p_{d-1} + 3\bp_1 - 4p_0} \right)^2  - \frac{(\bp_1 - p_{d-1})(\bp_1-p_0)}{p_{d-1} + 3\bp_1 - 4p_0}.
\end{align*}
Therefore, Equation~\eqref{eq:in_proof_scenario3} is equivalent to
\begin{align*}
(p_{d-1} + 3\bp_1 - 4p_0) \left(\eta -\frac{p_{d-1}+\bp_1-2p_0}{p_{d-1} + 3\bp_1 - 4p_0} \right)^2  - \frac{(\bp_1 - p_{d-1})(\bp_1-p_0)}{p_{d-1} + 3\bp_1 - 4p_0} > 0.
\end{align*}
Because $(p_{d-1} + 3\bp_1 - 4p_0) > 0$, this inequality is satisfied if $\eta < \eta_-$ or if $\eta > \eta_+$, where 
\begin{align*}
 \eta_- & \weq \frac{p_{d-1}+\bp_1 -2p_0-\sqrt{(\bp_1 - p_{d-1})(\bp_1 - p_0)}}{p_{d-1}+3\bp_1 -4p_0}, \\
 \eta_+ & \weq \frac{p_{d-1}+\bp_1 -2p_0+\sqrt{(\bp_1 - p_{d-1})(\bp_1 - p_0)}}{p_{d-1}+3\bp_1 -4p_0}.
\end{align*}
 Notice that 
$\eta_+ = \frac{\sqrt{\bp_1 - p_{d-1}}(2\sqrt{\bp_1 - p_0}-\sqrt{\bp_1 - p_{d-1}})}{2(p_{d-1}+3\bp_1 -4p_0)} + \frac{1}{2}> \frac{1}{2}$, and hence the condition $\eta > \eta_+$ cannot be verified (recall that $0 \le \eta < 1/2$).
In contrast, we have $\frac{p_{d-1}+\bp_1 -2p_0}{p_{d-1}+3\bp_1 -4p_0} - \frac{1}{2} = \frac{p_{d-1} - \bp_1}{2(p_{d-1}+3\bp_1 -4p_0)} \leq 0$ and hence $\eta_- \le \frac12$. 
Moreover, using
\begin{align*}
(p_{d-1}+\bp_1 -2p_0)^2 - (\bp_1 - p_{d-1})(\bp_1 - p_0) 
& \weq p_{d-1}^2 +4p_0^2 +3\bp_1 p_{d-1} - 5p_{d-1} p_0 -3p_0 \bp_1 \\
& \weq 3(p_{d-1}-p_0)(\bp_1 - p_0) +(p_{d-1}-p_0)^2,
\end{align*}
$p_{d-1}+\bp_1 -2p_0 > 0 $, and $(\bp_1 - p_{d-1})(\bp_1 - p_0) > 0$, it follows that $\eta_- > 0$. 
Therefore, when $ p_{d-1} \le \bp_1$, condition~\eqref{eq:robustness_threshold} is satisfied if $0 \leq \eta < \eta_-$. 
\end{proof}

\subsubsection{Proof of Lemma~\ref{lemma:bottomup_robustness_general}}
\label{subsec:proof_generalLemma}

\begin{proof}[Proof of Lemma~\ref{lemma:bottomup_robustness_general}]
We proceed by showing that the edge density $\edgedensity \left( \tC_a, \tC_b \right)$ is concentrated around its mean. We denote by $O_{ab} = C_a \cap \tC_b$ the nodes in cluster $C_a$ but assigned to cluster~$\tC_b$. In particular, the set of all misclassified nodes $O$ is given by $O = \cup_{ \substack{ a, b \in \cL_{\cT} \\ b \ne a } } O_{ab}$ and is independent of the edges. Therefore, $w\left( O_{ab}, O_{cd} \right) \sim \Bin\left( |O_{ab}| \cdot |O_{cd}|, p_{ac} \right)$. 
We note, from the definition of the edge density~\eqref{eq:def_edgedensity} and the fact that $\tC_a = \cup_{k \in \cL_{\cT} } O_{ka}$, that 
\begin{align}
\label{eq:in_proof_edgesensity_errorFramework}
 \edgedensity \left( \tC_a, \tC_b \right) 
  \weq \frac{ w \left( \tC_a, \tC_b \right) }{ \left| \tC_a \right| \cdot \left| \tC_b \right| } \quad \text{ with } \quad 
 w \left( \tC_a, \tC_b \right) \weq \sum_{k,\ell \in \cL_{\cT} } w \left( O_{ka}, O_{\ell b} \right),
\end{align}
We will express $\E w \left( O_{ka}, O_{\ell b} \right)$ for each scenario. The concentration $\edgedensity \left( \tC_a, \tC_b \right)$ around its mean then follows from Chernoff's bound.

In the remainder of the proof, we let $a, b, c \in \cL_{\cT}$ be three different leaves such that  $b$ is the closest leaf from~$a$, \textit{i.e.,} 
$b = \argmax\limits_{ k \in \cL_{\cT} } | \lca ( a, k ) |$ and $c \not \in \{a,b\}$. 
To conclude that the linkage procedure outputs the correct tree, we have to verify that 
\begin{align*}
\forall c \not \in \{ a,b \} \colon 
 \edgedensity \left( \tC_a, \tC_b \right) > \edgedensity \left( \tC_a, \tC_c \right).
\end{align*}
Since $ \edgedensity \left( \tC_a, \tC_c \right) = (1+o(1)) \edgedensity \left( \tC_b, \tC_c \right)$, Equation~\eqref{eq:extension_edgedensity_set_of_clusters} ensures $ \edgedensity \left( \tC_{a \cup b}, \tC_c \right) = (1+o(1)) \edgedensity \left( \tC_a, \tC_c \right)$ after the merge, which in turn guarantees that the new super community $a \cup b$ does not merge with any another community before all the bottom communities are merged. After all the bottom communities are merged, we again obtain balanced (super) communities and repeat these steps for these super communities as we did for the bottom communities.

The random variables $w \left( O_{ka}, O_{\ell b} \right)$ follow a Binomial distribution such that $\E w \left( O_{ka}, O_{\ell b} \right) \weq \left( \frac{N}{K} \right)^2 \zeta(|\lca(k, a)|) \zeta(|\lca(\ell, b)|) p_{|\lca(k,\ell)|}$, where $\zeta(h)$ is the fraction of nodes mislabeled from a bottom community with similarity $h$ to bottom community $a$.
We define $h_1 = |\lca(k,a)|$, $h_2 = |\lca(\ell, b)|$, and $h_{ac} = |\lca(a,c)|$. 

When $h_1 = h_{ac} < h_2 \leq d$, each node of the fraction $\zeta(h_1) = \zeta(h_{ac})$ of the bottom cluster $c$ is mislabeled as a node inside a uniformly randomly chosen bottom cluster whose similarity with the bottom cluster $a$ is no less than $h_1+1 = h_{ac}+1$. Furthermore, each node of the fraction $\zeta(h_2)$ of the bottom cluster $a$ is mislabeled as a node inside a uniformly randomly chosen bottom cluster whose similarity with the bottom cluster $a$ is $h_2~(< h_1 = h_{ac})$. Therefore, in this case, the similarity between a node from $a$ and a node from $c$ is exactly the same as when the two nodes are randomly picked from super community $C_u:~C_u \weq \bigcup_{ b \in \cL_{\cT} \colon b_{1:(h_{ac}+1)} = a_{1:(h_{ac}+1)}} C_{b}$. Hence, the expected edge probability between these two nodes is $\bp_{h_{1}+1}=\bp_{h_{ac}+1}$. 
Similarly, when $h_1 = h_2 < h_{ac}$, each node of the fraction $\zeta(h_1) = \zeta(h_{h_2})$ of both bottom clusters $a$ and $c$ is mislabeled as a node inside a uniformly randomly chosen bottom cluster whose similarity with the bottom cluster $a$ (and also with $c$) is $h_1=h_2$. Therefore, in this case, the expected edge probability between a node from $a$ and a node from $c$ is $\bp_{h_{1}+1}=\bp_{h_{2}+1}$. In the other cases, the similarity between bottom clusters $a$ and $c$ is $\min(p_{h_1}, p_{h_2}, p_{h_{ac}}) = \min(p_{h_1}, p_{h_{ac}})$.

Now we obtain whp that for any $c \neq a$,
\begin{align*}
\E w \left( O_{ka}, O_{\ell c} \right) \weq \left( \frac{N}{K} \right)^2 \times 
\begin{cases}
\zeta(h_1) \zeta(h_2) p_{h_1} & \text{ if } 0 \leq h_1 < h_{ac} \text{ and } h_1 < h_2 \leq d, \\
\zeta(h_{ac})\zeta(h_2) \bp_{h_{ac}+1}  & \text{ if } h_1 = h_{ac} \text{ and } h_{ac} < h_2 \leq d , \\
\zeta(h_1) \zeta(h_2) p_{h_{ac}} & \text{ if } h_{ac} < h_1 < h_2 \leq d, \\
\zeta^2(h_1) \bp_{h_1+1} & \text{ if } 0 \leq h_1 = h_2 < h_{ac} , \\
\zeta^2(h_1) p_{h_{ac}} & \text{ if } h_{ac} \leq h_1 = h_2 \leq d, \\  
 \end{cases}
\end{align*}

The special case $h_{ac} = d-1$ occurs when $c=b$. In that case, we have 
\begin{align*}
 \E w \left( O_{ka}, O_{\ell b} \right) \weq \left( \frac{N}{K} \right)^2 \times 
 \begin{cases}
    \zeta(h_1) \zeta(h_2) p_{h_1} & \text{ if } 0 \leq h_1 < d-1 \text{ and } h_1 < h_2 \leq d , \\
    \zeta(d-1) \zeta(d) p_{d}  & \text{ if } (h_1, h_2)=(d-1, d), \\
    \zeta^2(h_1) \bp_{h_1+1} & \text{ if } h_1 = h_2 < d-1, \\ 
    \zeta^2(h_1) p_{d-1} & \text{ if } d-1 \leq h_1 = h_2 \leq d, \\ 
 \end{cases}
\end{align*}

Because $c \notin \{a, b \}$, we have $h_{ac}<d-1$. Thus, $\E w \left( O_{ka}, O_{\ell b} \right) - \E w \left( O_{ka}, O_{\ell c} \right)$ is equal to  
\begin{align*}
\left( \frac{N}{K} \right)^2 \times 
\begin{cases}
0 & \text{ if } 0 \leq h_1 < h_{ac} \text{ and } h_1 < h_2 \leq d, \\
\zeta(h_{ac})\zeta(h_2) \left(p_{h_{ac}} - \bp_{h_{ac}+1} \right)  & \text{ if } h_1 = h_{ac} < h_2 \leq d , \\
\zeta(h_1) \zeta(h_2) \left( p_{h_1} - p_{h_{ac}} \right) & \text{ if } h_{ac} < h_1 < h_2 \leq d \text{ and } h_1 < d-1, \\
\zeta(d-1) \zeta(d) \left( p_{d} - p_{h_{ac}} \right) & \text{ if } (h_1, h_2) = (d-1,d), \\
0 & \text{ if } 0 \leq h_1 = h_2 < h_{ac}  , \\
\zeta^2(h_1) \left( \bp_{h_1+1} - p_{h_{ac}} \right) & \text{ if } h_{ac} \leq h_1 = h_2 < d-1, \\    
\zeta^2(h_1) \left( p_{d-1} - p_{h_{ac}} \right) & \text{ if } d-1 \leq h_1 = h_2 \leq d. \\  
 \end{cases}
\end{align*}
Because $B(h)$ is the number of bottom communities having a similarity $h$ with the bottom community $a$, we obtain 
\begin{align}
\label{eqa:robustness_edge_dense_raw}
 & \edgedensity \left( \tC_a, \tC_b \right) - \edgedensity \left( \tC_a, \tC_c \right) \notag\\
 \weq & (1+o(1)) \Bigg(
 - \sum_{h_2=h_{ac}+1}^d 2 B(h_{ac}) B(h_2) \zeta(h_{ac})\zeta(h_2) \left(\bp_{h_{ac}+1} - p_{h_{ac}} \right) \notag \\
 & \qquad \qquad \quad  + \sum_{h_1 = h_{ac}+1}^{d-2} \sum_{h_2 = h_1+1}^d 2 B(h_1) B(h_2) \zeta(h_1) \zeta(h_2) \left( p_{h_1} - p_{h_{ac}} \right)
+ 2 \zeta(d-1) \zeta(d) \left( p_{d} - p_{h_{ac}} \right) \notag \\
 & \qquad \qquad \quad +  \sum_{h_1 = h_{ac}}^{d-1} B^2(h_1) \zeta^2(h_1) \left( \bp_{h_1+1} - p_{h_{ac}} \right)
+ \zeta^2(d) \left( p_{d-1} - p_{h_{ac}} \right)
 \Bigg).
\end{align}
We show in Section~\ref{subsubsec:detailed_computation} that we can further express $\edgedensity \left( \tC_a, \tC_b \right) - \edgedensity \left( \tC_a, \tC_c \right) $ as 
\begin{align}
\label{eqa:robustness_edge_dense}
 & (1+o(1)) \Bigg( \sum_{h_1=h_{ac}+1}^{d-1} \sum_{h_2 = h_1+1}^{d} 2 B(h_1) B(h_2) \left(\zeta(h_1) - \zeta(h_{ac}) \right) \left(\zeta(h_2) - \zeta(h_{ac}) \right) \left( p_{h_1} - p_{h_{ac}} \right)  \notag\\
 & \qquad \qquad \qquad \qquad+  \sum_{h_1=h_{ac}+1}^{d-1} B^2(h_1) \left( \zeta(h_1) - \zeta(h_{ac}) \right)^2 \left( \bp_{h_1+1} - p_{h_{ac}} \right)  \notag\\
 & \qquad \qquad \qquad \qquad + 2 \left(\zeta(d-1) - \zeta(h_{ac}) \right) \left( \zeta(d) - \frac{\zeta(d-1) + \zeta(h_{ac})}{2} \right) \left( p_{d} - p_{d-1} \right) \notag \\
 & \qquad \qquad \qquad \qquad + \left(\zeta(d) - \zeta(h_{ac})\right)^2 \left( p_{d-1} - p_{h_{ac}} \right) \Bigg) .
\end{align}
Therefore, the condition $\edgedensity \left( \tC_a, \tC_b \right) - \edgedensity \left( \tC_a, \tC_c \right) > 0$ is equivalent to condition~\eqref{eq:robustness_threshold}.
\end{proof}

\subsubsection{Additional Computation for the Proof of Lemma~\ref{lemma:bottomup_robustness_general}}
\label{subsubsec:detailed_computation}
In this subsection, we detail the tedious computations that allow us to transform Equation~\eqref{eqa:robustness_edge_dense_raw} into Equation~\eqref{eqa:robustness_edge_dense}. for the proof of Lemma~\ref{lemma:bottomup_robustness_general}. In order to highlight the differences between the lines, we sometimes use bold characters.
Let $\Delta \edgedensity = \edgedensity \left( \tC_a, \tC_b \right) - \edgedensity \left( \tC_a, \tC_c \right) $. From Equation~\eqref{eqa:robustness_edge_dense_raw}, we have 
\begin{align*}
    \Delta \edgedensity
    \weq & (1+o(1)) \Bigg(
    - \sum_{h_2=h_{ac}+1}^d 2 B(h_{ac}) B(h_2) \zeta(h_{ac})\zeta(h_2) \left(\bp_{h_{ac}+1} - p_{h_{ac}} \right) \notag\\
    & \qquad \qquad \quad + \ \sum_{h_1 = h_{ac}+1}^{d-2} \sum_{h_2 = h_1+1}^d 2 B(h_1) B(h_2) \zeta(h_1) \zeta(h_2) \left( p_{h_1} - p_{h_{ac}} \right) \notag\\
    & \qquad \qquad \quad + 2 B(d-1)B(d)\zeta(d-1) \zeta(d) \left( p_{d} - p_{h_{ac}} \right) + \sum_{h_1 = h_{ac}}^{d-2} B^2(h_1) \zeta^2(h_1) \left( \bp_{h_1+1} - p_{h_{ac}} \right) \notag\\
    & \qquad \qquad \quad + \sum_{h_1 = d-1}^d B^2(h_1) \zeta^2(h_1) \left( p_{d-1} - p_{h_{ac}} \right)
    \Bigg) \\ 
    \weq & (1+o(1)) \Bigg(
    - \sum_{h_1=h_{ac}+1}^d \boldsymbol{\sum_{h_2=h_{ac}+1}^d} 2 B(h_1)B(h_2)\zeta(h_{ac})\zeta(h_2) \left( \boldsymbol{p_{h_1}} - p_{h_{ac}} \right) \notag\\
    & \qquad \qquad \quad + \sum_{h_1 = h_{ac}+1}^{d-2} \sum_{h_2 = h_1+1}^d 2 B(h_1) B(h_2) \zeta(h_1) \zeta(h_2) \left( p_{h_1} - p_{h_{ac}} \right) \notag\\
    & \qquad \qquad \quad +  2 \zeta(d-1) \zeta(d) \left( p_{d} - p_{h_{ac}} \right) + \sum_{h_1 = h_{ac}}^{d-2} \sum_{h=h_1+1}^d B(h_1)\boldsymbol{B(h)} \zeta^2(h_1) \left( \boldsymbol{p_{h}} - p_{h_{ac}} \right) \notag\\
    & \qquad \qquad \quad + \left(\zeta^2(d-1) + \zeta^2(d) \right)\left( p_{d-1} - p_{h_{ac}} \right)
    \Bigg). 
\end{align*}
Thus, 
\begin{align}
    \Delta \edgedensity
    \weq & (1+o(1)) \Bigg(
    - \sum_{h_1=h_{ac}+1}^d \sum_{h_2=h_{ac}+1}^d 2 B(h_1)B(h_2)\zeta(h_{ac})\zeta(h_2) \left( p_{h_1} - p_{h_{ac}} \right) \notag\\
    & \qquad \qquad \quad  + \sum_{h_1 = h_{ac}+1}^{\boldsymbol{d-1}} \sum_{h_2 = h_1+1}^d 2 B(h_1) B(h_2) \zeta(h_1) \zeta(h_2) \left( p_{h_1} - p_{h_{ac}} \right) \notag\\
    & \qquad \qquad \quad \boldsymbol{- \ 2 \zeta(d-1) \zeta(d) \left( p_{d-1} - p_{h_{ac}} \right)} + 2 \zeta(d-1) \zeta(d) \left( p_{d} - p_{h_{ac}} \right) \notag\\
    & \qquad \qquad \quad + \sum_{h_1 = h_{ac}}^{\boldsymbol{d-1}} \sum_{h=h_1+1}^d B(h_1)B(h) \zeta^2(h_1) \left( p_{h} - p_{h_{ac}} \right) \boldsymbol{- \zeta^2(d-1) \left( p_{d} - p_{h_{ac}}\right)} \notag\\ 
    & \qquad \qquad \quad + \left(\zeta^2(d-1) + \zeta^2(d) \right)\left( p_{d-1} - p_{h_{ac}} \right)
    \Bigg), \notag \\
   \weq & (1+o(1)) \Bigg(
    - \underbrace{\sum_{h_1=h_{ac}+1}^d \sum_{h_2=h_{ac}+1}^d 2 B(h_1)B(h_2)\zeta(h_{ac})\zeta(h_2) \left( p_{h_1} - p_{h_{ac}} \right)}_{T_1} \notag\\
    & \qquad \qquad \quad + \underbrace{\sum_{h_1 = h_{ac}+1}^{d-1} \sum_{h_2 = h_1+1}^d 2 B(h_1) B(h_2) \zeta(h_1) \zeta(h_2) \left( p_{h_1} - p_{h_{ac}} \right)}_{T_2} \notag\\
    & \qquad \qquad \quad + \underbrace{\sum_{h_1 = h_{ac}}^{d-1} \sum_{h=h_1+1}^d B(h_1)B(h) \zeta^2(h_1) \left( p_{h} - p_{h_{ac}} \right)}_{T_3} + \boldsymbol{\zeta^2(d)\left( p_{d-1} - p_{h_{ac}} \right)} \notag\\ 
    & \qquad \qquad \quad + \boldsymbol{\left(2 \zeta(d-1) \zeta(d) - \zeta^2(d-1) \right) \left( p_{d} - p_{d-1} \right)}
    \Bigg). \label{eq:in_proof_HorribleExpressions}
    & 
    & 
    & 
\end{align}
We compute $T_1$ and $T_3$ separately. We have
\begin{align*}
T_1
& \weq \sum_{h_1=h_{ac}+1}^{d-1} \sum_{h_2=h_1+1}^d 2 B(h_1)B(h_2)\zeta(h_{ac})\zeta(h_2) \left( p_{h_1} - p_{h_{ac}} \right) + \sum_{h_1=h_{ac}+1}^d 2 B^2(h_1)\zeta(h_{ac})\zeta(h_1) \left( p_{h_1} - p_{h_{ac}} \right) \\
& \quad +\sum_{h_1=h_{ac}+1}^d \sum_{h_2=h_{ac}+1}^{h_1 -1} 2 B(h_1)B(h_2)\zeta(h_{ac})\zeta(h_2) \left( p_{h_1} - p_{h_{ac}} \right), \\
& \weq \sum_{h_1=h_{ac}+1}^{d-1} \sum_{h_2=h_1+1}^d 2 B(h_1)B(h_2)\zeta(h_{ac})\zeta(h_2) \left( p_{h_1} - p_{h_{ac}} \right) + \sum_{h_1=h_{ac}+1}^d 2 B^2(h_1)\zeta(h_{ac})\zeta(h_1) \left( p_{h_1} - p_{h_{ac}} \right) \\
& \quad + \boldsymbol{\sum_{h_1=h_{ac}+1}^{d-1} \sum_{h_2=h_1+1}^{d}} 2 B(h_1)\zeta(h_{ac})\zeta(h_1) \left( \boldsymbol{B(h_2)}p_{h_2} - \boldsymbol{B(h_2)}p_{h_{ac}} \right), \\
& \weq \sum_{h_1=h_{ac}+1}^{d-1} \sum_{h_2=h_1+1}^d 2 B(h_1)B(h_2)\zeta(h_{ac})\zeta(h_2) \left( p_{h_1} - p_{h_{ac}} \right) + \sum_{h_1=h_{ac}+1}^d 2 B^2(h_1)\zeta(h_{ac})\zeta(h_1) \left( p_{h_1} - p_{h_{ac}} \right) \\
& \quad + \sum_{h_1=h_{ac}+1}^{d-1} 2 \boldsymbol{B^2(h_1)}\zeta(h_{ac})\zeta(h_1) \left( \boldsymbol{\bp_{h_1+1}} - p_{h_{ac}} \right),
\end{align*}
and
\begin{align*}
T_3
& \weq \sum_{h_1=h_{ac}+1}^d \sum_{h_2 = h_{ac}}^{h_1-1} B(h_2)B(h_1) \zeta^2(h_2) \left( p_{h_1} - p_{h_{ac}} \right), \\
& \weq \sum_{h_1=h_{ac}+1}^d \sum_{h_2 = \boldsymbol{h_{ac}+1}}^{h_1-1} B(h_1) B(h_2) \zeta^2(h_2) \left( p_{h_1} - p_{h_{ac}} \right)
\boldsymbol{+ \sum_{h_1=h_{ac}+1}^d B(h_{ac})B(h_1) \zeta^2(h_{ac}) \left( p_{h_1} - p_{h_{ac}} \right)}, \\
& \weq \sum_{h_1=h_{ac}+1}^d \sum_{h_2 = h_{ac}+1}^{h_1-1} B(h_1) B(h_2) \zeta^2(h_2) \left( p_{h_1} - p_{h_{ac}} \right) +  \sum_{h_1=h_{ac}+1}^d \boldsymbol{\sum_{h_2 = h_{ac}+1}^{d}} B(h_1) \boldsymbol{B(h_2)} \zeta^2(h_{ac}) \left( p_{h_1} - p_{h_{ac}} \right), \\
& \weq \sum_{h_1=h_{ac}+1}^d \sum_{h_2 = h_{ac}+1}^{h_1-1} B(h_1) B(h_2) \zeta^2(h_2) \left( p_{h_1} - p_{h_{ac}} \right) +  \sum_{h_1=h_{ac}+1}^d \sum_{h_2 = h_{ac}+1}^{\boldsymbol{h_1 - 1}} B(h_1) B(h_2) \zeta^2(h_{ac}) \left( p_{h_1} - p_{h_{ac}} \right) \\
& \quad+  \sum_{h_1=h_{ac}+1}^d \boldsymbol{B^2(h_1)} \zeta^2(h_{ac}) \left( p_{h_1} - p_{h_{ac}} \right)
+  \sum_{h_1=h_{ac}+1}^{d-1} \sum_{\boldsymbol{h_2 = h_1+1}}^{d} B(h_1) B(h_2) \zeta^2(h_{ac}) \left( p_{h_1} - p_{h_{ac}} \right). 
\end{align*}
We can further manipulate $T_3$ to obtain 
\begin{align*}
T_3 & \weq \sum_{h_1=h_{ac}+1}^d \sum_{h_2 = h_{ac}+1}^{h_1-1} B(h_1) B(h_2) \left(\zeta^2(h_2) + \zeta^2(h_{ac})\right) \left( p_{h_1} - p_{h_{ac}} \right) \\
& \quad+  \sum_{h_1=h_{ac}+1}^d B^2(h_1) \zeta^2(h_{ac}) \left( p_{h_1} - p_{h_{ac}} \right)
+  \sum_{h_1=h_{ac}+1}^{d-1} \sum_{h_2 = h_1+1}^{d} B(h_1) B(h_2) \zeta^2(h_{ac}) \left( p_{h_1} - p_{h_{ac}} \right), \\
& \weq \boldsymbol{\sum_{h_1=h_{ac}+1}^{d-1} \sum_{h_2 = h_1+1}^{d}} B(h_1) \left(\zeta^2(h_1) + \zeta^2(h_{ac})\right) \left( \boldsymbol{B(h_2)}p_{h_2} - \boldsymbol{B(h_2)}p_{h_{ac}} \right) \\
& \quad + \sum_{h_1=h_{ac}+1}^d B^2(h_1) \zeta^2(h_{ac}) \left( p_{h_1} - p_{h_{ac}} \right)
+  \sum_{h_1=h_{ac}+1}^{d-1} \sum_{h_2 = h_1+1}^{d} B(h_1) B(h_2) \zeta^2(h_{ac}) \left( p_{h_1} - p_{h_{ac}} \right), \\
& \weq \sum_{h_1=h_{ac}+1}^{d-1}  \boldsymbol{B^2(h_1)} \left(\zeta^2(h_1) + \zeta^2(h_{ac})\right) \left( \boldsymbol{\bp_{h_1+1}} - p_{h_{ac}} \right) +  \sum_{h_1=h_{ac}+1}^d B^2(h_1) \zeta^2(h_{ac}) \left( p_{h_1} - p_{h_{ac}} \right)\\
& \quad + \sum_{h_1=h_{ac}+1}^{d-1} \sum_{h_2 = h_1+1}^{d} B(h_1) B(h_2) \zeta^2(h_{ac}) \left( p_{h_1} - p_{h_{ac}} \right).
\end{align*}
Combining the expressions for $T_1$ and $T_3$, we obtain   
\begin{align*}
& - T_1 + T_3 \\
& \weq - \sum_{h_1=h_{ac}+1}^{d-1} \sum_{h_2=h_1+1}^d 2 B(h_1)B(h_2)\zeta(h_{ac})\zeta(h_2) \left( p_{h_1} - p_{h_{ac}} \right) - \sum_{h_1=h_{ac}+1}^d 2 B^2(h_1)\zeta(h_{ac})\zeta(h_1) \left( p_{h_1} - p_{h_{ac}} \right) \\
& \quad - \sum_{h_1=h_{ac}+1}^{d-1} 2 B^2(h_1)\zeta(h_{ac})\zeta(h_1) \left( \bp_{h_1+1} - p_{h_{ac}} \right) + \sum_{h_1=h_{ac}+1}^{d-1}  B^2(h_1) \left(\zeta^2(h_1) + \zeta^2(h_{ac})\right) \left( \bp_{h_1+1} - p_{h_{ac}} \right)  \\
&\quad+  \sum_{h_1=h_{ac}+1}^d B^2(h_1) \zeta^2(h_{ac}) \left( p_{h_1} - p_{h_{ac}} \right)
+  \sum_{h_1=h_{ac}+1}^{d-1} \sum_{h_2 = h_1+1}^{d} B(h_1) B(h_2) \zeta^2(h_{ac}) \left( p_{h_1} - p_{h_{ac}} \right) \\
& \weq \sum_{h_1=h_{ac}+1}^{d-1} \sum_{h_2 = h_1+1}^{d} B(h_1) B(h_2) \left( \zeta^2(h_{ac}) -  2 \zeta(h_{ac})\zeta(h_2) \right) \left( p_{h_1} - p_{h_{ac}} \right)  \\
&\quad + \sum_{h_1=h_{ac}+1}^d B^2(h_1) \left(\zeta^2(h_{ac}) - 2\zeta(h_{ac})\zeta(h_1) \right) \left( p_{h_1} - p_{h_{ac}} \right) \\ 
& \quad + \sum_{h_1=h_{ac}+1}^{d-1} B^2(h_1) \left( \zeta(h_1) - \zeta(h_{ac}) \right)^2 \left( \bp_{h_1+1} - p_{h_{ac}} \right). 
\end{align*}
Furthermore, 
\begin{align*}
& - T_1 + T_3 \\
& \weq \sum_{h_1=h_{ac}+1}^{d-1} \sum_{h_2 = h_1+1}^{d} B(h_1) B(h_2) \left( \zeta^2(h_{ac}) -  2 \zeta(h_{ac})\zeta(h_2) \right) \left( p_{h_1} - p_{h_{ac}} \right)  \\
& \quad+ \sum_{h_1=h_{ac}+1}^{\boldsymbol{d-1}} \boldsymbol{\sum_{h_2 = h_1+1}^d} B(h_1) \boldsymbol{B(h_2)} \left(\zeta^2(h_{ac}) - 2\zeta(h_{ac})\zeta(h_1) \right) \left( p_{h_1} - p_{h_{ac}} \right)  \\
& \quad + \boldsymbol{B^2(d) \left(\zeta^2(h_{ac}) - 2\zeta(h_{ac})\zeta(d) \right) \left( p_{d} - p_{h_{ac}} \right) } + \sum_{h_1=h_{ac}+1}^{d-1} B^2(h_1) \left( \zeta(h_1) - \zeta(h_{ac}) \right)^2 \left( \bp_{h_1+1} - p_{h_{ac}} \right), \\
& \weq \sum_{h_1=h_{ac}+1}^{d-1} \sum_{h_2 = h_1+1}^{d} B(h_1) B(h_2) \left( 2 \zeta^2(h_{ac}) -  2 \zeta(h_{ac})\zeta(h_2) - 2\zeta(h_{ac})\zeta(h_1)\right) \left( p_{h_1} - p_{h_{ac}} \right) \\
& \quad + \left(\zeta^2(h_{ac}) - 2\zeta(h_{ac})\zeta(d) \right) \left( p_{d} - p_{h_{ac}} \right) + \sum_{h_1=h_{ac}+1}^{d-1} B^2(h_1) \left( \zeta(h_1) - \zeta(h_{ac}) \right)^2 \left( \bp_{h_1+1} - p_{h_{ac}} \right),
\end{align*}
and 
\begin{align*}
& -T_1+T_3 - \left( p_{d} - p_{h_{ac}} \right) \left(\zeta^2(h_{ac}) - 2\zeta(h_{ac})\zeta(d) \right) + T_2 \\
& \weq \sum_{h_1=h_{ac}+1}^{d-1} \sum_{h_2 = h_1+1}^{d} B(h_1) B(h_2) \left( 2 \zeta^2(h_{ac}) -  2 \zeta(h_{ac})\zeta(h_2) - 2\zeta(h_{ac})\zeta(h_1)\right) \left( p_{h_1} - p_{h_{ac}} \right) \\
& \quad+ \sum_{h_1=h_{ac}+1}^{d-1} B^2(h_1) \left( \zeta(h_1) - \zeta(h_{ac}) \right)^2 \left( \bp_{h_1+1} - p_{h_{ac}} \right) \\
& \quad + \sum_{h_1 = h_{ac}+1}^{d-1} \sum_{h_2 = h_1+1}^d 2 B(h_1) B(h_2) \zeta(h_1) \zeta(h_2) \left( p_{h_1} - p_{h_{ac}} \right), \\
& \weq \sum_{h_1=h_{ac}+1}^{d-1} \sum_{h_2 = h_1+1}^{d} 2 B(h_1) B(h_2)\left( \zeta(h_1) \zeta(h_2) + \zeta^2(h_{ac}) - \zeta(h_{ac})\zeta(h_2) -\zeta(h_{ac})\zeta(h_1)\right) \left( p_{h_1} - p_{h_{ac}} \right)  \\
& \quad + \sum_{h_1=h_{ac}+1}^{d-1} B^2(h_1) \left( \zeta(h_1) - \zeta(h_{ac}) \right)^2 \left( \bp_{h_1+1} - p_{h_{ac}} \right), \\
& \weq \sum_{h_1=h_{ac}+1}^{d-1} \sum_{h_2 = h_1+1}^{d} 2 B(h_1) B(h_2) \boldsymbol{\left( \zeta(h_1) - \zeta(h_{ac}) \right) \left(\zeta(h_2) -\zeta(h_{ac}) \right)} \left( p_{h_1} - p_{h_{ac}} \right) \\
& \quad + \sum_{h_1=h_{ac}+1}^{d-1} B^2(h_1) \left( \zeta(h_1) - \zeta(h_{ac}) \right)^2 \left( \bp_{h_1+1} - p_{h_{ac}} \right).
\end{align*}
Finally we come back to Equation~\eqref{eq:in_proof_HorribleExpressions} and obtain
\begin{align*}
    \Delta \edgedensity
    & \weq (1+o(1)) \Big(- T_1 + T_2 + T_3 + \zeta^2(d)\left( p_{d-1} - p_{h_{ac}} \right) \\
    & \qquad \qquad \qquad \quad +\left(2 \zeta(d-1) \zeta(d) - \zeta^2(d-1) \right) \left( p_{d} - p_{d-1} \right) 
    \Big), \\ 
    & \weq (1+o(1)) \Big(\sum_{h_1=h_{ac}+1}^{d-1} \sum_{h_2 = h_1+1}^{d} 2 B(h_1) B(h_2) \left( \zeta(h_1) - \zeta(h_{ac}) \right) \left(\zeta(h_2) -\zeta(h_{ac}) \right) \left( p_{h_1} - p_{h_{ac}} \right) \\
    & \qquad \qquad \qquad \quad + \sum_{h_1=h_{ac}+1}^{d-1} B^2(h_1) \left( \zeta(h_1) - \zeta(h_{ac}) \right)^2 \left( \bp_{h_1+1} - p_{h_{ac}} \right) 
    + \zeta^2(d)\left( p_{d-1} - p_{h_{ac}} \right) \\
    & \qquad \qquad \qquad \quad + \left(\zeta^2(h_{ac}) - 2\zeta(h_{ac})\zeta(d) \right)\left( p_{d} - p_{h_{ac}} \right) + \left(2 \zeta(d-1) \zeta(d) - \zeta^2(d-1) \right) \left( p_{d} - p_{d-1} \right) 
    \Big), \\ 
    & \weq (1+o(1)) \Big(\sum_{h_1=h_{ac}+1}^{d-1} \sum_{h_2 = h_1+1}^{d} 2 B(h_1) B(h_2) \left( \zeta(h_1) - \zeta(h_{ac}) \right) \left(\zeta(h_2) -\zeta(h_{ac}) \right) \left( p_{h_1} - p_{h_{ac}} \right) \\
    & \qquad \qquad \qquad \quad+ \sum_{h_1=h_{ac}+1}^{d-1} B^2(h_1) \left( \zeta(h_1) - \zeta(h_{ac}) \right)^2 \left( \bp_{h_1+1} - p_{h_{ac}} \right) 
    + \zeta^2(d)\left( p_{d-1} - p_{h_{ac}} \right) \\
    & \qquad \qquad \qquad \quad+ \left(\zeta^2(h_{ac}) - 2\zeta(h_{ac})\zeta(d) \right) \left( \left( p_{d} \boldsymbol{- p_{d-1}} \right) + \left( \boldsymbol{p_{d-1}} - p_{h_{ac}} \right) \right) \\
    & \qquad \qquad \qquad \quad + \left(2 \zeta(d-1) \zeta(d) - \zeta^2(d-1) \right) \left( p_{d} - p_{d-1} \right) 
    \Big), \\ 
    &= (1+o(1)) \Bigg( \sum_{h_1=h_{ac}+1}^{d-1} \sum_{h_2 = h_1+1}^{d} 2 B(h_1) B(h_2) \left(\zeta(h_1) - \zeta(h_{ac}) \right) \left(\zeta(h_2) - \zeta(h_{ac}) \right) \left( p_{h_1} - p_{h_{ac}} \right)  \notag\\
    & \qquad \qquad \qquad \quad +  \sum_{h_1=h_{ac}+1}^{d-1} B^2(h_1) \left( \zeta(h_1) - \zeta(h_{ac}) \right)^2 \left( \bp_{h_1+1} - p_{h_{ac}} \right)  \notag\\
    & \qquad \qquad \qquad \quad +  2 \left(\zeta(d-1) - \zeta(h_{ac}) \right) \left( \zeta(d) - \frac{\zeta(d-1) + \zeta(h_{ac})}{2} \right) \left( p_{d} - p_{d-1} \right) \notag \\
    & \qquad \qquad \qquad \quad + \left(\zeta(d) - \zeta(h_{ac})\right)^2  \left( p_{d-1} - p_{h_{ac}} \right) \Bigg), 
\end{align*}
and this last expression establishes Equation~\eqref{eqa:robustness_edge_dense}.

\section{Proofs of Sections~\ref{sec:intermediate_exact_recovery} and~\ref{sec:discussion}}

\subsection{Proof of Theorem~\ref{thm:exact_recovery_intermediate_levels}}
\label{appendix:proof_thm:exact_recovery_intermediate_levels}

 Let us first recall some notations and results for the general SBM. Let $\pi \in (0,1)^K$ be a probability vector and $p \in (0,1)^{K \times K}$ a symmetric matrix. We denote $G \sim \SBM(N,\pi, p)$~if
\begin{itemize}
    \item each node $i \in [N]$ of $G$ is assigned to a unique community $C_k$, with $k \sim \Multi(1, [K], \pi)$;
    \item two nodes $i \in C_k$ and $j \in C_{\ell}$ are connected with probability $p_{k\ell}$.
\end{itemize}
In the following, we denote by $\cA = \{A_1, \cdots, A_t\}$ a collection of $t$ non-empty and non-overlapping subsets of $[K]$ such that $\cup_{r=1}^t A_r = [K]$. An algorithm exactly recovers $\cA$ if it assigns each node $i$ in $G$ to an element of $\{ A_1, \cdots, A_t \}$ that contains its true community (up to a relabelling of the $A_r$'s) with probability $1-o(1)$. 
For two non-overlapping subsets $A, B \subset [K]$, we denote by $ \CH \left( A, B \right) = \CH(A, B, \pi, p)$ the quantity
\begin{align*}
 \CH \left( A, B \right) \weq \min_{ \substack{ a \in A \\ b \in B } } \, \sup_{ t \in (0,1) } (1-t) \sum_{ c = 1}^K \pi_c \dren_t \left( \Ber(p_{ab}) \| \Ber(p_{bc}) \right),
\end{align*}
and by $I(\cA) = I(\cA, \pi, p)$ the quantity
\begin{align*}
 I(\cA) \weq \min_{A_r \ne A_t \in \cA} \CH \left( A_r, A_t \right).
\end{align*}
We define the finest partition of $[K]$ with threshold $\tau$ the partition $\cA^*_\tau$ of $[K]$ such that 
\begin{align*}
    \cA^*_\tau \weq \argmax_{\cA} \left\{ \left| \cA \right| \colon I(\cA) > \tau  \right\}.
\end{align*}
It is the partition of $[K]$ in the largest number of subsets, among all partitions that verify $I(\cA) > \tau$. The following theorem holds.
\begin{theorem}
\label{thm:exact_recovery_general_SBM}
 Let $G \sim \SBM(N,\pi,p)$ and $\cA$ a partition of $[K]$ in $t$ non-empty and non-overlapping elements. Suppose that no two rows of $p$ are equal. Note that if $\tau$ is too large, such a partition might not exist. The following holds.
 \begin{enumerate}[(i)]
  \item if $I(\cA) < (1-\Omega(1)) \, \, N^{-1} \log N$, then no algorithm can exactly recover $\cA$;
  \item the \textit{agnostic-degree-profiling} algorithm of~\cite{abbe2015recovering} exactly recovers the finest partition with threshold $\tau = (1+\Omega(1)) \, \, N^{-1} \log N$. 
 \end{enumerate}
\end{theorem}
\begin{proof}The proof of point~(i) can be found in \cite[Theorem~1]{abbe2015community} while point~(ii) corresponds to \cite[Theorem~4]{abbe2015recovering}.
\end{proof}

Finally, we have the following lemma for the estimation of the link probabilities $p_{ab}$.

\begin{lemma}
\label{lemma:convergence_empiricalDistance_betweenClusters_exactEstimatorCase}
 Let $\hcC$ be an exact estimator of $\cC$ and suppose that $\min_{u \in \cT} p(u) = \omega( N^{-2} )$. Then with high probability $\edgedensity(\hC_a, \hC_b) = (1+o(1)) p_{ab}$. 
\end{lemma}
\begin{proof}
    We assume that the permutation $\tau$ in the definition of the loss function (Equation~\eqref{eq:def_ace}) is the identity. Furthermore, we shorten $\edgedensity(\hC_a, \hC_b)$ by $\hp_{ab}$. Since $\hcC$ is an exact estimator of $\cC$, we have for $N$ large enough that $C_a = \hC_a$ for all $a \in \cL_{\cT}$. Hence,
    \begin{align*}
     \hp_{ab} \weq \frac{ \sum_{ i \in C_a, j \in C_b } A_{ij} }{ |C_a | \cdot |C_b| }.
    \end{align*}
    Since $\sum_{ i \in C_a, j \in C_b } A_{ij} \sim \Bin( |C_a| \cdot |C_b|, p_{ab})$ and $p_{ab} = \omega( N^{-2} )$ as well as $|C_a|, |C_b| = \Theta(N)$, we conclude that $\hp_{ab} = (1+o(1)) p_{ab}$ using the concentration of binomial distribution.
\end{proof}

We can now prove Theorem~\ref{thm:exact_recovery_intermediate_levels}. 
\begin{proof}[Proof of Theorem~\ref{thm:exact_recovery_intermediate_levels}]
Since HSBM is a special instance of the general SBM, in which the communities are indexed by elements of $\cL_{\cT}$ instead of elements of $[K]$, we can directly apply Theorem~\ref{thm:exact_recovery_general_SBM}. 
The set $\cS_q$ of nodes at depth $q$ naturally forms a partition $\cA = (A_t)_{t \in \cS_q}$ of $\cL_{\cT}$ as follows: $\ell \in A_t$ iff $\lca(t,\ell) = t$. Exactly recovering this partition is equivalent to recovering exactly $\mc(q,\cC,\cT)$, and therefore point~(i) of Theorem~\ref{thm:exact_recovery_intermediate_levels} follows from point~(i) of Theorem~\ref{thm:exact_recovery_general_SBM}.

Conversely, point~(ii) of Theorem~\ref{thm:exact_recovery_general_SBM} implies that we can recover the finest partition $q^*$. In the case of an HSBM, this corresponds to the largest $q$ such that $I_q > (1+\Omega(1)) \, N^{-1} \log N$. 
Since $q \mapsto I_q$ is non-decreasing we have $q \ge q^*$. 
Moreover, using Lemma~\ref{lemma:convergence_empiricalDistance_betweenClusters_exactEstimatorCase}, we proceed as in the proof of Theorem~\ref{thm:performance_bottomUp} to show that we recover the tree $\cT[\cS_{\le q^*}]$. In particular, we can recover all the super-communities at any higher depth $q \ge q^*$, which are exactly the depths verifying $I_q > (1+\Omega(1)) \, N^{-1} \log N$. 
\end{proof}

\subsection{Proof of Lemma~\ref{lemma:expression_Iq_BTSBM}}
\label{appendix:proof_expression_Iq_BTSBM}

\begin{proof}[Proof of Lemma~\ref{lemma:expression_Iq_BTSBM}]
 Let $a, b$ be two leaves of $\cT$ such that their least common ancestor is at a depth $s$ strictly less than $q$, \textit{i.e.,} $|\lca(a,b)| = s < q$. For any leaf $c \in \cL_T$ we have $p_{ac} = p_{bc}$ if $c \not \in \cL_{\cT[u]}$. Therefore, the sum in~\eqref{eq:CH_definition} can be limited to $c \in \cL_{\cT[u]}$ so that
 \begin{align*}
  \CH \left( a,b \right) 
   & \weq \frac1K \sup_{ t \in (0,1) } (1-t) \sum_{ c \in \cL_{ \cT[u] } } \dren_t \left( \Ber( p_{a c} ) \| \Ber( p_{ bc } ) \right).
 \end{align*}
 In the following, we let $u = \lca(a,b)$ and $s = |u|$. For any two nodes $v,w \in \cT$, we denote by $\similarity(v,w)$ the depth of the least common ancestor to $v$ and $w$, that is $\similarity(v,w) = |\lca(v,w)|$. 
 Finally, we denote by $\Gamma_{a,b} (v,w)$ the set of leaves of $\cT[u]$ for which the common ancestor with $a$ is $v$ and the common ancestor with $b$ is $w$, \textit{i.e.,} 
 \begin{align*}
  \Gamma_{a,b} (v,w) \weq \left\{ c \in \cL_{\cT[u]} \colon \lca(a,c) = v \text{ and } \lca(b,c) = w \right\}.
 \end{align*}
 We have 
 \begin{align*}
   \left| \Gamma_{a,b}(v,w) \right| \weq 
   \begin{cases}
     1 & \text{ if } \similarity( a, v) = d \text{ and } \similarity(b,w) = s \text{ (this is equivalent to $v=a$ and $w=u$)}, \\
     1 & \text{ if } \similarity( a, v) = s \text{ and } \similarity(b,w) = d \text{ (this is equivalent to $v=u$ and $w=b$)}, \\
     2^{k-1} & \text{ if } \similarity( a, v) = d-k \text{ and } \similarity(b,w) = s \text{ for some } k \in [d-s-1], \\
     2^{k-1} & \text{ if } \similarity( a, v) = s \text{ and } \similarity(b,w) = d-k \text{ for some } k \in [d-s-1], \\
     0 & \text{ otherwise.} 
   \end{cases}
 \end{align*}
  Moreover, for any $c\in \Gamma_{a,b} (v,w)$ we have $p_{ac} = p(\lca(a,c)) = p_{\similarity(a,v)}$ and similarly $p_{bc} = p_{\similarity(b,w)}$. Thus,
  \begin{align}
  \label{eq:in_proof}
    \CH(a,b) \weq \frac1K \sup_{t \in (0,1) } (1-t) \left\{ \dren_t\left( p_d \| p_s \right) + \dren_t\left( p_s \| p_d \right) + \sum_{k=1}^{d-s-1} 2^{k-1} \left[ \dren_t\left( p_{d-k} \| p_s \right) + \dren_t\left( p_{s} \| p_{d-k} \right) \right] \right\},
  \end{align}
  where we shortened $\dren_t(\Ber(p)\|\Ber(q))$ by $\dren_t(p,q)$. 
  
  Finally, let $P, Q$ be two probability distributions. By the concavity of $t \mapsto (1-t) \dren_t( P \| Q )$ and the relation $(1-t)\dren_t( P \| Q ) = t \dren_{1-t} ( Q \| P )$ (see for example~\cite{vanErven2014renyi}), the function inside the $\sup$ of Equation~\ref{eq:in_proof} is concave and symmetric around $t = 1/2$. Therefore the $\sup$ of Equation~\ref{eq:in_proof} is achieved at $t=1/2$. Hence, 
  \begin{align*}
      \CH(a,b) \weq \frac1K \left[ \dren_{1/2}\left( p_s, p_d \right) + \sum_{k=1}^{d-s-1} 2^{k-1} \dren_{1/2} \left( p_s, p_{d-k} \right) \right] \weq H_s.
  \end{align*}
  This last quantity $H_s$ depends on $a,b$ only via their similarity $\similarity(a,b) = s$. Moreover, since the network is assortative we have $p_{s} < p_{s+1}$ and thus $\dren_{1/2} \left( p_s, p_{\ell} \right) > \dren_{1/2} \left( p_{s+1}, p_{\ell} \right)$ for $\ell \ge s+1$. Therefore $s\mapsto H_s$ is decreasing, and consequently
  \begin{align*}
    I_q \weq \min_{ \substack{ a \ne b \in \cL_\cT \\ \similarity(a,b) < q } } \CH(a,b) \weq \min_{s<q} H_s \weq H_{q-1}. 
  \end{align*}
\end{proof}

\subsection{Comparing Top-down versus Bottom-up Conditions}
\label{subsection:proof_bottom_up_vs_top_down}

\begin{lemma}
\label{lemma:bottom_up_vs_top_down}
 For any $q \in \{1, \cdots, d-1 \}$ we have $J_q^{ \bua } > J_q^{ \tda }$. When $q=d$, we have $J_d^{ \bua } = J_d^{ \tda }$.
\end{lemma}

\begin{proof}[Proof of Lemma~\ref{lemma:bottom_up_vs_top_down}]
 We have
 \begin{align*}
  2^d J_q^{ \bua } & \weq a_q + 2^{d-q} a_{q-1} + \sum_{k=1}^{d-q} 2^{k-1} a_{d-k} - 2 \sqrt{ a_{q-1} } \left( \sqrt{a_d} + \sum_{k=1}^{d-q} 2^{k-1} \sqrt{ a_{d-k} } \right), \\
  2^d J_q^{ \tda } & \weq a_q + 2^{d-q} a_{q-1} + \sum_{k=1}^{d-q} 2^{k-1} a_{d-k} - 2 \sqrt{ 2^{d-q} a_{q-1} } \sqrt{ a_d + \sum_{k=1}^{d-q} 2^{k-1} a_{d-k} }.
 \end{align*}
 Hence,
 \begin{align*}
   2^d \left( J_q^{ \bua } - J_q^{ \tda } \right) 
   & \weq 2 \sqrt{a_{q-1}} \left( 2^{ \frac{d-q}{2}} \sqrt{ a_d + \sum_{k=1}^{d-q} 2^{k-1} a_{d-k} } -  \left( \sqrt{a_d} + \sum_{k=1}^{d-q} 2^{k-1} \sqrt{ a_{d-k} } \right) \right) \\
   & \weq  2 \sqrt{a_{q-1}} (D-E) \\
   & \weq  \frac{2 \sqrt{a_{q-1}}}{D+E} (D^2-E^2), 
 \end{align*}
where 
\begin{align*}
 D \weq 2^\frac{d-q}{2} \sqrt{\left( a_d + \sum_{k=1}^{d-q} 2^{k-1} a_{d-k} \right)} 
 \quad \text{ and } \quad 
 E \weq \sqrt{a_d} + \sum_{k=1}^{d-q} 2^{k-1} \sqrt{ a_{d-k} }.
\end{align*}
Since $\frac{2 \sqrt{a_{q-1}}}{D+E} > 0$, we focus on showing $D\geq E$.

First, we notice that when $d=q$, $D=E=\sqrt{a_d}$ and hence $J_d^{\tda} = J_d^{\bua}$.

Next, when $d-q \geq 1$, we have
\begin{align*}
    D^2 - E^2 &\weq 2^{d-q} \left( a_d + \sum_{k=1}^{d-q} 2^{k-1} a_{d-k} \right) - \left(\sqrt{a_d} + \sum_{k=1}^{d-q} 2^{k-1} \sqrt{ a_{d-k} } \right)^2\\
    & \weq (2^{d-q}-1) a_d + \sum_{k=1}^{d-q} 2^{k-1}\left(2^{d-q} -2^{k-1}\right)a_{d-k} -\sum_{k=1}^{d-q} 2^{k} \sqrt{a_d} \sqrt{a_{d-k}} -\sum_{k=1}^{d-q-1} \sum_{l>k}^{d-q} 2^{k+l-1} \sqrt{a_{d-k}} \sqrt{a_{d-l}}.
\end{align*}
Noticing that $\sum_{k=1}^{d-q} 2^{k-1} = 2^{d-q} - 1$ and that $\sum_{l=1, l \neq k}^{d-q} 2^{k+l-2} = 2^{k-1} \left(\sum_{l=1}^{d-q} 2^{l-1} - 2^{k-1} \right) = 2^{k-1} (2^{d-q} - 2^{k-1} -1)$ leads to
\begin{align*}
 D^2 - E^2 
 \weq & \left(\sum_{k=1}^{d-q} 2^{k-1} a_d + \sum_{k=1}^{d-q}2^{k-1} a_{d-k} - 2 \sum_{k=1}^{d-q} 2^{k-1} \sqrt{a_d} \sqrt{a_{d-k}} \right) \\ 
 & \qquad + \left( \sum_{k=1}^{d-q} \sum_{ \substack{ l=1 \\ l \neq k }}^{d-q} 2^{k+l-2} a_{d-k} - 2 \sum_{k=1}^{d-q-1} \sum_{l>k}^{d-q} 2^{k+l-2} \sqrt{a_{d-k}} \sqrt{a_{d-l}} \right) \\
 \weq & \sum_{k=1}^{d-q} 2^{k-1} \left( \sqrt{a_d} - \sqrt{a_{d-k}} \right)^2 + \sum_{k=1}^{d-q-1} \sum_{l>k}^{d-q} 2^{k+l-2} \left( \sqrt{a_{d-k}} - \sqrt{a_{d-l}} \right)^2.
\end{align*}
Using the network's assortativity, we conclude that this last quantity is strictly greater than zero, and hence $J_q^{\bua} > J_q^{\tda}$ for all $q \le d-1$.
\end{proof}

\section{Additional Numerical Experiments}

\subsection{Synthetic Data Sets}

\subsubsection{Ternary Tree SBMs}
As the hierarchical community structure cannot always be represented by a binary tree, we also perform experiments on ternary-tree stochastic block models with depth 3 (the ternary tree is drawn in Figure~\ref{fig:ternary_tree}), 100 nodes in each bottom cluster, 
and the probability of an edge between two nodes whose lowest common ancestor has depth~$k$ is $p_k = a_k \log N / N$. 

We first show in Figure~\ref{fig:ternary_dendrograms} the dendrograms and trees obtained by {top-down} and {bottom-up} algorithms. Although both algorithms generate binary trees, ternary structures appear in Figures~\ref{fig:ternary_dendrograms_bottomup} and~\ref{fig:ternary_dendrograms_topdown}. Nevertheless, we observe in Figure~\ref{fig:ternary_dendrograms_topdown} that the dendrograms obtained by the {top-down} algorithm show some inversions.

Next, we proceed as in Section~\ref{subsection:experiments_BTSBMs}, by fixing $a_0$ to 40 and $a_3$ to 100, and by varying the values of $a_1$ and $a_2$ from $a_0$ to $a_3$ (with the condition $a_1 < a_2$). We observe in Figure~\ref{fig:performance_on_ternary_tree} differences in the performances of {top-down} and {bottom-up}. Indeed, in this setting, the bottom-up algorithm recovers communities exactly up to the theoretical thresholds. Moreover, the accuracy obtained by {top-down} is lower than the one obtained by {bottom-up}. 

\begin{figure}[!ht]
 \centering
  \begin{subfigure}[b]{0.32\textwidth}
    \includegraphics[width=\textwidth]{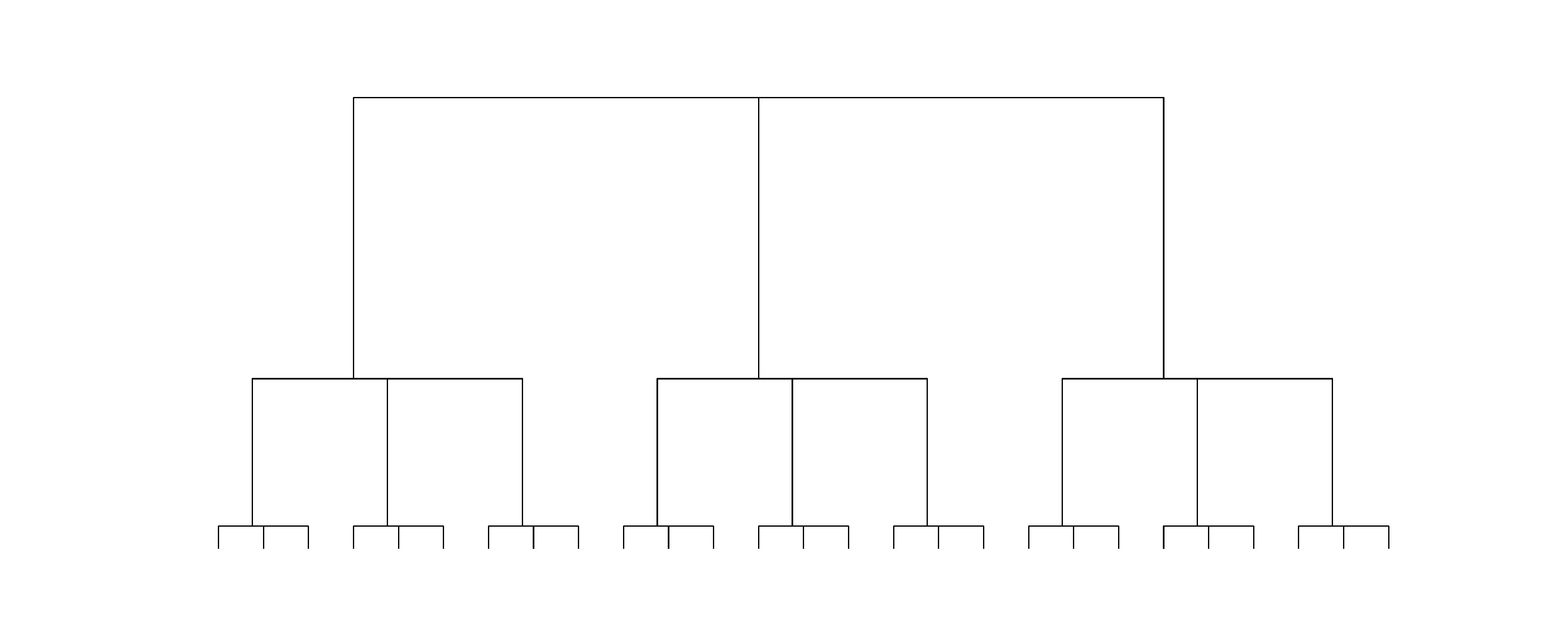}
    \caption{Ground truth tree}
    \label{fig:ternary_tree}
 \end{subfigure}
\hfill 
 \begin{subfigure}[b]{0.32\textwidth}
     \includegraphics[width=\textwidth]{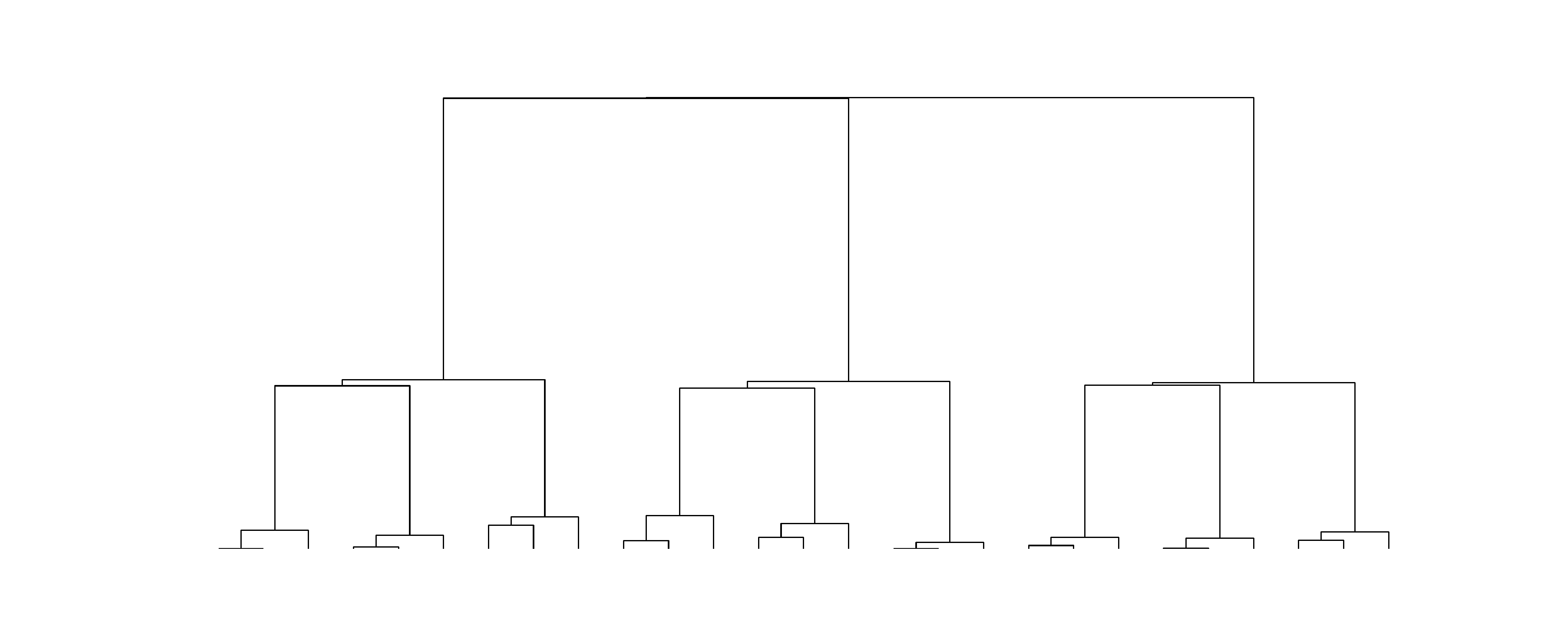}
     \caption{\textit{bottom-up}}
     \label{fig:ternary_dendrograms_bottomup}
 \end{subfigure}
 \hfill
  \begin{subfigure}[b]{0.32\textwidth}
     \includegraphics[width=\textwidth]{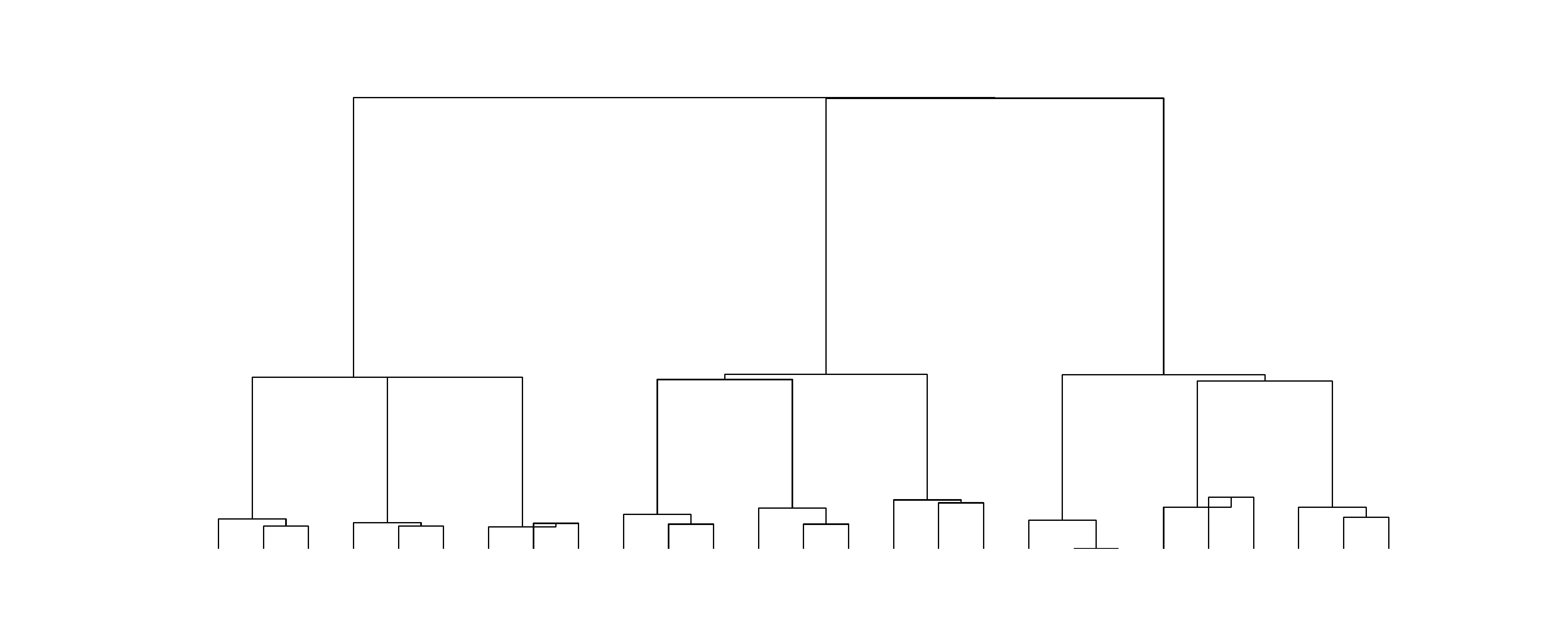}
     \caption{\textit{top-down}}
     \label{fig:ternary_dendrograms_topdown}
 \end{subfigure}
 \caption{(a) A ternary tree of depth 3 used as ground truth. (b)-(c) Dendrograms obtained by {bottom-up} and {top-down} algorithms on a Ternary Tree SBM of depth 3, $N = 2700$, and interaction probabilities $p_k = a_k \log N / N$ with $(a_0, a_1, a_2, a_3) = (10, 30, 40, 130)$. 
 }
 \label{fig:ternary_dendrograms}
\end{figure}

\begin{figure}[!ht]
 \centering
 \includegraphics[width=0.9\textwidth]{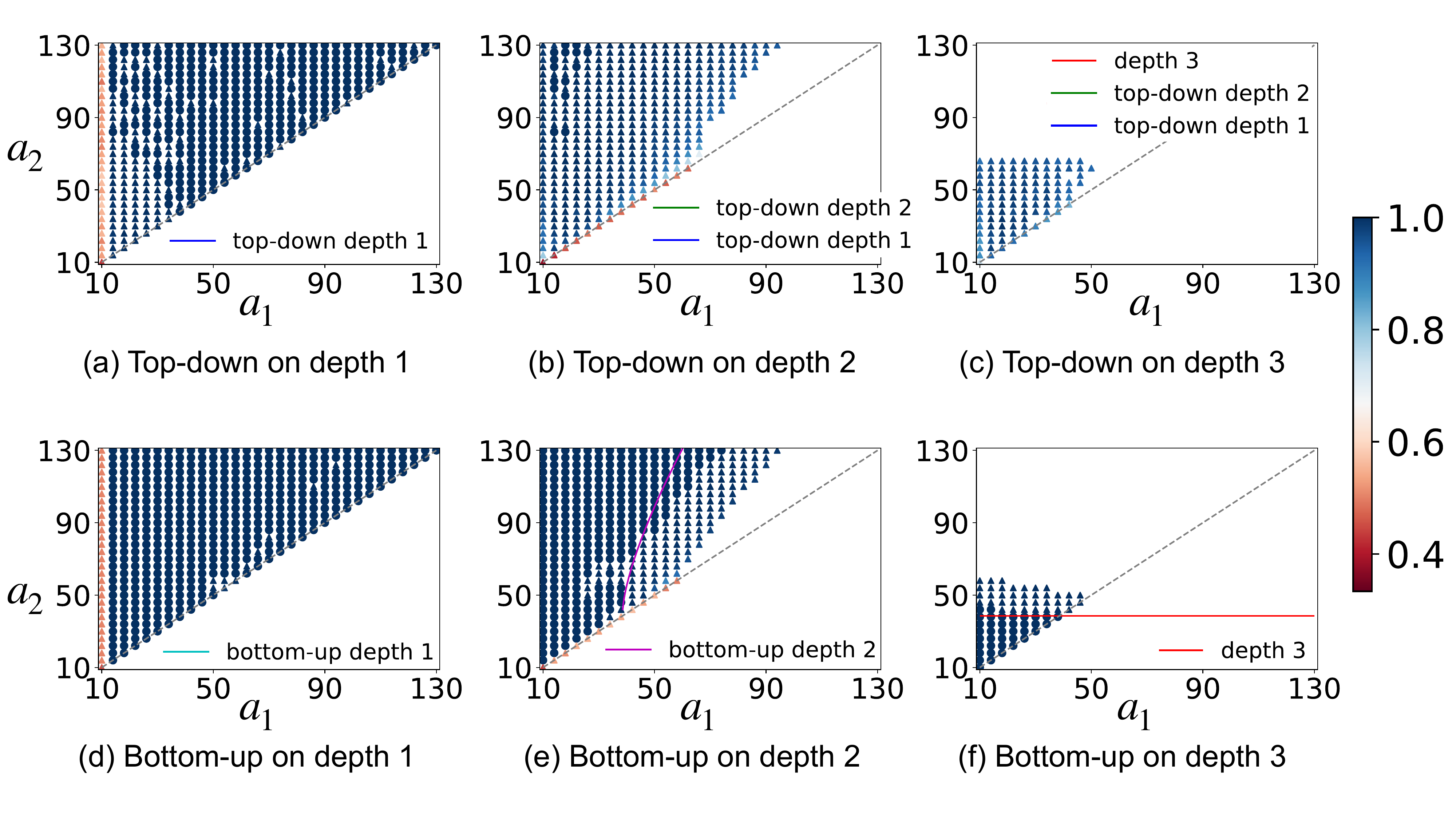}
 
 \caption{Performance of {bottom-up} and {top-down} algorithms on Ternary Tree SBMs of depth 3, $ N = 2700$ nodes, and interaction probabilities $p_k = a_k \log N / N$ with $a_0=10$ and $a_3 = 130$, as a function of~$a_1$ and~$a_2$. We vary $a_1 \le a_2$ from~$a_0$ to~$a_3$. 
 The performance of the algorithms is measured by the accuracy at each depth (averaged over 10 realizations), and the exact recovery threshold at different depths is shown in coloured solid lines. Exact recovery is shown with large circles and non-exact recovery with small crosses.
 }
 \label{fig:performance_on_ternary_tree}
\end{figure}

\subsubsection{Unbalanced HSBM}
\label{appendix:unbalanced_SBM}

We evaluate the performance of HCD algorithms on HSBM whose binary tree is not necessarily full and balanced. Similar to the BTSBM, we assume that the depth of the tree~$\cT$ determines the link probabilities, \textit{i.e.,} $p(u) = p_{|u|}$ for all $u \in \cT$. Because the tree is unbalanced, the bottom communities no longer have the same depth. In our experiments, the size of a bottom community having depth $k$ is $100 \cdot 2^{5-k}$, so that the total number of nodes is $N=3200$. 

To assess the accuracy of tree recovery, we define the similarity matrix $S(\cT,\cC)$ the $N$-by-$N$ matrix such that for $i \in C_a$ and $j \in C_b$ (with $a, b \in \cL_{\cT}$) we have 
\begin{align*}
 \Big( S\left( \cT, \cC \right) \Big)_{i j} \weq | \lca_{\cT}(a, b) |,
\end{align*}
where $ | \lca_{\cT}(a, b) |$ is the depth of the lowest common ancestor of $C_a$ and $C_b$ in the tree $\cT$. The tree recovery error is then defined as $$ \frac{ \| S ( \hcT, \hcC ) - S(\cT,\cC) \|^2_F} { \| S(\cT,\cC) \|^2_F }.$$  This metric quantifies the discrepancy between the estimated and true similarity matrices. 

Figures~\ref{fig:performance_algos_unbalanced_ex1} and~\ref{fig:performance_algos_unbalanced_ex2} compare different HCD algorithms on HSBMs whose corresponding unbalanced trees are shown in Figures~\ref{fig:unbalanced_tree_1} and~\ref{fig:unbalanced_tree_2}. We observe that bottom-up and top-down approaches perform similarly well, though one may outperform the other depending on $\beta$ and the chosen evaluation metric.

\begin{figure}[!ht]
 \centering
  \begin{subfigure}[b]{0.24\textwidth}
   \includegraphics[width=\textwidth]{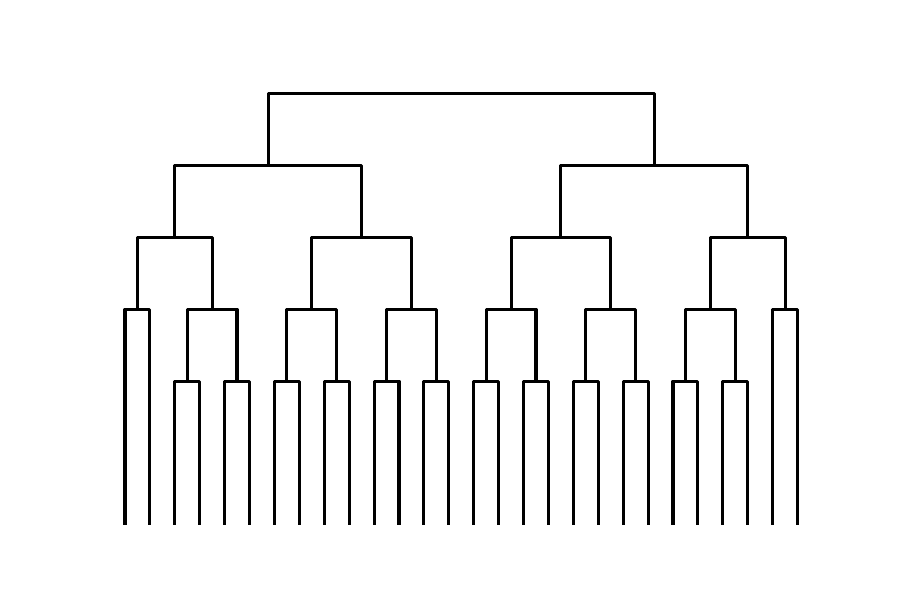}
   \caption{Unbalanced tree}
   \label{fig:unbalanced_tree_1}
  \end{subfigure}
  \hfill 
 \begin{subfigure}[b]{0.24\textwidth}
     \includegraphics[width=1.1\textwidth]{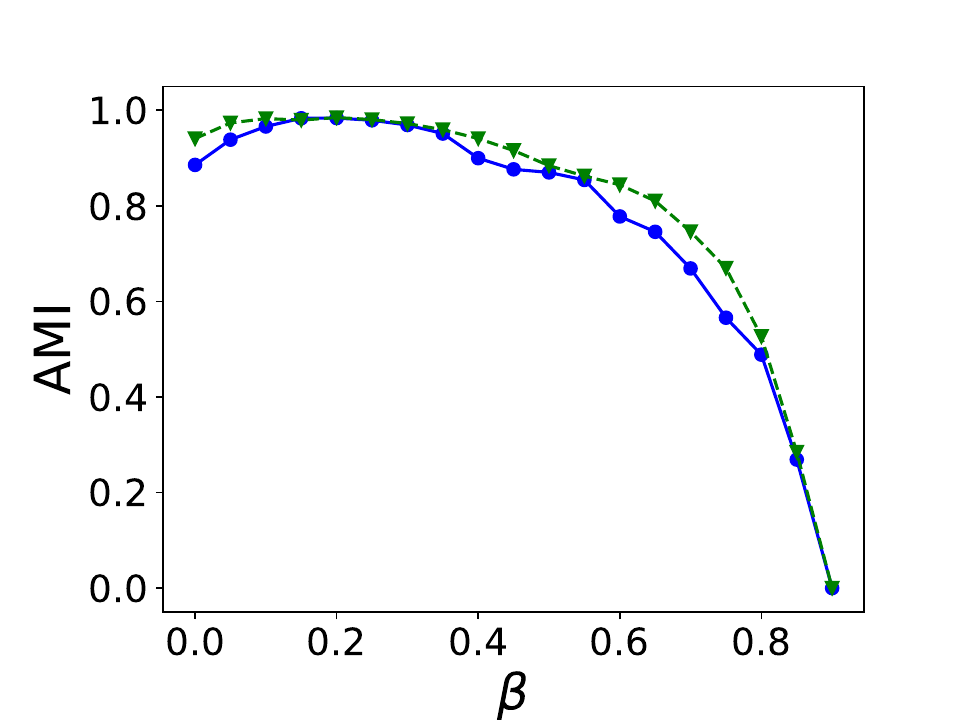}
     \caption{AMI}
 \end{subfigure}
 \hfill
  \begin{subfigure}[b]{0.24\textwidth}
     \includegraphics[width=\textwidth]{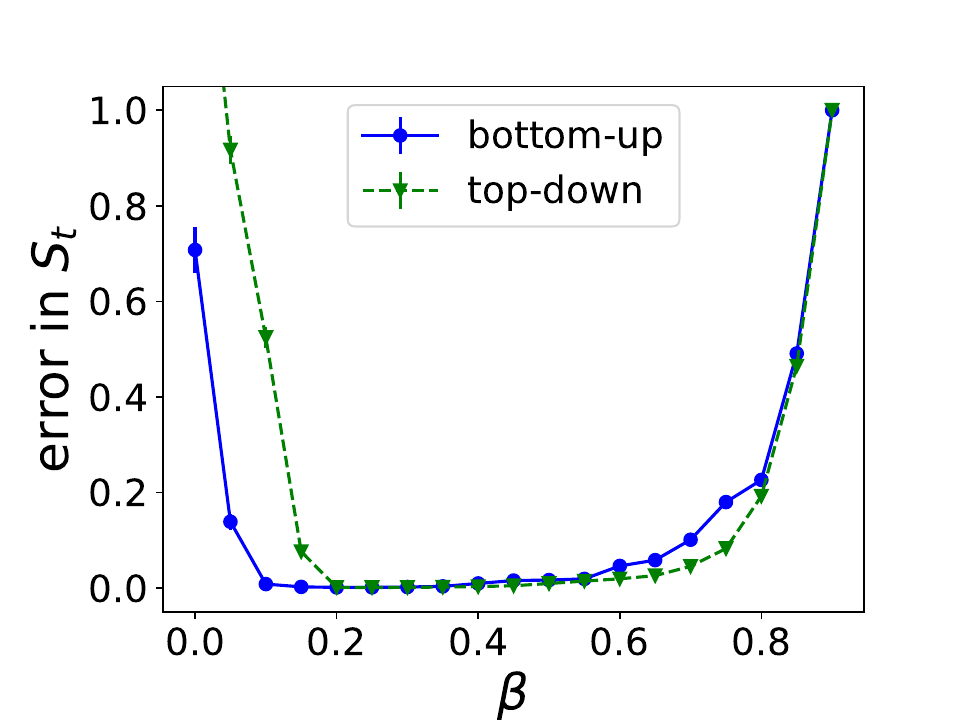}
     \caption{Tree similarity}
 \end{subfigure}
 \hfill
  \begin{subfigure}[b]{0.24\textwidth}
     \includegraphics[width=\textwidth]{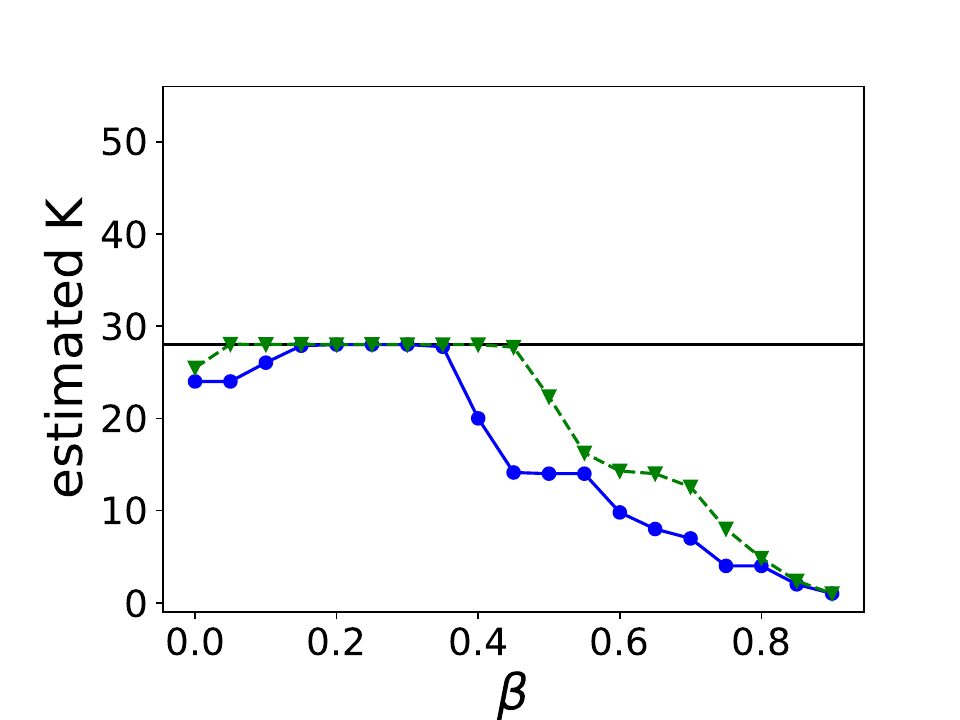}
     \caption{Estimated $K$}
 \end{subfigure}
 \caption{Performance of top-down and bottom-up algorithms on HSBMs, where the tree is given in Figure~\ref{fig:unbalanced_tree_1}, with $N = 3200$ nodes and $p_k = 64 \, \beta^{5-k} \, \frac{\log{N}}{N}$. The results are averaged over 100 realizations, and error bars show the standard error but are typically smaller than the symbols.}
  \label{fig:performance_algos_unbalanced_ex1}
\end{figure}

\begin{figure}[!ht]
    \centering
    \begin{subfigure}[b]{0.24\textwidth}
     \includegraphics[width=\textwidth]{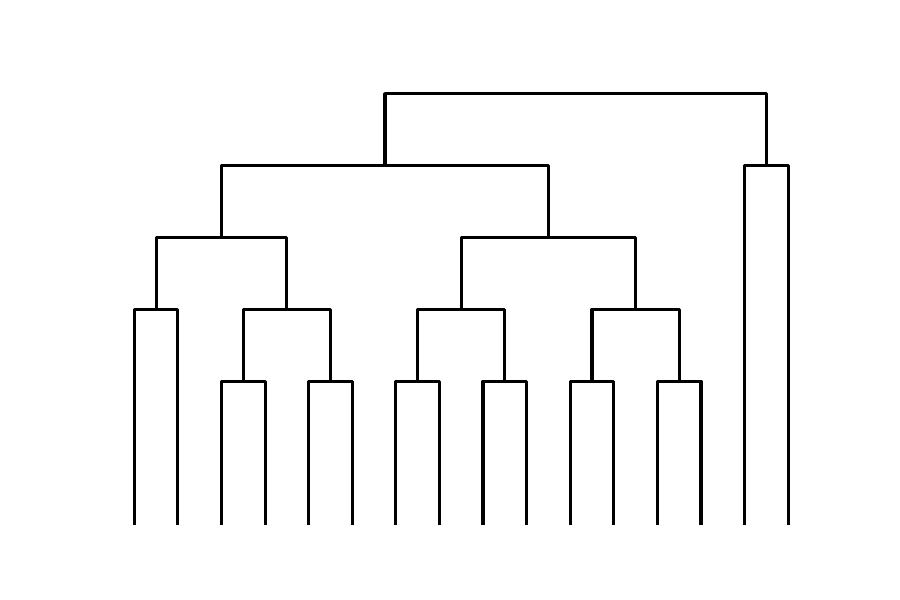}
     \caption{Unbalanced tree}
     \label{fig:unbalanced_tree_2}
     \end{subfigure}
     \hfill
    \begin{subfigure}[b]{0.24\textwidth}
        \includegraphics[width=\textwidth]{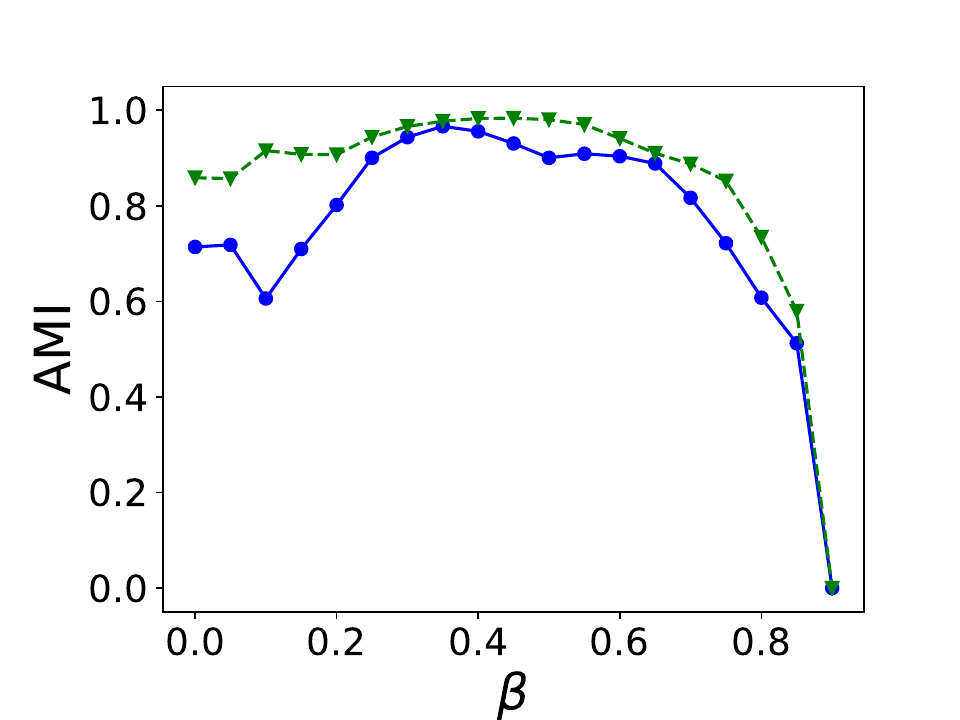}
        \caption{AMI}
    \end{subfigure}
    \begin{subfigure}[b]{0.24\textwidth}
        \includegraphics[width=\textwidth]{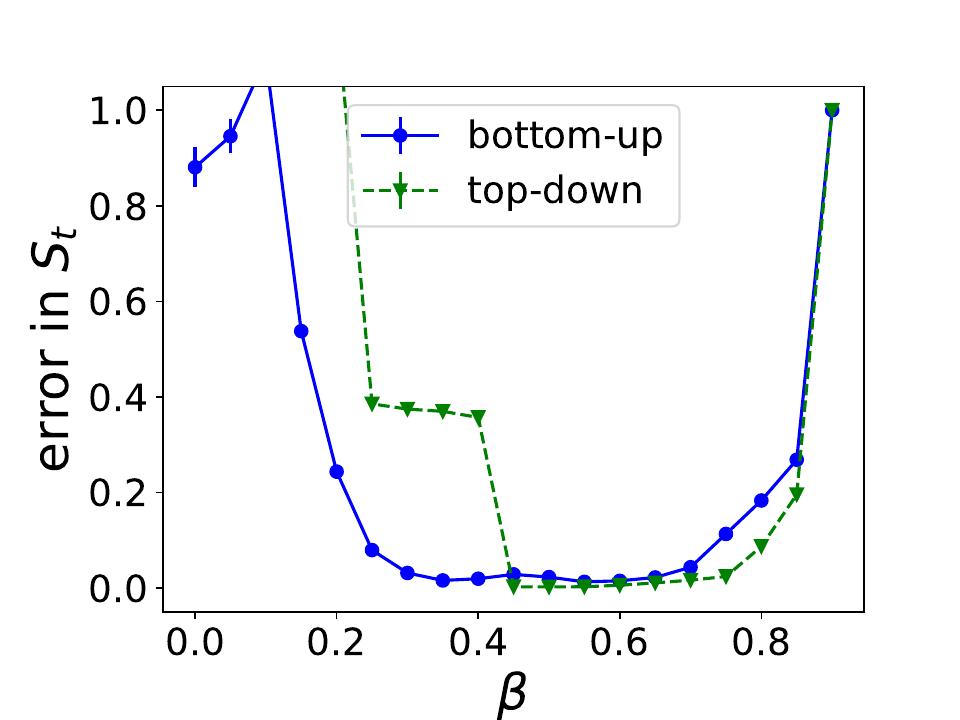}
        \caption{Tree similarity}
    \end{subfigure}
     \begin{subfigure}[b]{0.24\textwidth}
        \includegraphics[width=\textwidth]{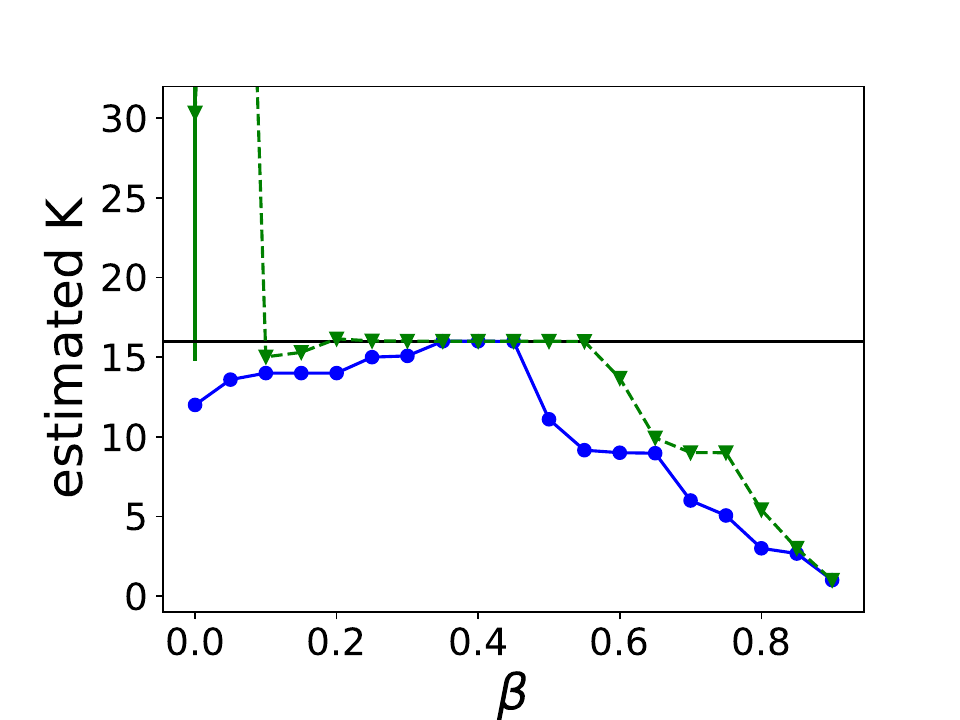}
        \caption{Estimated $K$}
    \end{subfigure}
    \caption{Performance of top-down and bottom-up algorithms on HSBMs, where the tree is given in Figure~\ref{fig:unbalanced_tree_2} and with $N=3200$ nodes and $p_k = 144 \, \beta^{5-k} \, \frac{\log{N}}{N}$. The results are averaged over 100 realizations, and error bars show the standard error but are typically smaller than the symbols.}
\label{fig:performance_algos_unbalanced_ex2}
\end{figure}

\subsubsection{Increased Number of Bottom Communities}
\label{appendix:deeper_SBM}
We evaluate the accuracy of HCD algorithms on a BTSBM of depth 6, resulting in $K = 2^6 = 64$ bottom communities. Figures~\ref{fig:performance_algos_K64} compare the performance of different HCD algorithms on BTSBM. We observe that bottom-up and top-down approaches perform similarly well, though one may outperform the other depending on $\beta$ and the chosen evaluation metric. 

\begin{figure}[!ht]
    \centering
    \begin{subfigure}[b]{0.3\textwidth}
        \includegraphics[width=\textwidth]{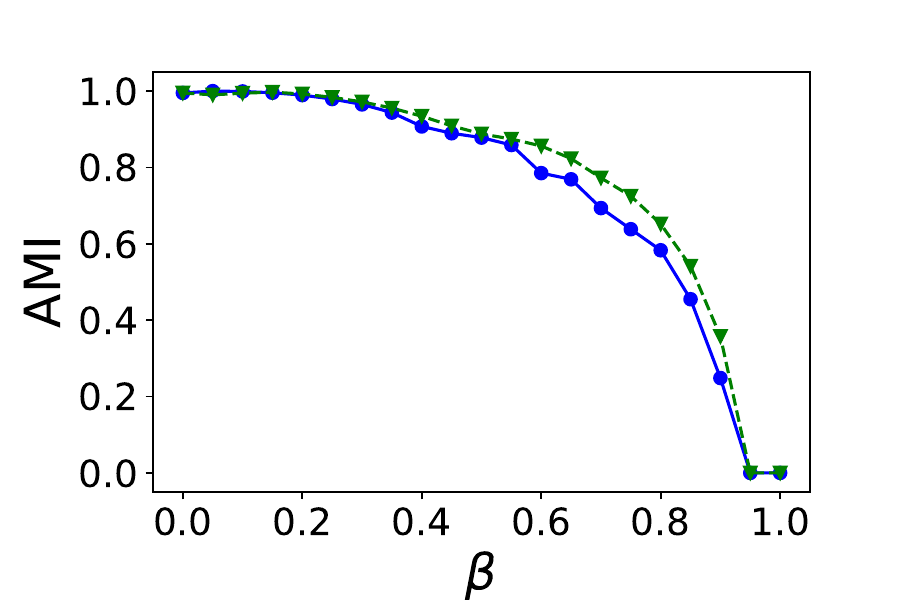}
        \caption{AMI}
    \end{subfigure}
    \hfill
    \begin{subfigure}[b]{0.3\textwidth}
        \includegraphics[width=\textwidth]{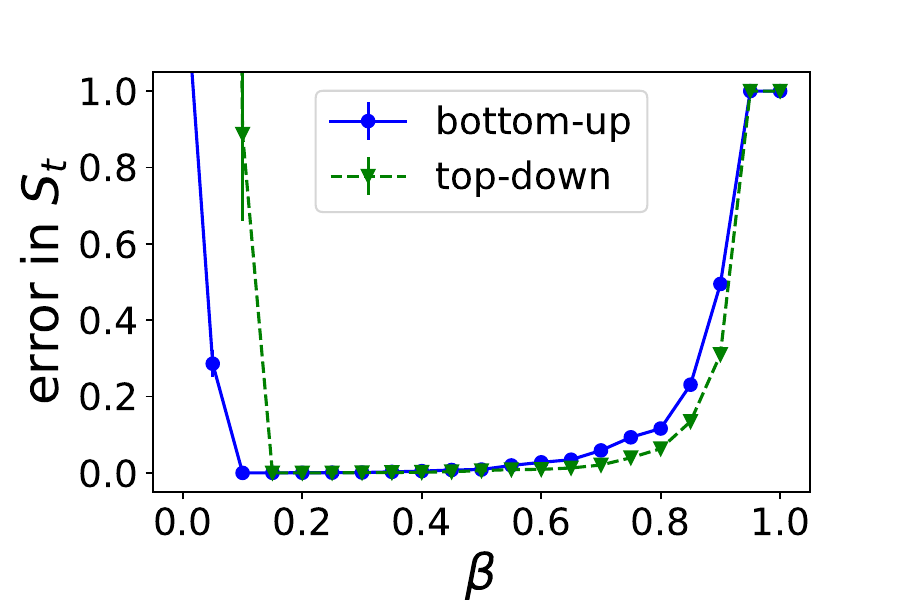}
        \caption{Tree similarity}
    \end{subfigure}
    \hfill
     \begin{subfigure}[b]{0.3\textwidth}
        \includegraphics[width=\textwidth]{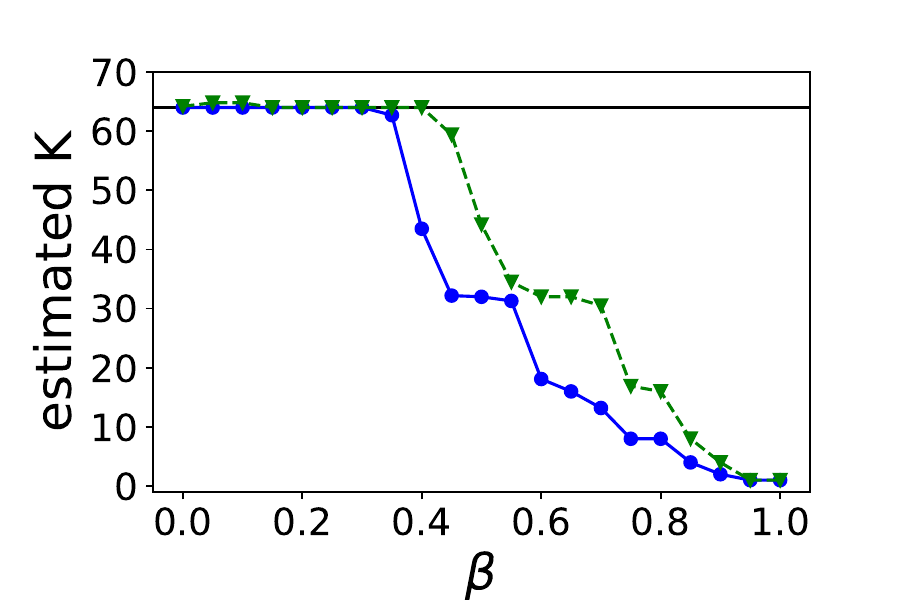}
        \caption{Estimated $K$}
    \end{subfigure}
    \caption{Performance of top-down and bottom-up algorithms on BTSBMs of depth 6, with $N=6400$ nodes ($K=64$ and 100 nodes per community) and $p_k = 81 \, \beta^{6-k} \, \frac{\log{N}}{N}$. The results are averaged over 10 realizations, and error bars show the standard error but are typically smaller than the symbols.}
\label{fig:performance_algos_K64}
\end{figure}

\subsubsection{Another Bottom-up Approach}
In this subsection, we consider another type of hierarchical community detection algorithm that first identifies bottom communities then aggregates them. The method called synthesis~\cite{fang2023t}, differs from linkage-based approaches by employing the (sparse) neighbor-joining algorithm for aggregation. As a result, synthesis can handle more general tree structures than the binary rooted trees produced by bottom-up and top-down methods, and it outputs unrooted, potentially non-binary trees.

Figure~\ref{fig:vs_synthesis} reports the performance of our focal bottom-up algorithm with synthesis. We evaluate performance using three metrics: bottom community accuracy, Robinson-Foulds distance from the ground-truth tree, and the recovery rate of the true hierarchical structures. The Robinson-Foulds distance is a metric to measure the distance between two potentially unrooted trees and is normalized to range from 0 to 1. The tree recovery rate is defined as the proportion of outputs with a Robinson-Foulds distance of zero. Although Fang and Rohe~\cite{fang2023t} use sparse neighbor-joining to accommodate non-binary trees, we apply the standard neighbor-joining algorithm for a fair comparison since our ground-truth tree structures are binary.
\begin{figure}[!ht]
    \centering
    \begin{subfigure}[b]{0.3\textwidth}
        \includegraphics[width=\textwidth]{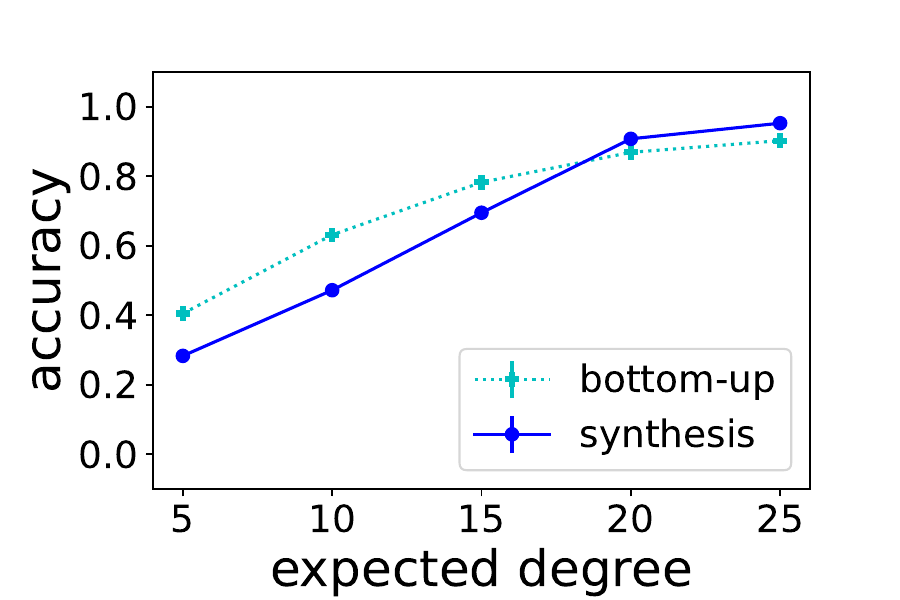}
        \caption{Bottom accuracy}
    \end{subfigure}
    \begin{subfigure}[b]{0.3\textwidth}
        \includegraphics[width=\textwidth]{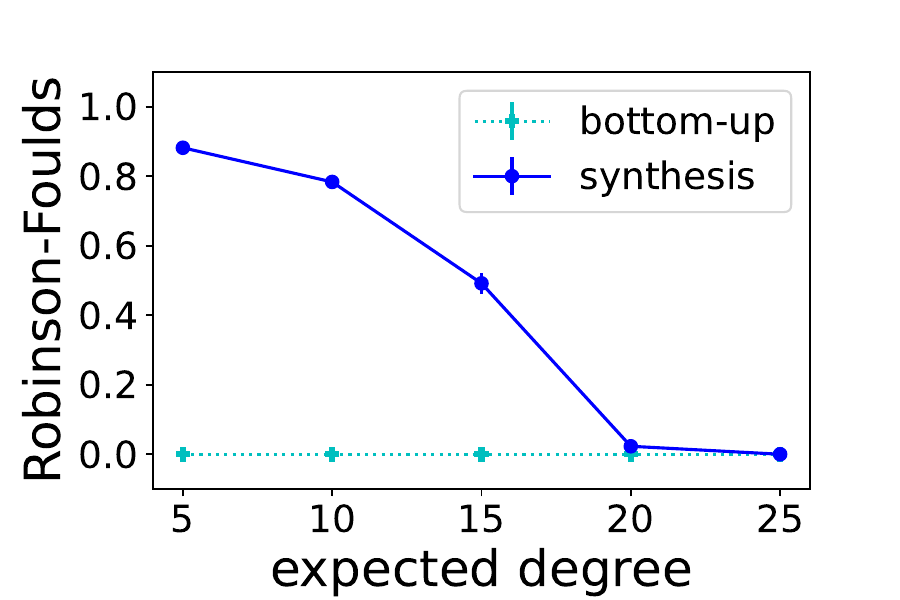}
        \caption{Robinson-Foulds distance}
    \end{subfigure}
     \begin{subfigure}[b]{0.3\textwidth}
        \includegraphics[width=\textwidth]{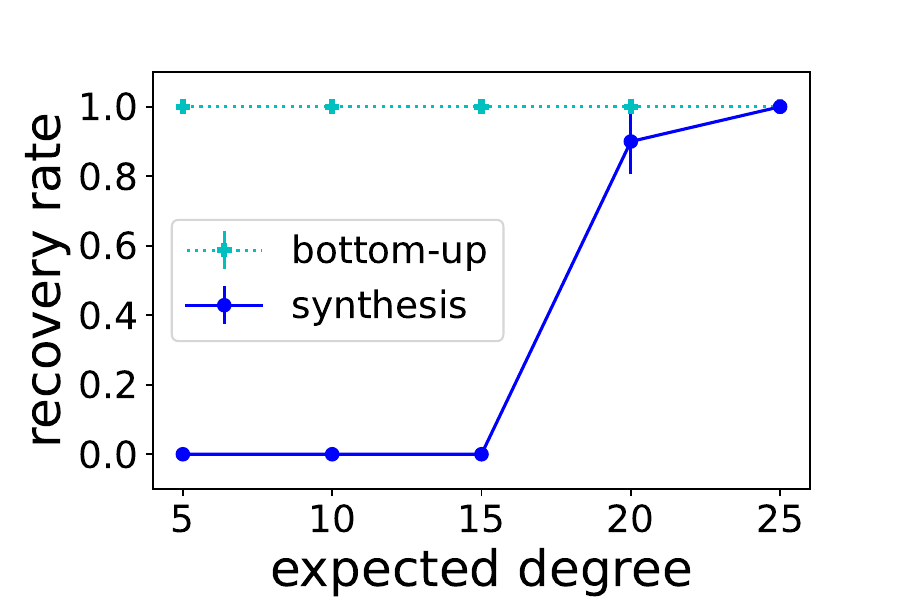}
        \caption{Recovery rate}
    \end{subfigure}
    \caption{Performance of synthesis and bottom-up algorithms on BTSBMs of depth 5, with $N=6400$ nodes ($K=32$ and 200 nodes per community) and $p_{k-1}/p_{k}=0.36$. The results are averaged over 10 realizations, and error bars show the standard error but are typically smaller than the symbols.}
\label{fig:vs_synthesis}
\end{figure}

\subsection{Real Data Sets}
\subsubsection{Military Inter-alliance}

We next consider the network of military alliances between countries. The data is provided by the Alliance Treaty Obligations and Provisions (ATOP) project~\cite{leeds2002alliance}. We select the year 2018 (as this is the most recent year available). We define two countries as allied if they share a defensive alliance (we do not consider non-aggression pacts, as those are more numerous and historically not necessarily well respected). This leads to a network of 133 countries and 1391 alliances. Some important countries such as India or Switzerland are missing as they do not share any defensive alliances with anybody. Moreover, the graph is not connected, as a small component made of three countries (China, North Korea and Cuba) is disconnected from the rest of the world.

Figures~\ref{fig:bottom_up_military_alliances} and~\ref{fig:top_dwon_military_alliances} show the output of \textit{bottom-up} and \textit{top-down} algorithms. \textit{Bottom-up} predicts 7 bottom communities, which represent geopolitical alliances based on political affiliation and geography (European countries, Eurasian countries, Arabic countries, Western African countries and Central/Southern African countries). The top level splits the graph's largest connected component into 3 clusters: Western countries, Eurasian countries, and African and Middle-East countries. 
While some of these clusters are also recovered by \textit{top-down} HCD, the separation of African countries by \textit{top-down} algorithm appears worse.

\begin{figure}[!ht]
\centering
\begin{subfigure}{0.49\textwidth}
    \includegraphics[width=\textwidth]{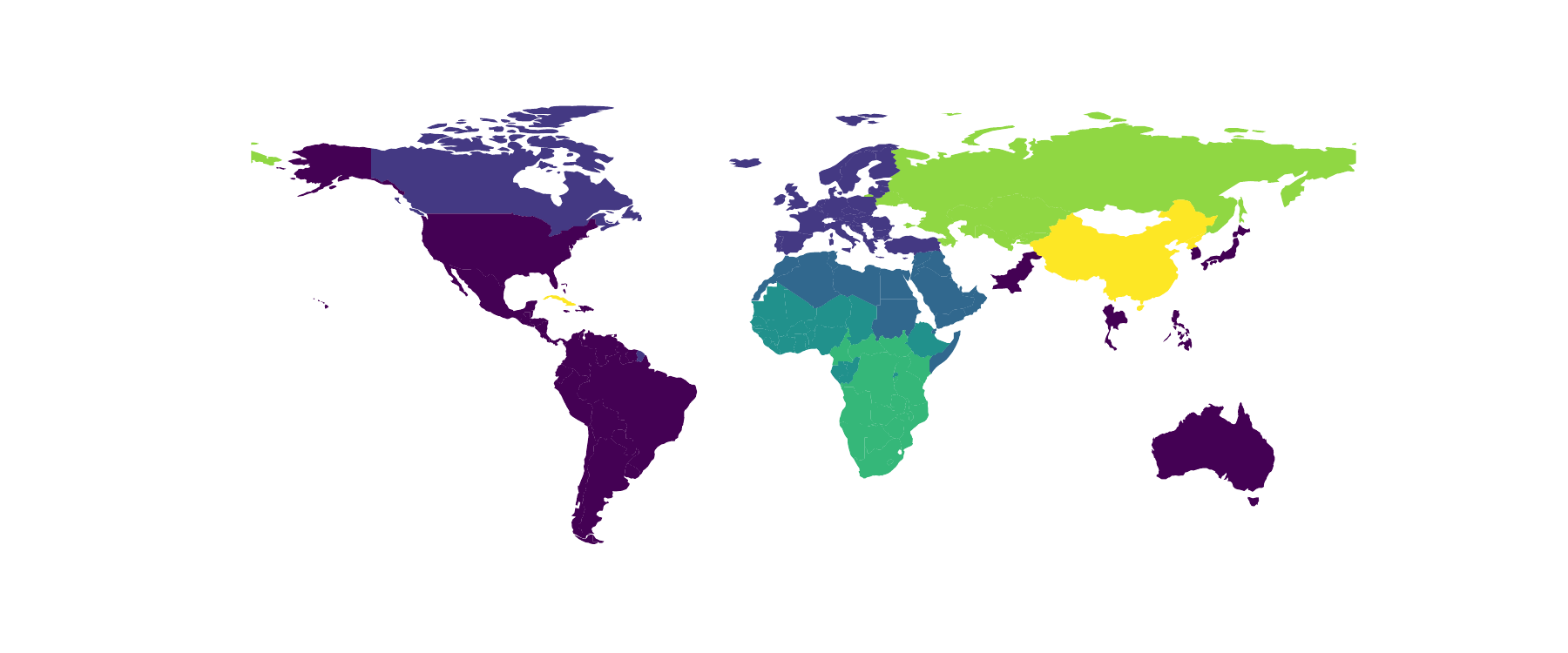}
    \caption{Highest depth.}
\end{subfigure}    
\hfill
\begin{subfigure}{0.49\textwidth}
    \includegraphics[width=\textwidth]{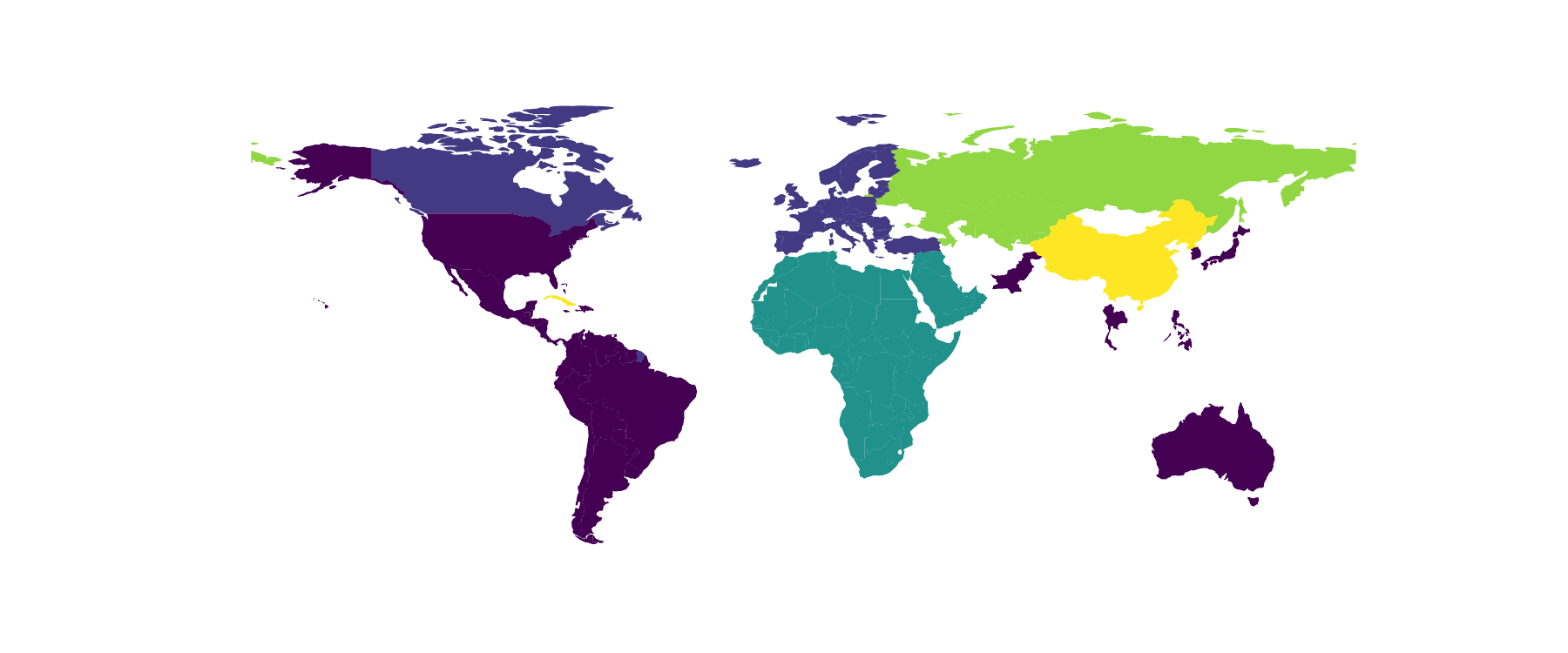}
    \caption{Middle depth (after 2 merges).}
\end{subfigure}
\hfill
\begin{subfigure}{0.49\textwidth}
    \includegraphics[width=\textwidth]{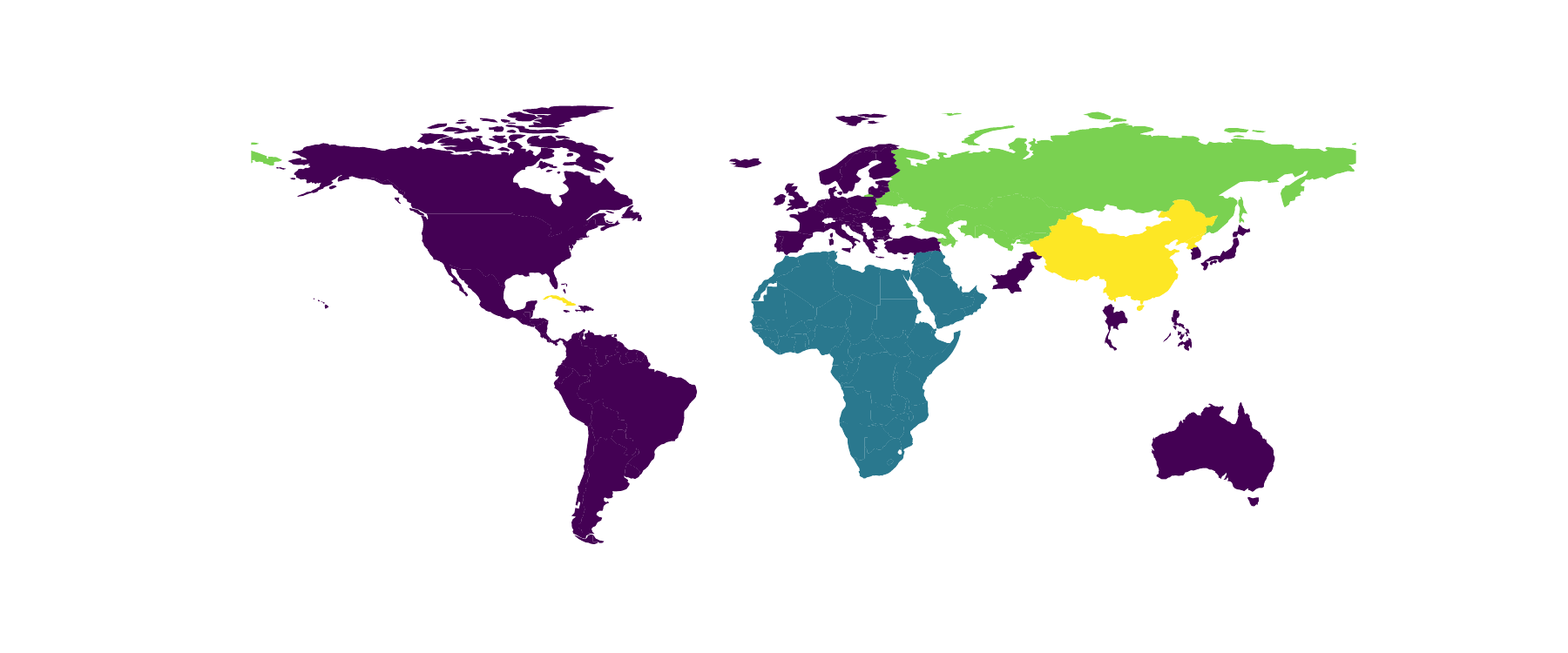}
    \caption{Lowest depth (after 3 merges).}
\end{subfigure}
\hfill
\begin{subfigure}{0.49\textwidth}
    \includegraphics[width=\textwidth]{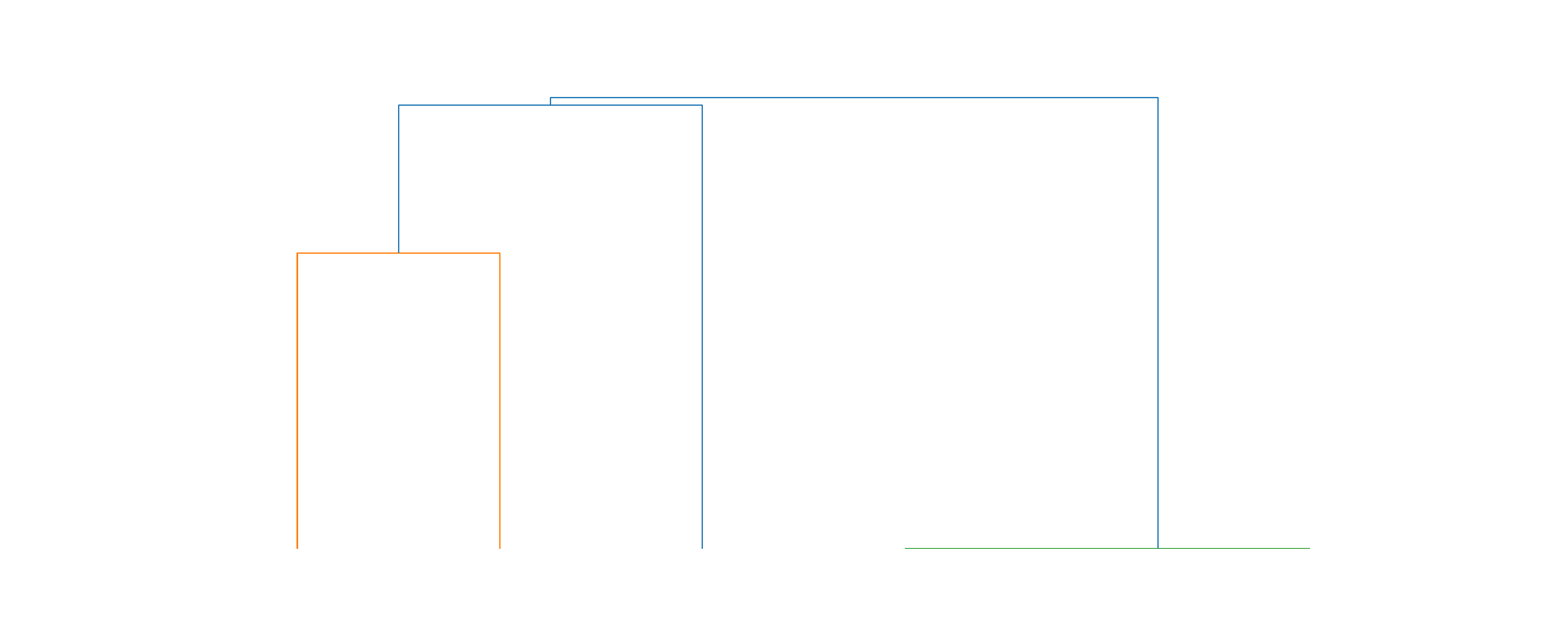}
    \caption{Dendrogram.}
\end{subfigure}    
\caption{Output of \textit{bottom-up} algorithm on the military alliance network. The dendrogram does not show the disconnected component (China, Cuba, North Korea).}
\label{fig:bottom_up_military_alliances}
\end{figure}

\begin{figure}[!ht]
\centering
\begin{subfigure}{0.49\textwidth}
    \includegraphics[width=\textwidth]{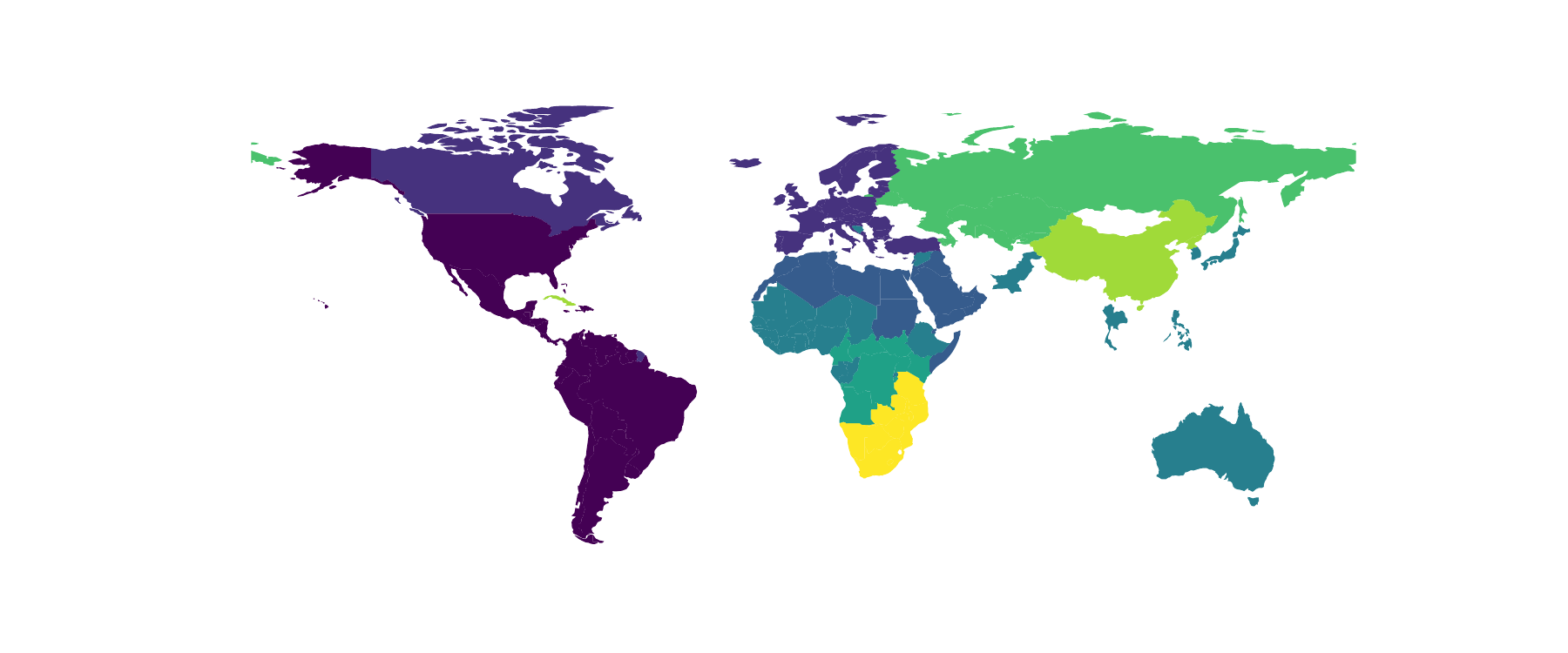}
    \caption{Highest depth (8 clusters).}
\end{subfigure}    
\hfill
\begin{subfigure}{0.49\textwidth}
    \includegraphics[width=\textwidth]{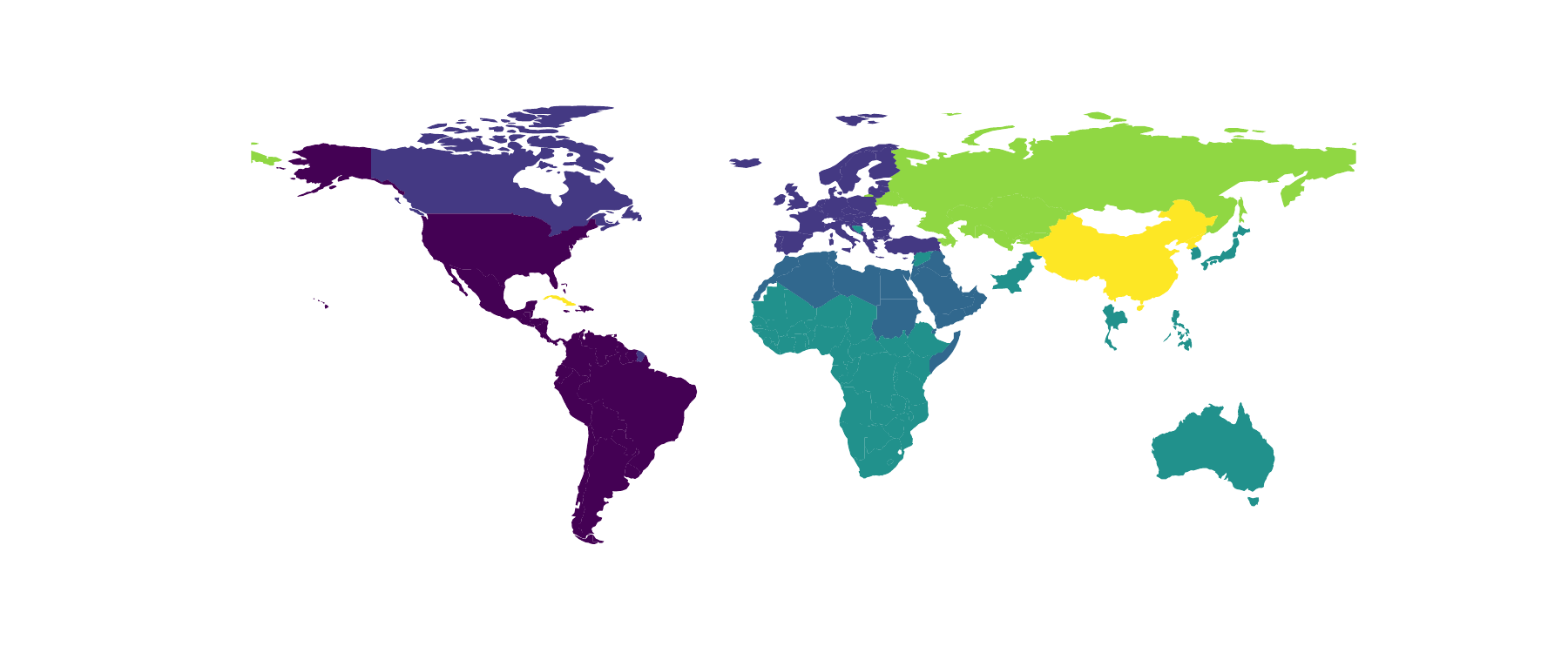}
    \caption{Middle depth (6 clusters).}
\end{subfigure}
\hfill
\begin{subfigure}{0.49\textwidth}
    \includegraphics[width=\textwidth]{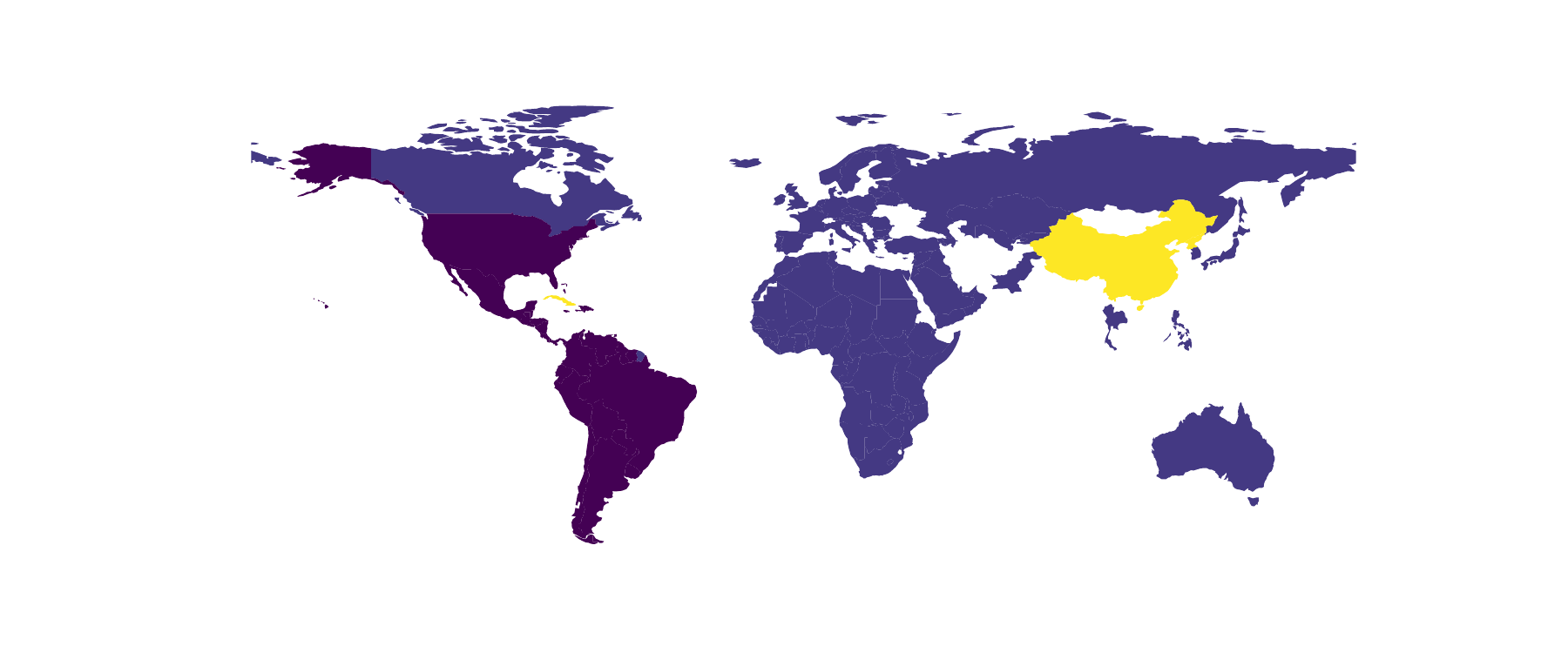}
    \caption{Lowest depth (3 clusters).}
\end{subfigure}
\hfill
\begin{subfigure}{0.49\textwidth}
    \includegraphics[width=\textwidth]{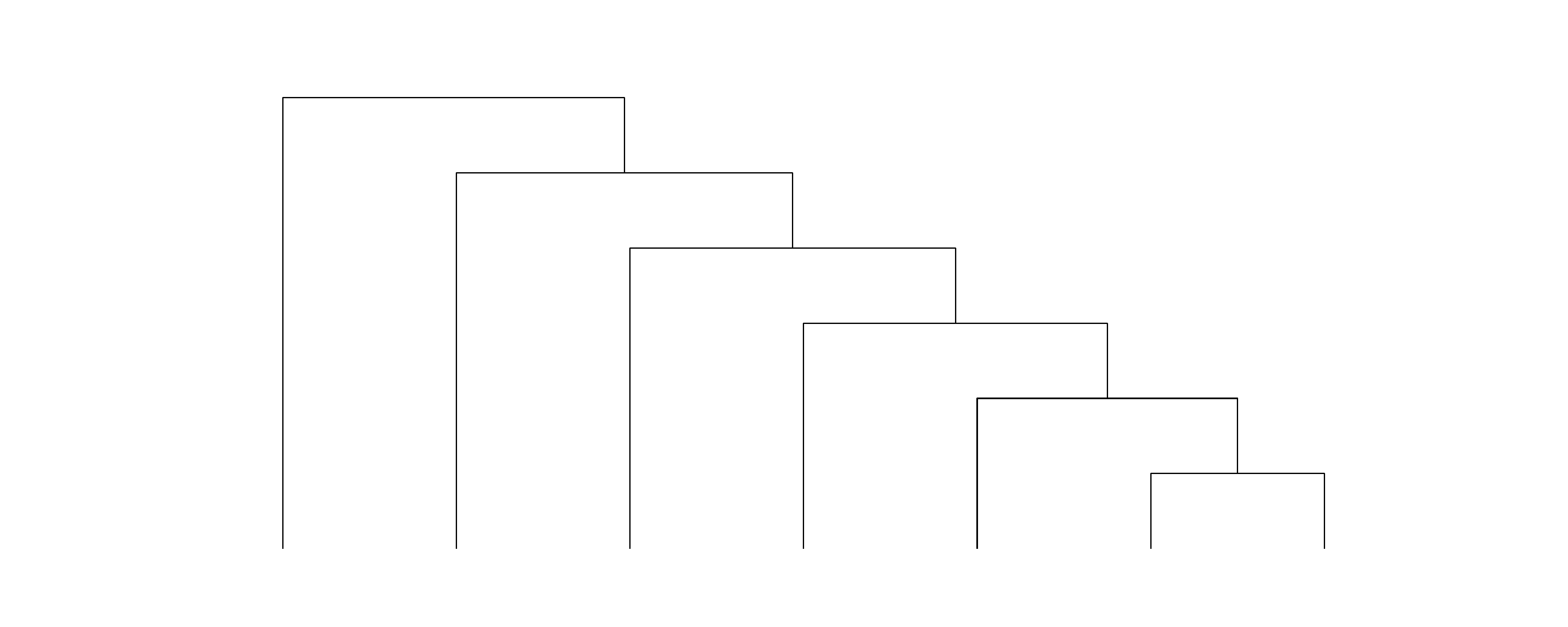}
    \caption{Dendrogram.}
\end{subfigure}
\caption{Output of \textit{top-down} algorithm on the military alliance network. 
The dendrogram does not show the disconnected component (China, Cuba, North Korea).
}
\label{fig:top_dwon_military_alliances}
\end{figure}

\subsubsection{Football Data Set}
We also test HCD algorithms on the United States college (American) football dataset~\cite{girvan2002community}. This network represents the schedule of Division I games for the 2000 season. Each node in the network corresponds to a college football team, and edges represent the regular-season games between the teams. The teams are categorized into 11 conferences, in which games are more frequent between the members. Each conference has 8 to 12 teams. We exclude the "independent" teams which do not belong to any conferences.
Since the original community labels appear to be based on the 2001 season, while the edges represent the games played during the 2000 season, we proceed to the same correction as in~\cite{evans2010clique}. 

The results obtained by the different HCD algorithms are given in Figure \ref{fig:results_football}. First, we observe that all the algorithms perform well (AMI for \textit{bottom-up}, \textit{top-down}, \textit{Paris}, and \textit{Bayesian} are respectively 0.962, $0.892$, 0.965, and  0.976). However, \textit{top-down} has more errors than the other algorithms. 
Interestingly, \textit{bottom-up}, \textit{top-down}, and \textit{Bayesian} algorithms predict 10 clusters (more precisely, \textit{Bayesian} detects 10.1 communities averaged over 100 runs), as they tend to infer \textit{Big West} and \textit{Mountain West} conferences as a single cluster. 
Finally, we can restore some geographical proximity among conferences from the hierarchy inferred by \textit{bottom-up}. For example, Conference USA is composed of teams located in the Southern United States, while the Southeastern Conference's member institutions are located primarily in the South Central and Southeastern United States. Another example is that teams belonging to Pacific Ten, Big West, and Mountain West are all located in the West, and these conferences are also close in the \textit{bottom-up} dendrogram.

\begin{figure}[!ht]
 \centering
 \begin{subfigure}[b]{0.49\textwidth}
     \includegraphics[width=\textwidth]{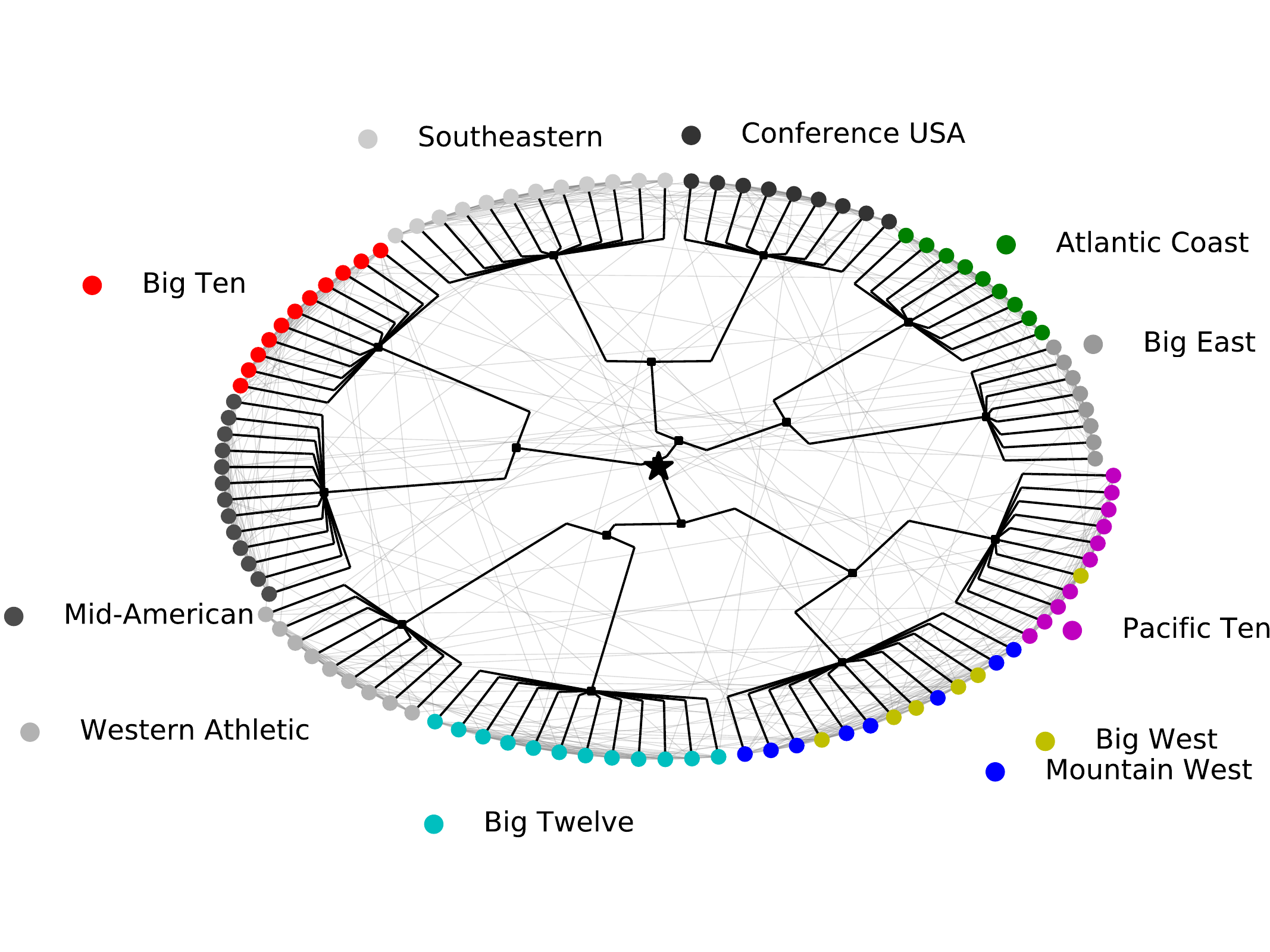}
     \caption{\textit{bottom-up}}
 \end{subfigure}
 \hfill
   \begin{subfigure}[b]{0.49\textwidth}
     \includegraphics[width=\textwidth]{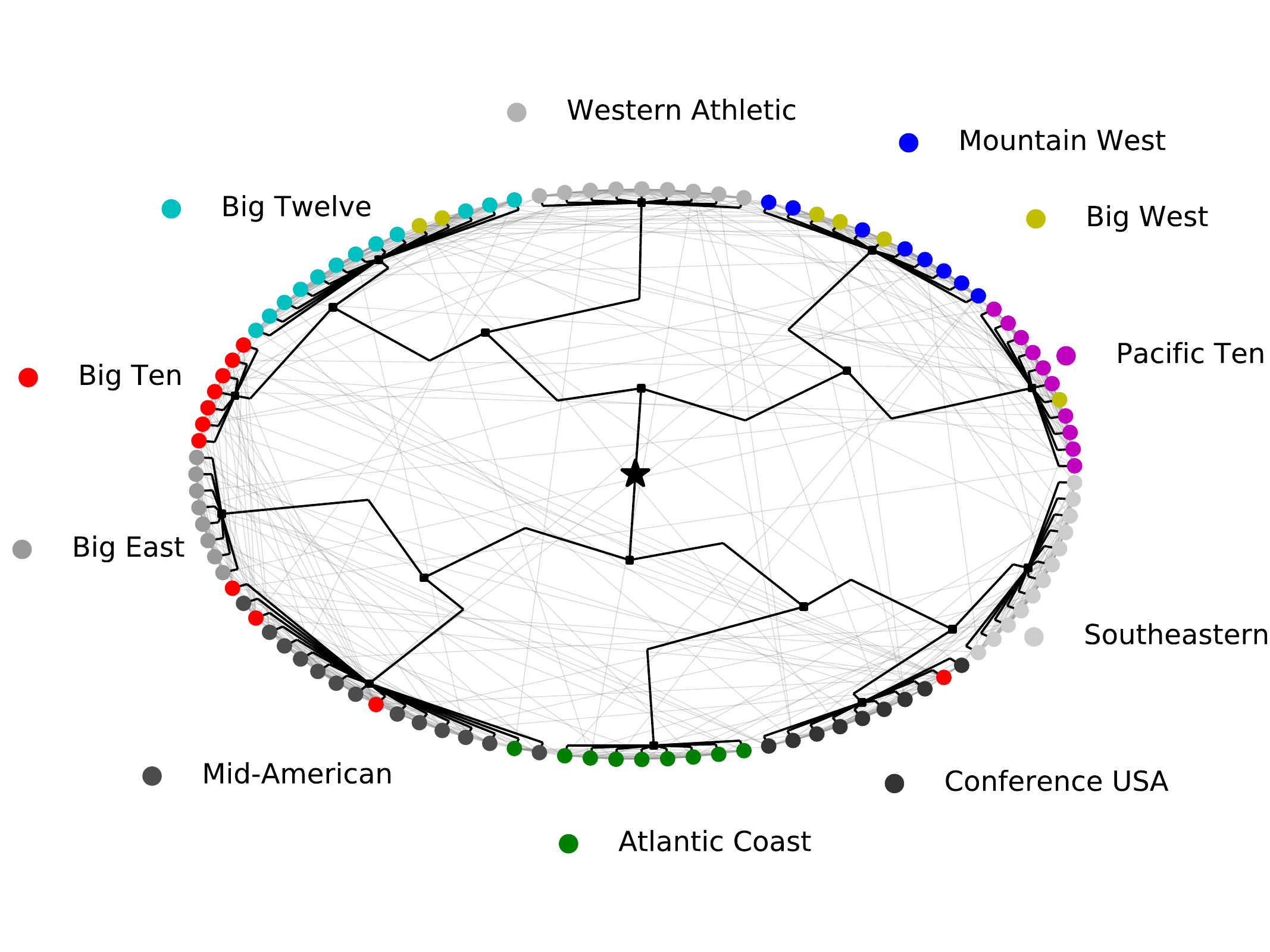}
     \caption{\textit{top-down}}
 \end{subfigure}
 \caption{ Output of \textit{bottom-up} and \textit{top-down}  algorithms on the \textit{football} data set. The colours correspond to conferences, and grey edges indicate having regular-season games between the two teams. The hierarchical tree is drawn in black, and its root is marked by a star symbol. }
 \label{fig:results_football}
\end{figure}

\end{document}